\documentclass[12pt]{amsart}
\usepackage[margin=1in]{geometry}
\usepackage[english]{babel}
\usepackage{amsmath}
\usepackage{amsfonts}
\usepackage{amssymb}
\usepackage{url}
\usepackage{amsthm}
\usepackage{enumerate}
\usepackage{array}
\usepackage[utf8]{inputenc}
\usepackage{hhline}
\usepackage{fullpage,caption}
\usepackage{tikz}
\usetikzlibrary{arrows,shapes,calc,snakes,shadows,positioning,cd}
\usepackage{tabularx}
\usepackage{graphicx}
\usepackage{esint}
\usepackage{pdflscape}
\usepackage{comment}
\usepackage{hyperref}
\usepackage{mathrsfs}

\newtheorem{theorem}{Theorem}[section]

\newtheorem{proposition}[theorem]{Proposition}

\newtheorem{definition}[theorem]{Definition}
\newtheorem{remark}[theorem]{Remark}

\newcommand{\Ci}{\mathbb{C}} 
\newcommand{\D}{\mathbb{D}}
\newcommand{\I}{\mathbb{I}}
\newcommand{\N}{\mathbb{N}}

\newcommand{\R}{\mathbb{R}}
\newcommand{\T}{\mathbb{T}}
\newcommand{\Z}{\mathbb{Z}}

\newcommand{\A}{\mathcal{A}}
\newcommand{\B}{\mathcal{B}}

\newcommand{\K}{\mathcal{K}}

\newcommand{\Oh}{\mathcal{O}}

\newcommand{\aloop}{\mathfrak{a}}
\newcommand{\bloop}{\mathfrak{b}}
\newcommand{\LD}{\mathfrak{L}}
\newcommand{\s}{\mathfrak{s}}

\newcommand{\dist}{\text{dist}}

\newcommand{\dotbold}[1]{\dot{\boldsymbol{#1}}}

\newcommand{\dd}{\text{d}}

\newcommand{\out}{\text{(out)}}
\newcommand{\HM}{\text{HM}}
\newcommand{\myarcsin}{\mathcal{U}}

\newcommand{\bbm}{\begin{bmatrix}}
\newcommand{\ebm}{\end{bmatrix}}

\DeclareMathOperator{\Ai}{Ai}
\DeclareMathOperator{\Real}{Re}
\DeclareMathOperator{\Imag}{Im}

\newcounter{RHP}
\newenvironment{RHP}[1][]{\refstepcounter{RHP}

    \textbf{Riemann-Hilbert Problem~\theRHP: #1}}

\newcommand{\LT}[1]{\begin{bmatrix}
    1 & 0 \\
    #1 & 1
\end{bmatrix}}

\newcommand{\UT}[1]{\begin{bmatrix}
    1 & #1 \\
    0 & 1
\end{bmatrix}} 
\title{Large-Parameter Asymptotics of \\ Generalized Hasting-McLeod Functions}
\author{Kurt Schmidt}
\address[K. Schmidt]{Department of Mathematical Sciences\\ University of Cincinnati\\ PO Box 210025\\ Cincinnati, OH 45221.}
\email{schmidku@mail.uc.edu}
\author{Robert Buckingham}
\address[R. Buckingham]{Department of Mathematical Sciences\\ University of Cincinnati\\ PO Box 210025\\ Cincinnati, OH 45221.}
\email{buckinrt@uc.edu}

\begin{document}
\phantom{.}
\vspace{-1cm}
\maketitle
\vspace{-1cm}
\begin{abstract}
    The generalized Hastings-McLeod solutions to the inhomogeneous Painlev\'{e}-II equation arise in multi-critical unitary random matrix ensembles, the chiral two-matrix model for rectangular matrices, non-intersecting squared Bessel paths, and non-intersecting Brownian motions on the circle. We establish the leading-order asymptotic behavior of the generalized Hastings-McLeod functions as the inhomogeneous parameter approaches infinity using the Deift-Zhou nonlinear steepest-descent method for Riemann-Hilbert problems. This analysis is done in both the pole-free region and pole region. The asymptotic formulae show excellent agreement with numerically computed solutions in both regions.
\end{abstract}

\tableofcontents

\newpage 

\section{Introduction}
\subsection{Background}
We consider the inhomogeneous Painlev\'e-II equation
\begin{equation} \label{PII} 
    u''(y) = 2u(y)^3 + yu(y)- \alpha,
\end{equation}
where $\alpha$ is the inhomogeneity parameter. The well-known Hastings-McLeod function, $u_\HM$, solves \eqref{PII} when $\alpha=0$ with boundary conditions $u_{\HM}(y) \sim \sqrt{-y/2}$ as $y \to - \infty$ and $u_{\HM}(y) \sim \Ai(y)$ as $y \to +\infty$ (where $\Ai$ is the Airy function) \cite{hastings1980boundary}. It arises in the Tracy-Widom distribution describing a wide range of random  phenomena, including, but not limited to, the largest eigenvalue of a  Gaussian Unitary Ensemble random matrix \cite{largestEigienVale} and the longest increasing subsequence of a permutation \cite{Baik1998OnTD}.

In 2008, Claeys, Kuijlaars, and Vanlessen found that if one perturbs the Gaussian Unitary Ensemble so that the limiting mean density function of eigenvalues vanishes quadratically at the origin, then perturbations of the Hastings-McLeod function arise in a double-scaling limit \cite{Claeys2005MulticriticalUR}. These perturbations solve \eqref{PII} for $\alpha>-1/2$ with boundary conditions 
\begin{equation}\label{BonCon}
    \begin{cases} \displaystyle
        u_{\HM}^{(\alpha)}(y)=\sqrt{-\frac{y}{2}}\left( 1+\Oh\left( 1/(-y)^{3/2} \right) \right) & \text{ as } y \to -\infty \\[10 pt] \displaystyle
        u_{\HM}^{(\alpha)}(y) = \frac{\alpha}{y}\left( 1+\Oh\left( 1/y^3 \right) \right) & \text{ as } y \to +\infty \hspace{.5 cm} (\alpha \not = 0). 
    \end{cases}
\end{equation}
This family of solutions to \eqref{BonCon} is named the generalized Hastings-McLeod functions.  Recently, these functions have been observed in interacting particle systems \cite{Buckingham2017TheKP, NonIntersectingSB} and in a different random matrix model \cite{chiral}. 
We establish asymptotic formulae for generalized Hastings-McLeod functions, $u_\HM^{(\alpha)}$, for large, positive half-integers $\alpha$. 

The starting point of our analysis is a Riemann-Hilbert representation for generalized Hastings-McLeod functions  discovered by Buckingham and Liechty in 2018 while investigating Dyson Brownian bridges \cite{Buckingham2017TheKP}. We apply the Deift-Zhou method of nonlinear steepest descent developed in \cite{Deift1993ASD}. This involves nontrivial transformations that force the jump discontinuities to be exponentially small as $\alpha \to \infty$ except at certain points or arcs. Ignoring these small jumps results in a  Riemann-Hilbert problem that is explicitly solvable and is called a \textit{model problem}. This provides a rigorous approximation, with controlled error, to the solution of the original Riemann-Hilbert problem for the generalized Hastings-McLeod functions. 


 For large $\alpha$, under the scaling 
 \begin{equation}
    x:=  - \frac{2^{1/3}}{\left(\alpha- \frac{1}{2}\right)^{2/3}} y,
 \end{equation}
 $u^{(\alpha)}_{\HM}$ has two distinct behaviors in the $x$-plane. On one hand, the leading-order asymptotic formula is an analytic function in a certain region. We call this section of the $x$-plane \textit{the pole-free region}. On the other hand, in the complement of the closure of the pole-free region, the leading-order asymptotic formula is not analytic and has poles at specific points (depending on $\alpha$). We call this section of the $x$-plane \textit{the pole region}. (See Proposition \ref{defOfGenusZeroRegion} for a precise description of the pole-free region.) Notice, for $\alpha$ large enough the scaled generalized Hastings-McLeod functions are analytic in the pole-free region and meromorphic in the pole region. Our Riemann-Hilbert analysis detects this bifurcation. Indeed, in the pole-free region the model problem is solved in terms of elementary functions, while in the pole region, it is solved in terms of Riemann $\Theta$-functions.

\subsubsection{Notation} Throughout the paper 
all matrices are denoted by a bold, uppercase letter with the exception of the Pauli matrix
\begin{equation}
    \sigma_3 := \begin{bmatrix}
        1 & 0 \\ 0 & -1
    \end{bmatrix}.
\end{equation}
We also use the following short hand:
\begin{equation}\label{notationDefs}
    \boldsymbol{U}_{k}(f) := \UT{e^{kf}}, \hspace{1cm} \boldsymbol{L}_{k}(f):= \LT{e^{kf}}, \hspace{1cm} \boldsymbol{T}(a):= \begin{bmatrix}
        0 & -a^{-1} \\
        a & 0
    \end{bmatrix}.
\end{equation}
Finally, given a smooth, oriented curve $\Sigma$ and a function $f$ that can be continuously extended to $z_0 \in \Sigma$,  we let $f_{+}(z_{0})$ denote the non-tangential limit of $f(z)$ as $z$ approaches $z_0$ from the left of $\Sigma$. Similarly, $f_{-}(z_0)$ denotes the non-tangential limit as $z$ approaches $z_0$ from the right of $\Sigma$.
\subsubsection{Outline of the Paper}
First, we state our main results in Theorem \ref{genus0 theorem} and Theorem \ref{genus1 theorem} and verify these results by comparing our asymptotic formulae to numerics. Second, in Section \ref{OGRHP} we state the Riemann-Hilbert representation (and its scaled version) of  the generalized Hastings-McLeod functions. Next, in Section \ref{GenusZeroAnalysis}, we give an explicit and precise definition for the pole-free region of the $x$-plane. Then, we proceed to apply the Deift-Zhou steepest-decent method in the pole-free region, leading to a proof of Theorem \ref{genus0 theorem}. In Section \ref{GenusOneAnalysis}, we carry out similar analysis in the pole region to prove Theorem \eqref{genus1 theorem}. Finally, in Section \ref{Numerics}, we review the numerical methods we implemented to verify our asymptotic formulae.

\subsection{Acknowledgments} 
The first author was supported by the Charles Phelps Taft Research Center by a Dissertation Fellowship.  The second author was supported by the National Science Foundation via grant DMS–2108019.  The authors also thank Deniz Bilman, Karl Liechty, Peter Miller, and Andrei Prokhorov for useful discussions.
\subsection{Results} 
It is convenient to state our results in terms of the large parameter $k$, where
\begin{equation}
    k:= \alpha - \frac{1}{2}.
\end{equation}

\begin{theorem} \label{genus0 theorem}
    Let $x$ be in the pole-free region of the $x$-plane, as defined in Section \ref{poleFreeRegSection}. Set $\Sigma_S$ to be the union of the closed rays emanating from the points $x= 3 e^{\frac{2\pi i}{3}}$ and $x= 3 e^{-\frac{2\pi i}{3}}$ with arguments $2\pi/3$ and $-2\pi/3$ respectively. Let $S(x) \equiv S$ denote the solution to the cubic equation
    \begin{equation}\label{cubicEqu}
        S^3 +xS -2i =0
    \end{equation}
    cut on $\Sigma_S$ which has asymptotic behavior $S(x) =-i\sqrt{x} + \Oh(1/x)$ as $x \to + \infty$ along the positive real line and $S(x)= 2i/x + \Oh\left( 1/x^2 \right)$ as $x \to - \infty$ along the negative real line. Then, for large $k \in \N$,
    \begin{equation}\label{genus0Result}
        -(2k)^{-1/3}u^{(k+1/2)}_{\HM}\left( - \frac{k^{2/3}}{2^{1/3}}x\right)= -i \frac{S(x)}{2} + \Oh\left( 1/k \right).
    \end{equation}
\end{theorem}

\begin{theorem} \label{genus1 theorem}
    Let $x$ be in the pole region of the $x$-plane, i.e.\ the complement of the closure of the pole-free region. Suppose $A,B,C,D \in \Ci$ satisfy the moment conditions \eqref{g1MomentConditions} and the Boutroux conditions \eqref{Boutroux}. For large integers $k$, if $x \in \mathscr{S}_k$ (see \eqref{SwissCheese}), then
    \begin{equation}\label{StatementThatDoesNotSayMuch}
        -(2k)^{-1/3}u^{(\alpha)}_{\HM}\left( - \frac{k^{2/3}}{2^{1/3}}x\right) = i\left(  \frac{-[\dotbold{Q}_{-2}]_{12}+[\dotbold{Q}_{-1}]_{22}[\dotbold{Q}_{-1}]_{12}+ \Oh(1/k)}{ [\dotbold{Q}_{-1}]_{12} + \Oh(1/k)} \right),
    \end{equation}
    where the quantities $[\dotbold{Q}_{-2}]_{12}$, $[\dotbold{Q}_{-1}]_{12}$, and $[\dotbold{Q}_{-1}]_{22}$ are defined in \eqref{dotQ 12-entry -1}, \eqref{dotQ 12-entry -2}, and \eqref{dotQ 22-entry -1}. Further, if $[\dotbold{Q}_{-1}]_{12} \ne 0$, then the left-hand side of \eqref{StatementThatDoesNotSayMuch} can be expressed in terms of Riemann $\Theta$-functions. That is,
    \begin{equation}\label{genusOneResult}
        -(2k)^{-1/3}u^{(\alpha)}_{\HM}\left( - \frac{k^{2/3}}{2^{1/3}}x\right) = i \left(\A_{-1} \cdot (\LD_{22}-\LD_{12})- \frac{B^2+D^2-A^2-C^2}{2(B+D-A-C)} \right) + \Oh(1/k),
    \end{equation} where
    \begin{equation}\label{L22}
        \LD_{22} = \frac{\Theta'(\A(\infty)+\A(Q)+\K+kF_1U;\B)}{\Theta(\A(\infty)+\A(Q)+\K+kF_1U;\B)} - \frac{\Theta'(\A(\infty)+\A(Q)+\K;\B)}{\Theta(\A(\infty)+\A(Q)+\K;\B)},
    \end{equation}
    \begin{equation}\label{L12}
        \LD_{12} =  \frac{\Theta'(\A(\infty)-\A(Q)-\K+kF_1U;\B)}{\Theta(\A(\infty)-\A(Q)-\K+kF_1U;\B)} - \frac{\Theta'(\A(\infty)-\A(Q)-\K;\B)}{\Theta(\A(\infty)-\A(Q)-\K;\B)},
    \end{equation}
    and $\Theta(\cdot)$, $\B$, $\A_{-1}$, $\A(\infty)$, $\A(\cdot)$, $Q$, $\K$, $F_1$, and $U$ are defined in \eqref{defOfRiemannThetaFunction}, \eqref{defOfScriptedB}, \eqref{LargeZAbelMap}, \eqref{AbelMapDef}, \eqref{zero of fDia}, \eqref{RiemannConstant}, \eqref{F1Def}, and \eqref{Udef}.
    
\end{theorem}
Here, we compare the results between the asymptotic formulae derived via the steepest-decent method and the numerical methods discussed in Section \ref{Numerics}. In Figure \ref{realSlices} we take a slice on the real $x$-line. Since the real line is contained in the pole-free region, the asymptotic formula does not depend on $k$. So, as $k \to \infty$, $u^{(k+\frac{1}{2})}_{\HM}$ converges  (under the appropriate scaling) to the right-hand side of \eqref{genus0Result}. Notice, we can observe this convergence even with small values of $k$. In Figure \ref{slicePlots} we take a horizontal slice (namely the slice $\Imag(x)=-9i$) through the pole region of the $x$-plane. For $x$ not in a neighborhood of a pole, we again note an excellent match between the asymptotics and numerics even for small $k$.
\begin{figure}[h]
    \centering
    \begin{tabular}{cc}
        \scalebox{.585}{\includegraphics{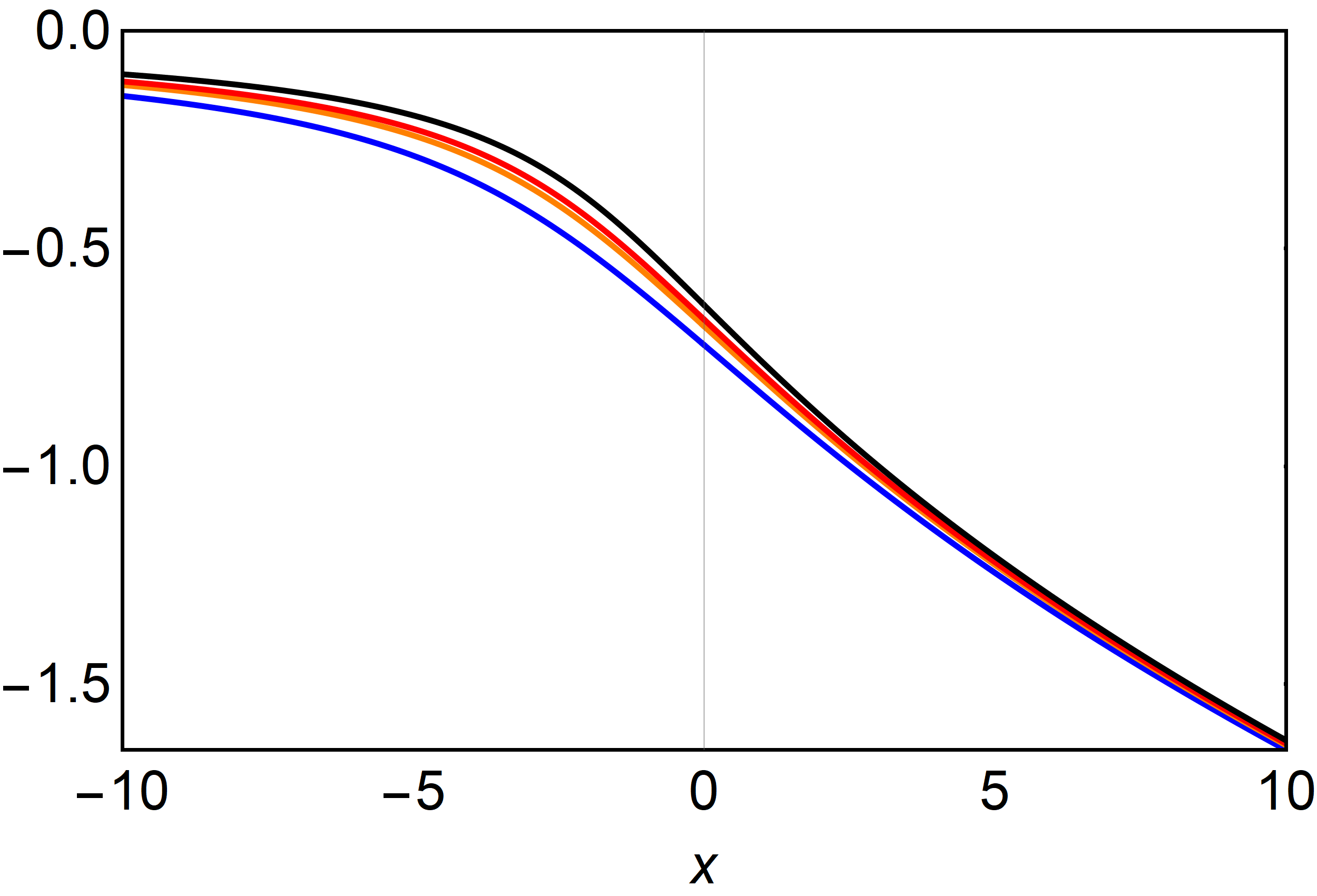}} & \scalebox{.585}{\includegraphics{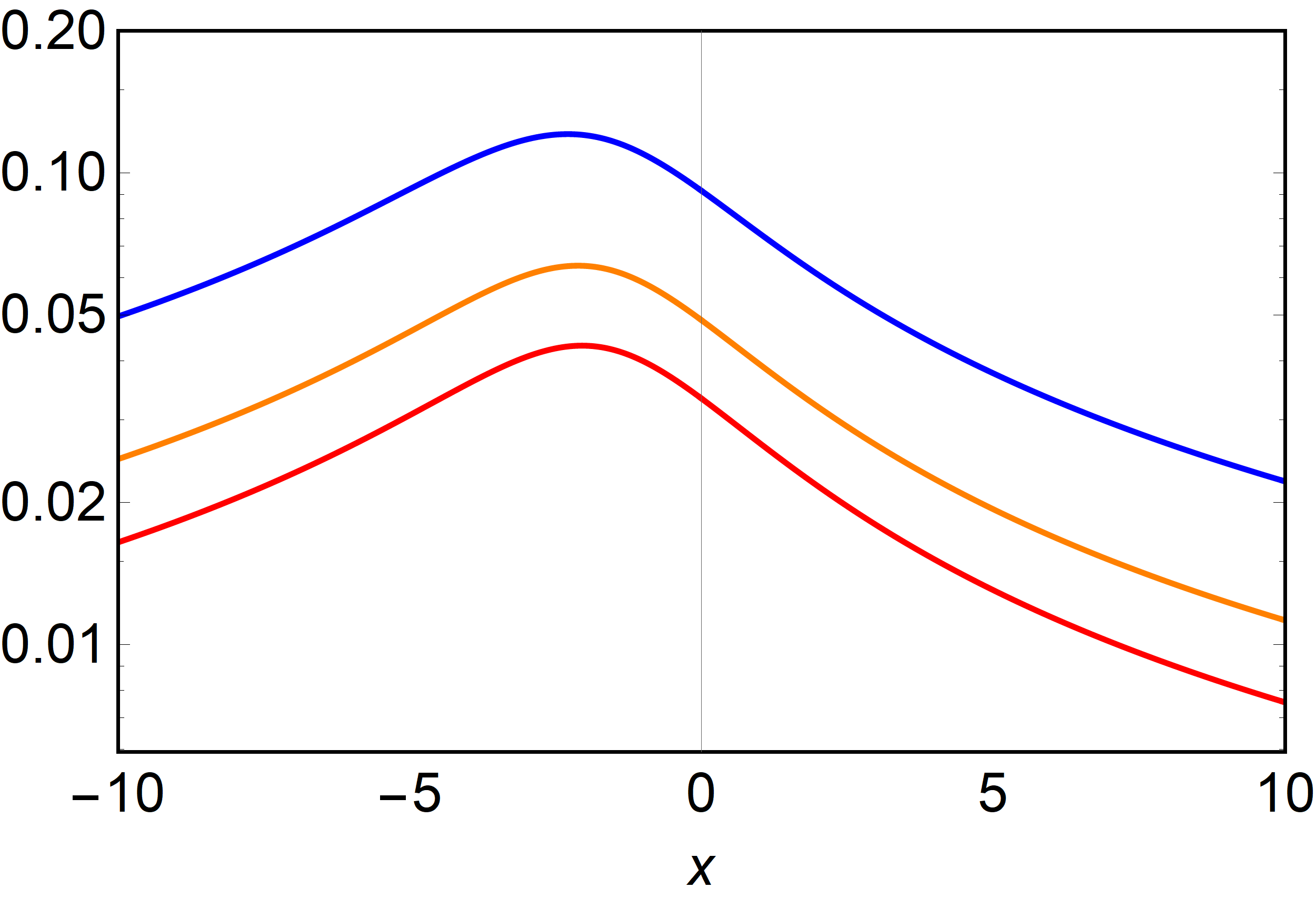}}
    \end{tabular}
    \caption{The left panel is a plot of the left-hand side of \eqref{genus0Result} evaluated on the real line for $k=1$ (blue), $2$ (orange), $3$ (red) compared to the black curve that is the right-hand side of \eqref{genus0Result}. The right panel is a logarithmic plot of the errors.} \label{realSlices}
\end{figure}

\begin{figure}
    \centering
    \begin{tabular}{ccc}
        $|\cdot|$ & $\Real(\cdot)$ & $\Imag(\cdot)$\\
        \scalebox{.39}{\includegraphics{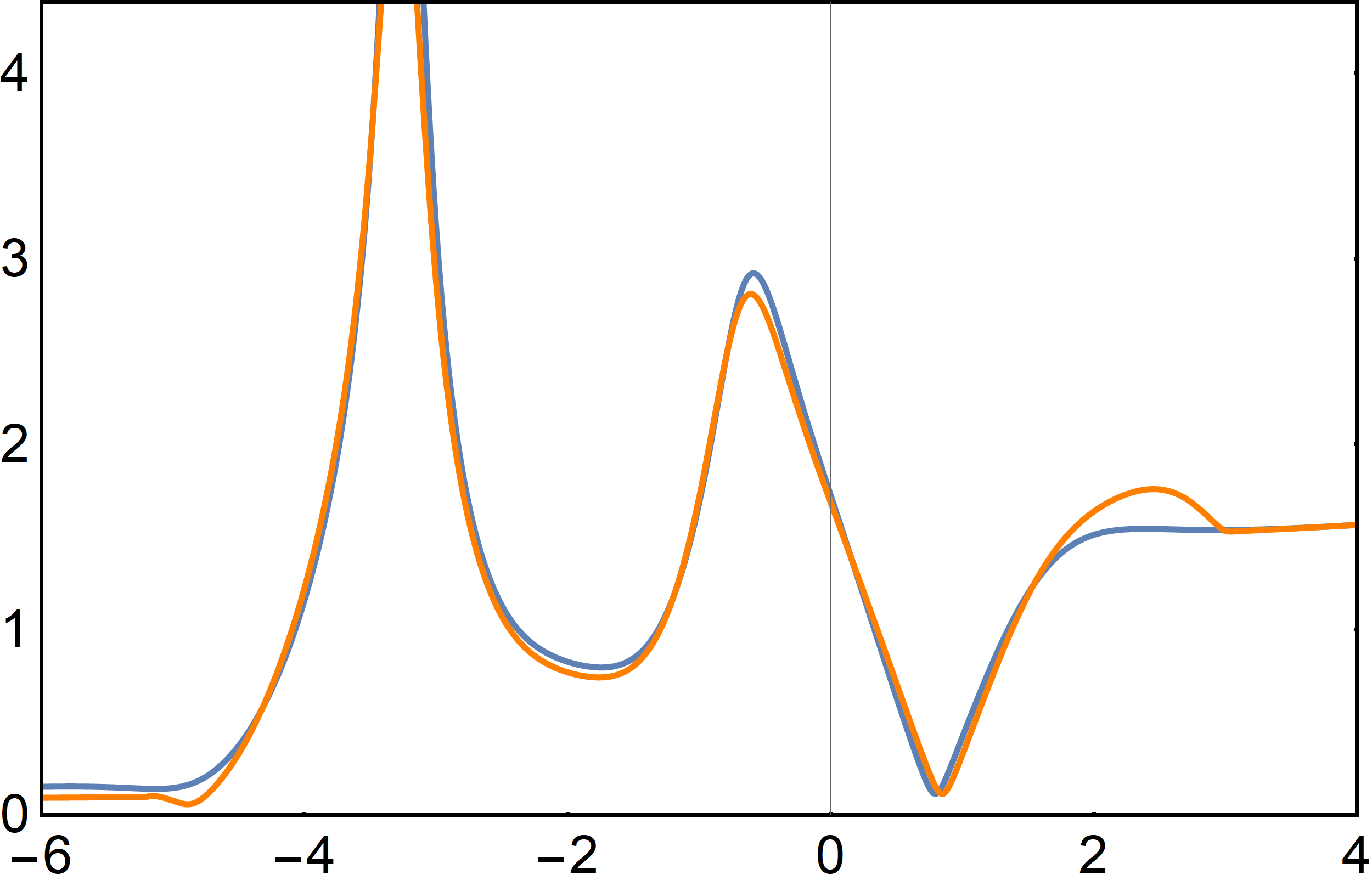}} & \scalebox{.39}{\includegraphics{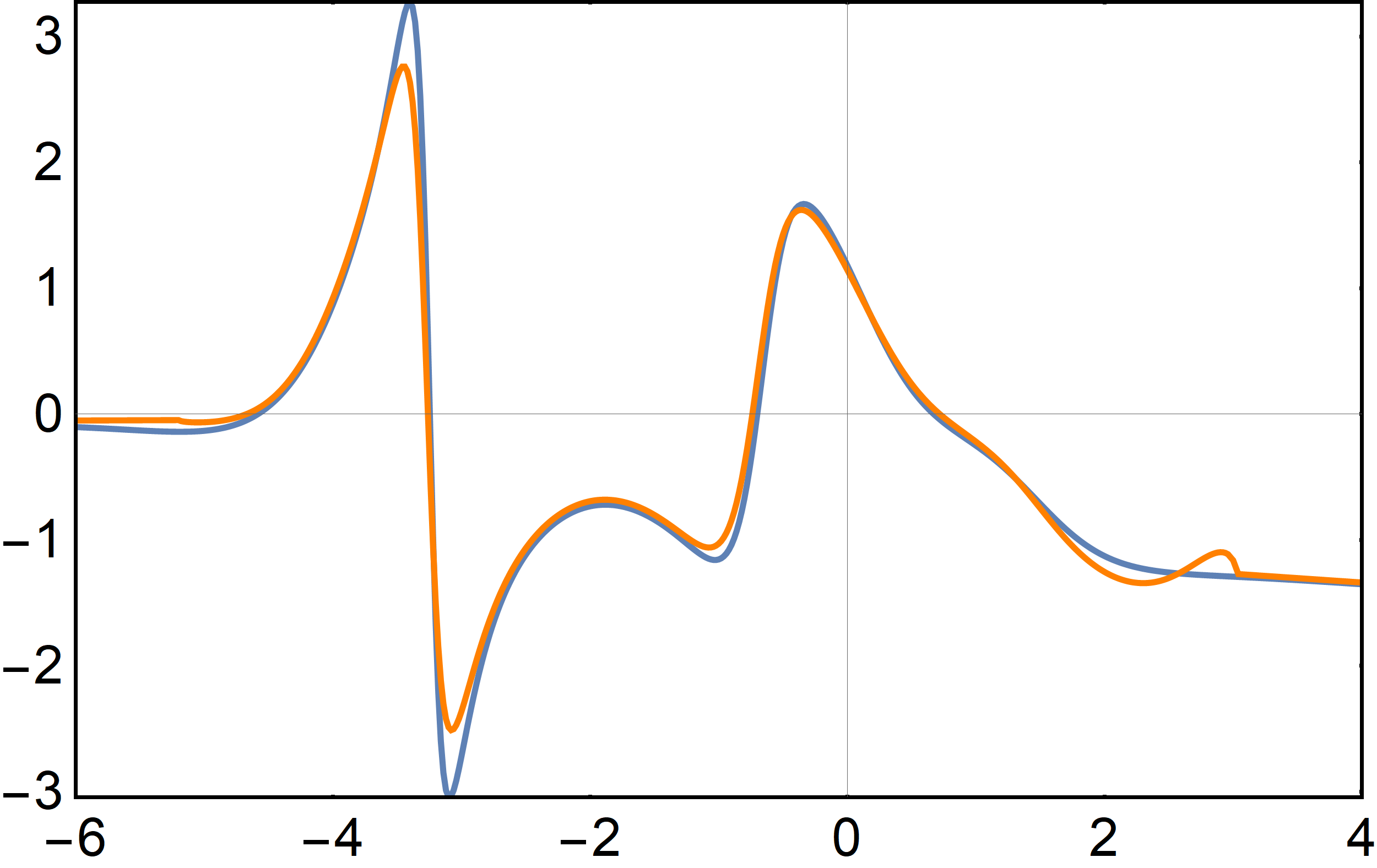}} & \scalebox{.39}{\includegraphics{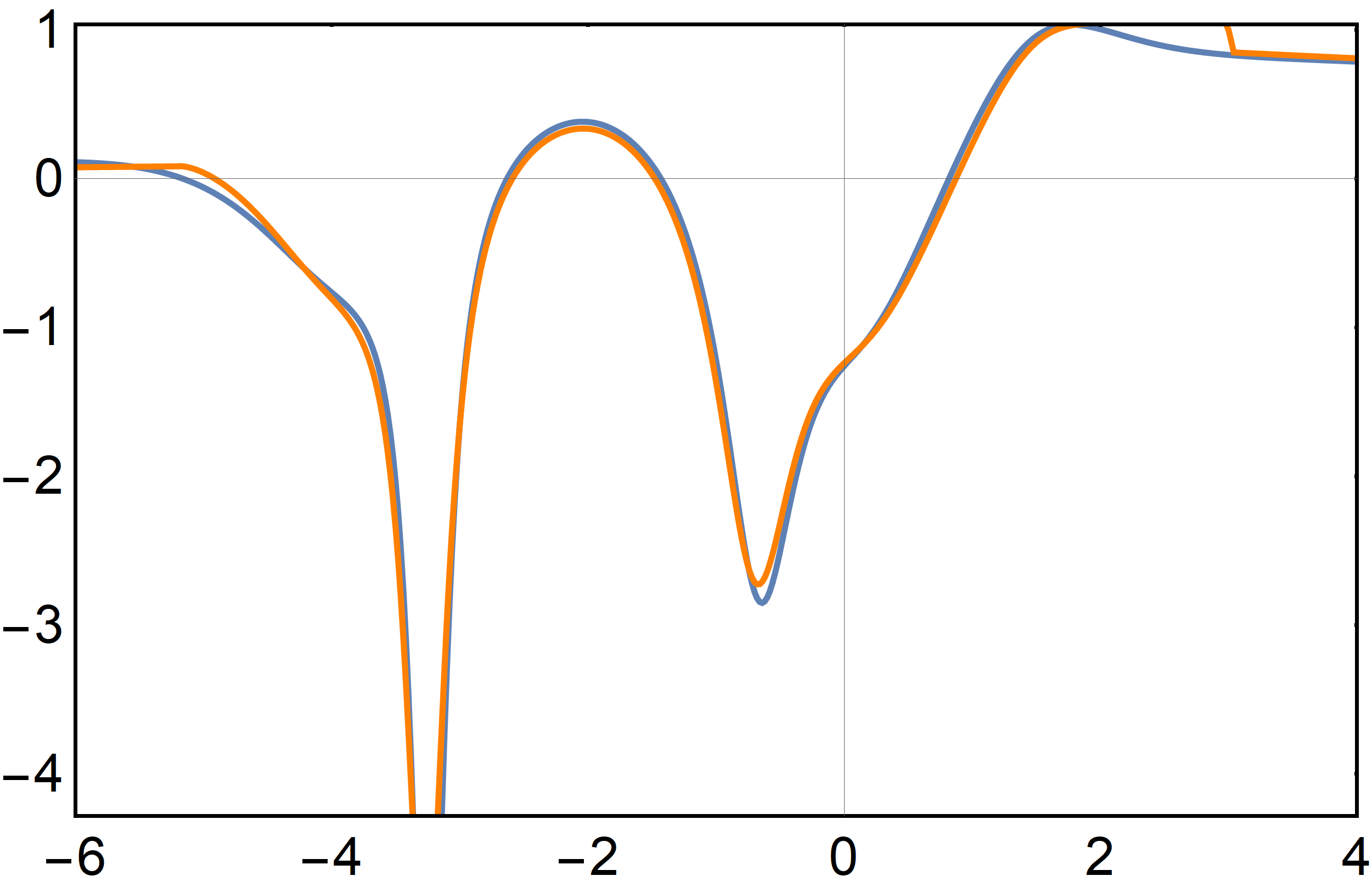}}\\
        \scalebox{.39}{\includegraphics{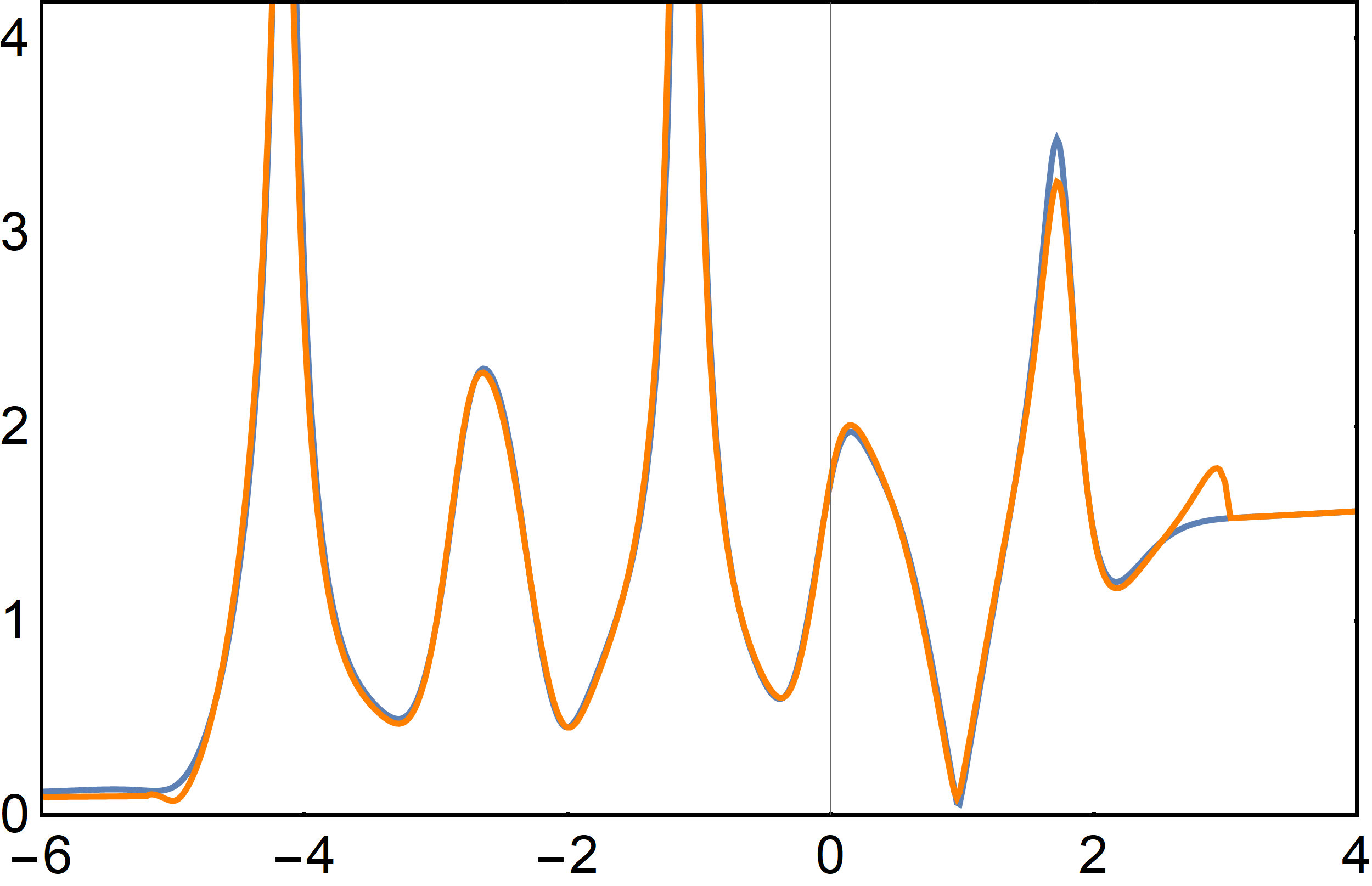}} & \scalebox{.39}{\includegraphics{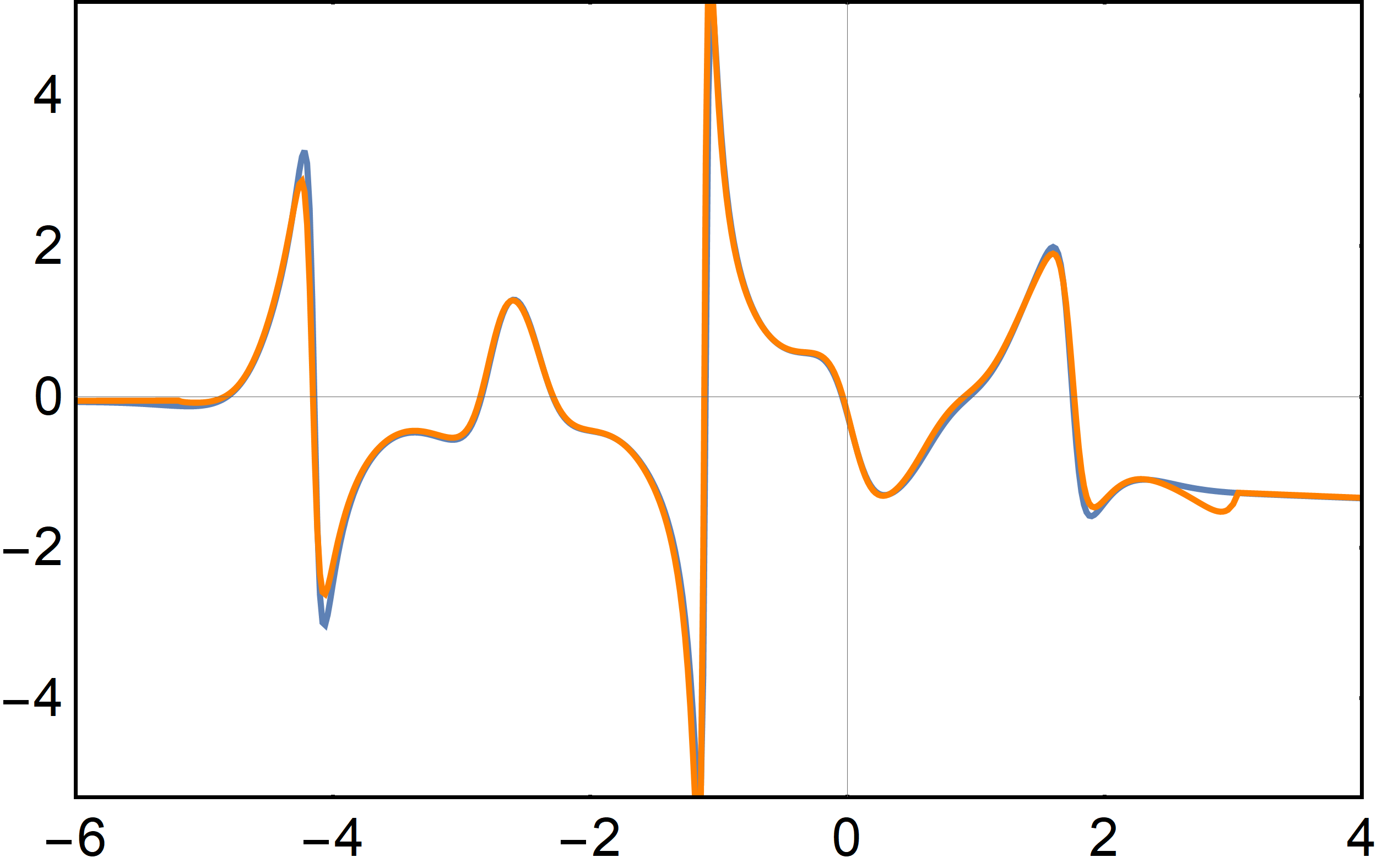}} & \scalebox{.39}{\includegraphics{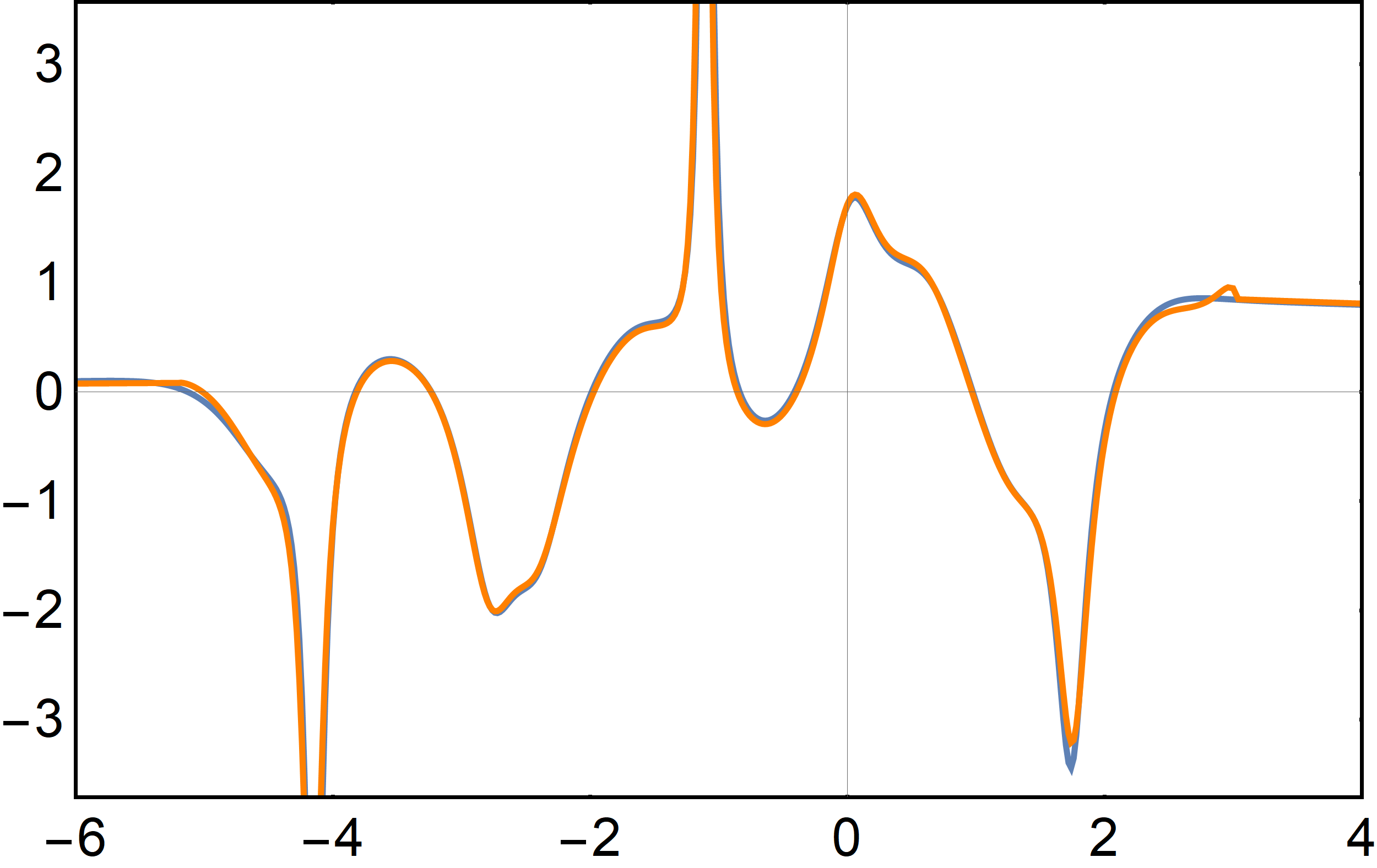}}\\
        \scalebox{.39}{\includegraphics{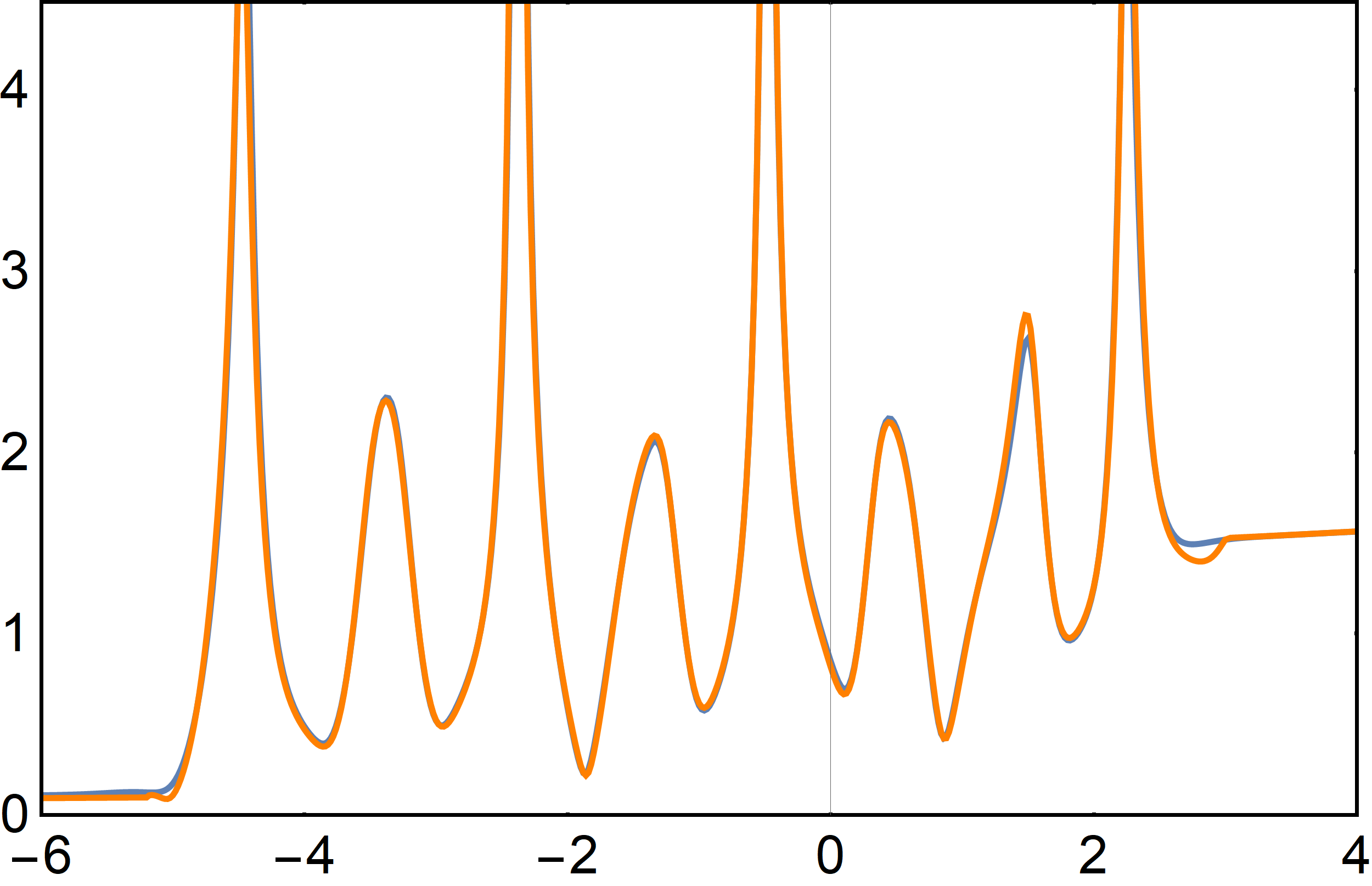}} & \scalebox{.39}{\includegraphics{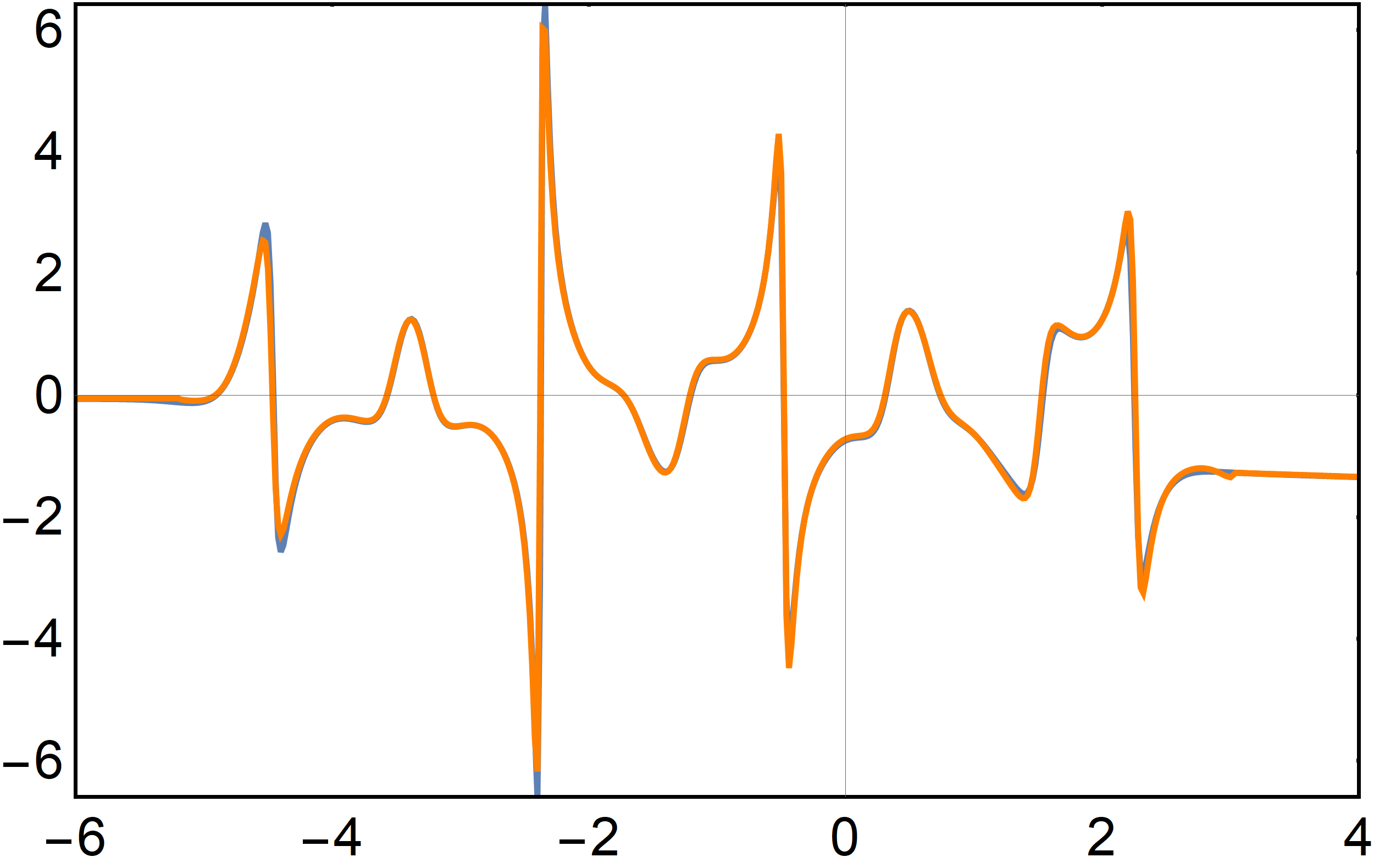}} & \scalebox{.39}{\includegraphics{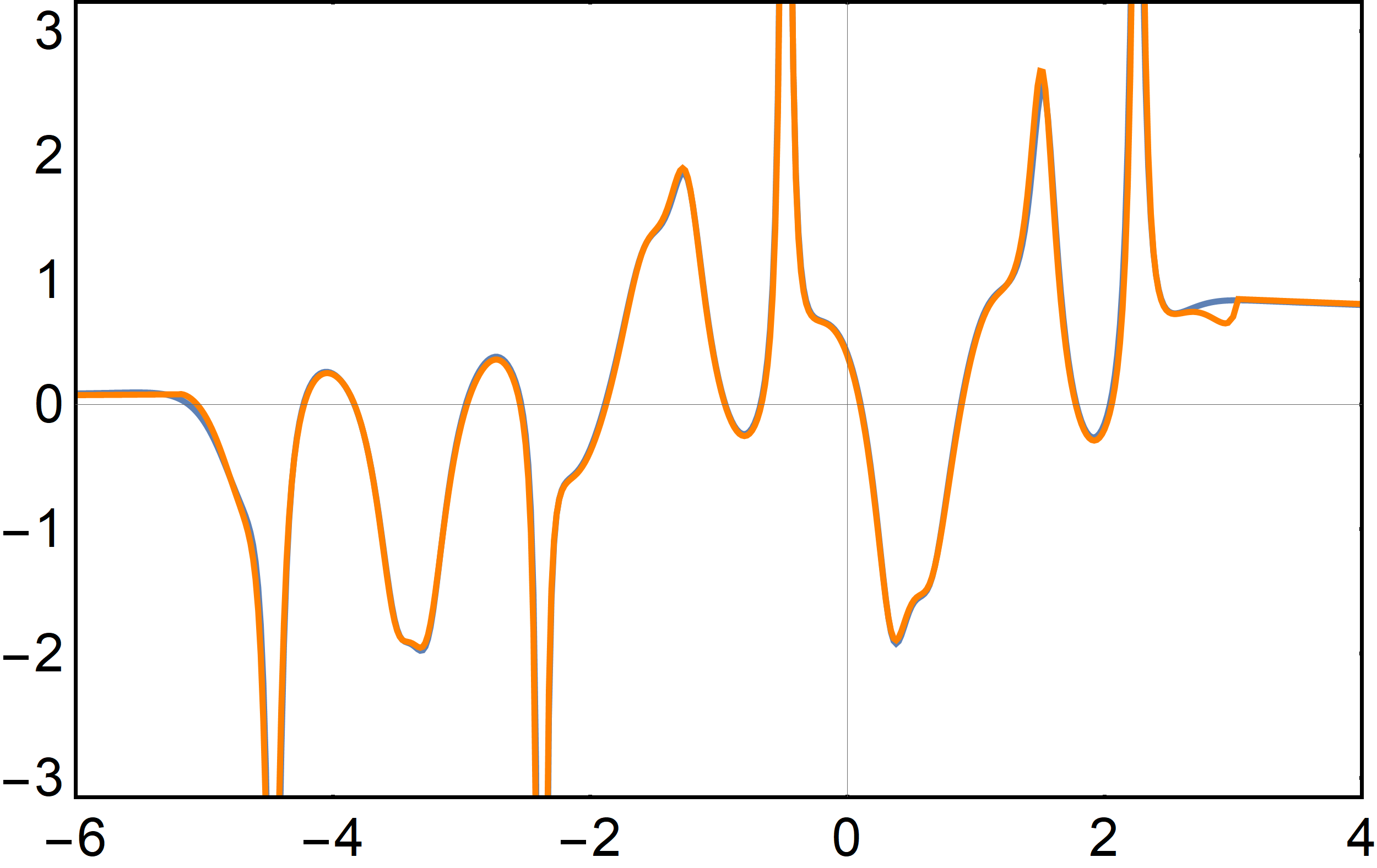}}
    \end{tabular}
    \caption{Plots of the left-hand side of \eqref{StatementThatDoesNotSayMuch} evaluated on the slice $\Imag(x)=-9$ in the $x$-plane for $k=1$ (top), $2$ (middle), $3$ (bottom). The orange curves show the asymptotic formulae described in Theorems \ref{genus0 theorem} and \ref{genus1 theorem}, while the blue curves are the result of the numerical methods discussed in Section \ref{Numerics}.} \label{slicePlots}
\end{figure}

\begin{figure}[]
    \centering
    \begin{tabular}{ccc}
        Asymptotic Formula & & Numerical Plot \\
        \scalebox{.49}{\includegraphics{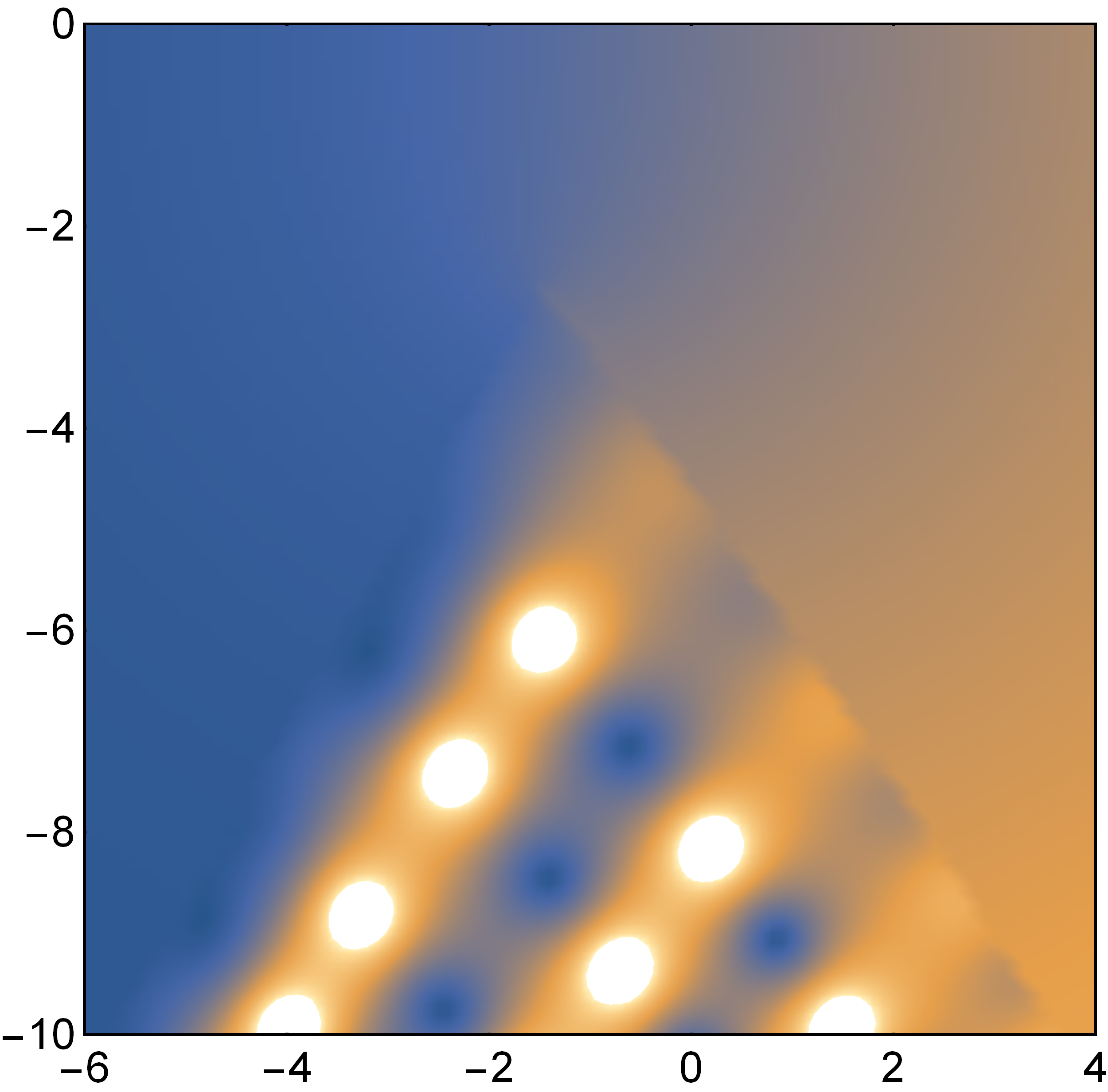}} & & \scalebox{.49}{\includegraphics{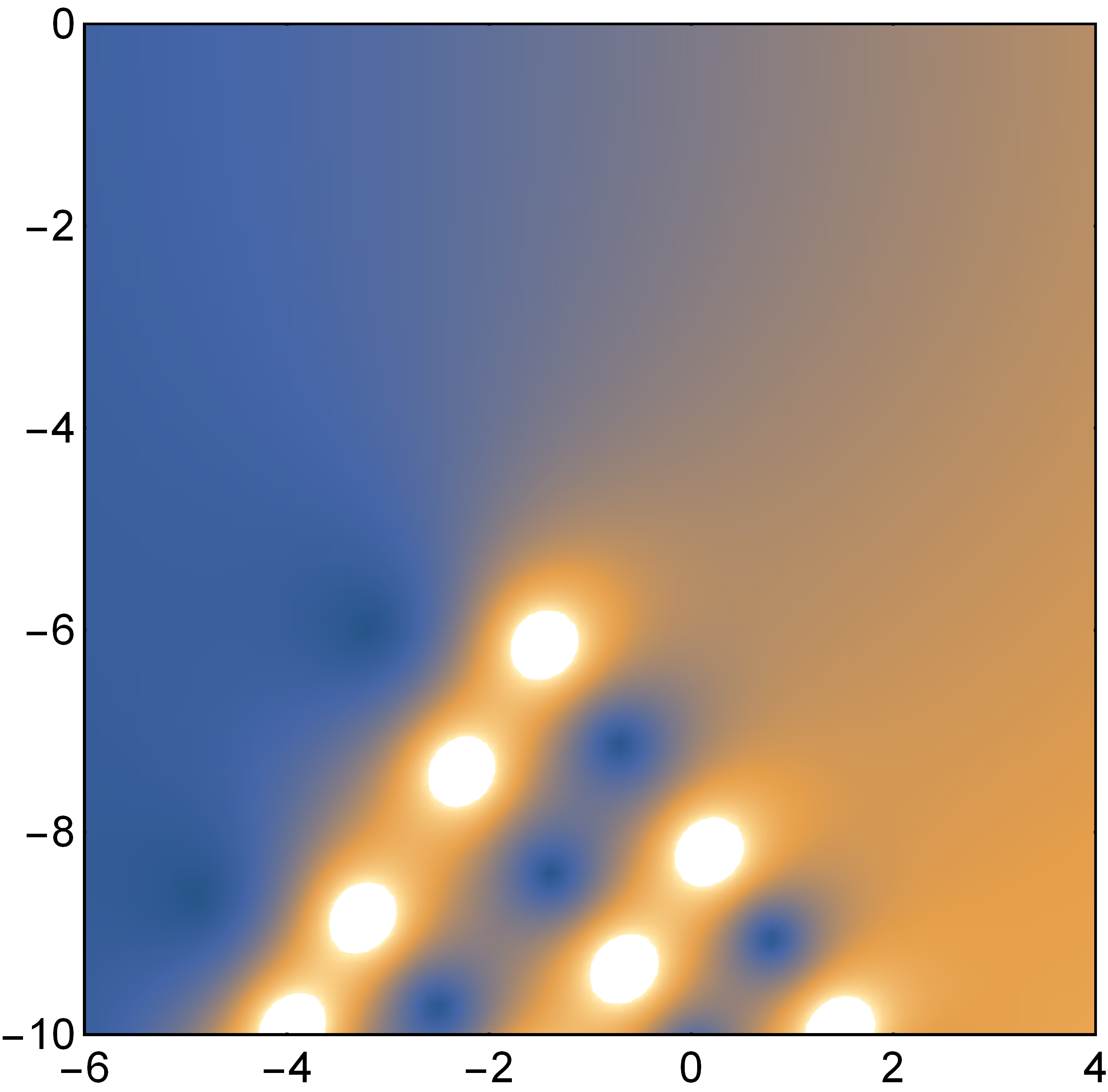}} \\
        \scalebox{.49}{\includegraphics{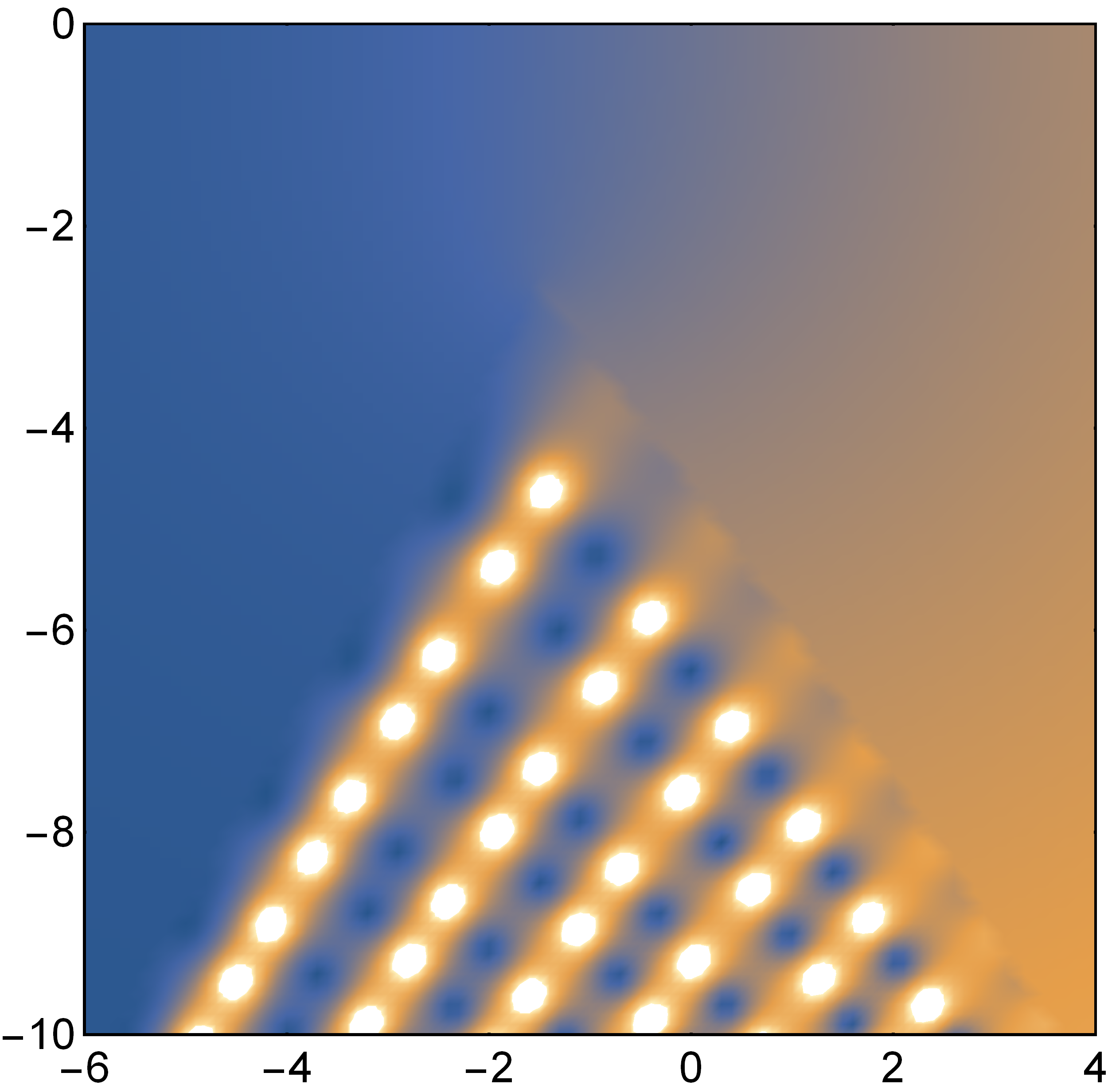}} & & \scalebox{.49}{\includegraphics{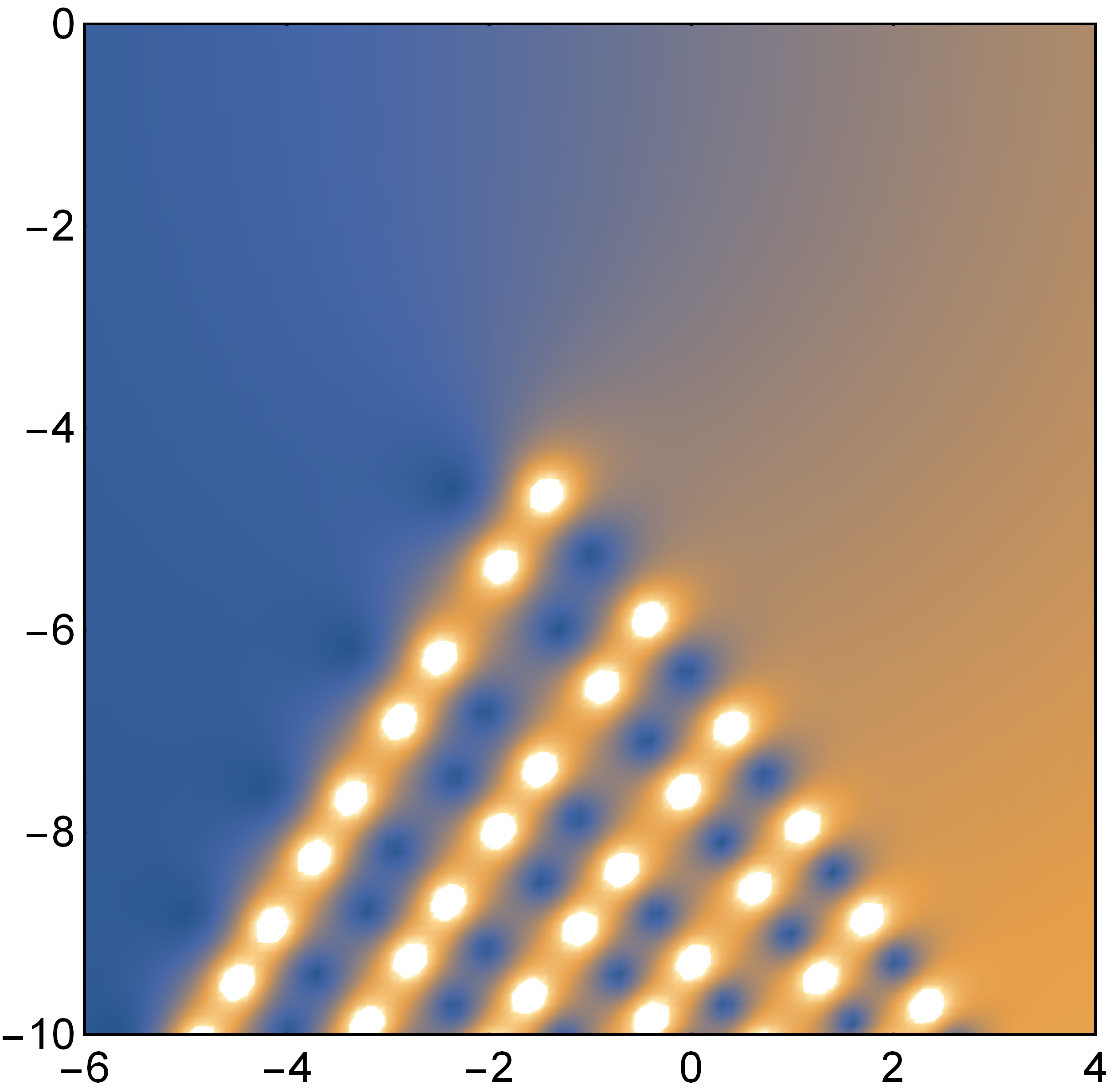}} \\
        \scalebox{.49}{\includegraphics{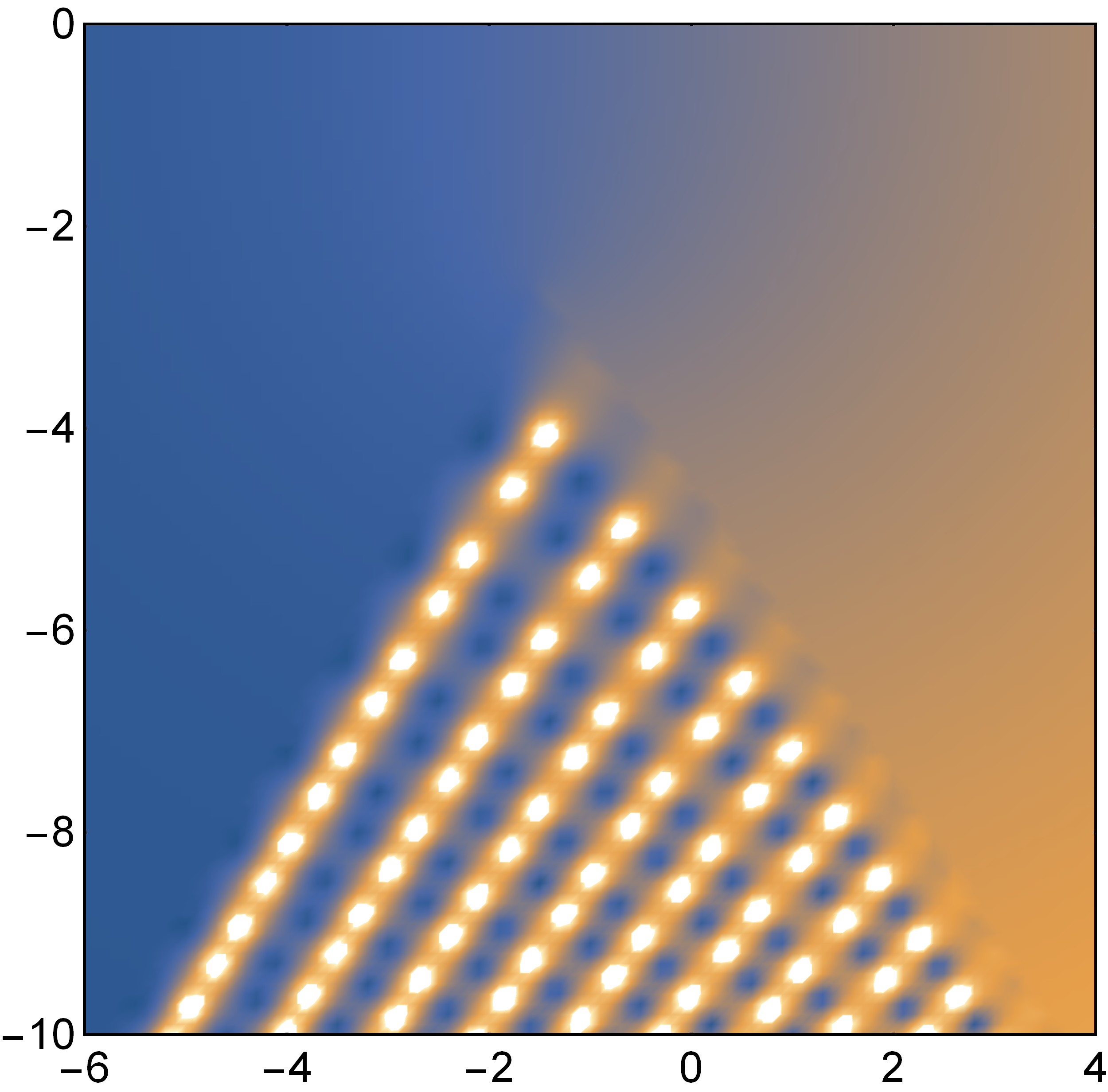}} & & \scalebox{.49}{\includegraphics{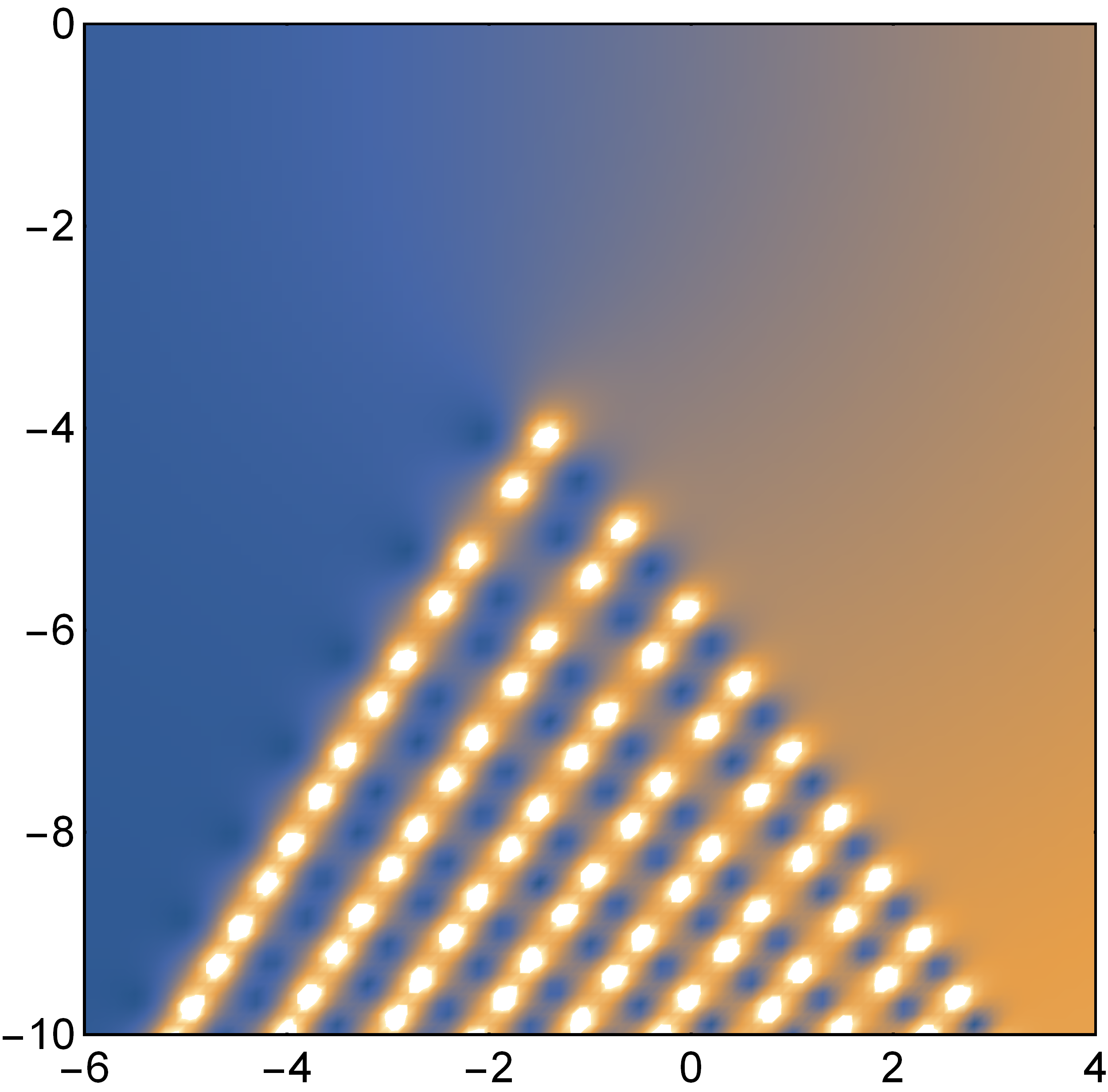}} 
    \end{tabular}
    \caption{Density plots in the $x$-plane of the scaled generalized Hastings-McLeod functions for $k=1,2,3$. The left column uses the asymptotic formulae described in Theorems \ref{genus0 theorem} and \ref{genus1 theorem}, while the right column uses the numerical methods discussed in Section \ref{Numerics}.}
\end{figure}

\section{The Initial Riemann-Hilbert Problem} \label{OGRHP}

The Painlev\'e-II equation \eqref{PII} is equivalent to 
\begin{equation}\label{NoncanonPII}
    \frac{1}{2} p''(s) = p(s)^3 - sp(s) + \frac{1}{2} + k.
\end{equation}
Indeed, since  $u_\HM^{(\alpha)}(y)$ satisfies \eqref{PII} it follows
\begin{equation}\label{AnnoyingScalingFormula}
    p_k(s) := -2^{1/3} u_\HM^{(k+ 1/2)} \left( -2^{1/3} s \right)
\end{equation}
satisfies \eqref{NoncanonPII} (recall, $k=\alpha+1/2$).  

In 2018, Buckingham and Liechty discovered $p_k(s)$ can be extracted  from the following $2 \times 2$ Riemann-Hilbert problem \cite{Buckingham2017TheKP}.

\begin{RHP}\label{RHPZ}
    For each $k \in \N$ and $s \in \Ci$, find a $2 \times 2$ matrix-valued function $\boldsymbol{Z}\left( \zeta ; s ,k \right) \equiv \boldsymbol{Z}(\zeta)$ so that
    \begin{enumerate}
        \item \textbf{Analyticity}: The matrix-valued function $\boldsymbol{Z}(\zeta)$ is analytic in $\zeta$ except on the four rays $\arg(\zeta) \in \left\{\pm\frac{\pi}{6}, \pm \frac{5\pi}{6}\right\}$ and is H\"older continuous up to the boundary in a neighborhood of $\zeta=0$.
        
        \item \textbf{Jump Condition:}  $\boldsymbol{Z}(\zeta)$ can be continuously extended to the boundary and the values taken by $\boldsymbol{Z}(\zeta)$ on the four rays of discontinuity are related by the jump condition $\boldsymbol{Z}_+(\zeta) =\boldsymbol{Z}_-(\zeta) \boldsymbol{V}^{(\boldsymbol{Z})}(\zeta)$, where the jump matrix $\boldsymbol{V}^{(\boldsymbol{Z})}(\zeta)$ is as shown in Figure \ref{RHPJumpsZ} and all rays are oriented towards infinity.
        \begin{figure}[h]
            \centering
            \begin{tikzpicture}[]
        
                \draw[thick] (0,0) -- (30:1.5cm);
                \draw[->, thick] (30:1.5cm) -- (30:3cm) node[right=0cm] { $
                \begin{bmatrix}
                    1 & 0 \\
                    e^{2i\left(\frac{4}{3}\zeta^3+s\zeta\right)} & 1
                \end{bmatrix}
                $};
                \draw[thick] (0,0) -- (150:1.5cm);
                \draw[->, thick] (150:1.5cm) -- (150: 3cm) node[left = 0cm] {
                    $
                    \begin{bmatrix}
                        1 & 0 \\
                        -e^{2i\left(\frac{4}{3}\zeta^3+s\zeta\right)} & 1
                    \end{bmatrix}
                    $};
                \draw[thick] (0,0) -- (210: 1.5 cm);
                \draw[->, thick] (210: 1.5cm) -- (210: 3cm) node[left = 0cm] {    $
                \begin{bmatrix}
                    1 & e^{-2i\left(\frac{4}{3}\zeta^3+s\zeta\right)} \\
                    0 & 1
                \end{bmatrix}
                $};
                \draw[thick] (0,0) -- (330: 1.5cm);
                \draw[->, thick] (330: 1.5 cm) -- (330: 3cm) node[right = 0cm] {    $
                \begin{bmatrix}
                    1 & -e^{-2i\left(\frac{4}{3}\zeta^3+s\zeta\right)} \\
                    0 & 1
                \end{bmatrix}
                $};
                \node at (0,-.25) {$0$};    
            \end{tikzpicture}
            \caption{The jumps of $\boldsymbol{Z}(\zeta;s,k)$.}  
            \label{RHPJumpsZ}
        \end{figure}
      
        \item \textbf{Normalization:} There exist constant (in $\zeta$) matrices $\boldsymbol{Z}_{-1}(s;k)\equiv \boldsymbol{Z}_{-1}$ and $\boldsymbol{Z}_{-2}(s;k) \equiv \boldsymbol{Z}_{-2}$ such that, for large $\zeta$,
        \begin{equation}
            \boldsymbol{Z}(\zeta)\zeta^{-k\sigma_3} = \I + \frac{\boldsymbol{Z}_{-1}}{\zeta} + \frac{\boldsymbol{Z}_{-2}}{\zeta^2} + \Oh(1/\zeta^3).
        \end{equation}
    \end{enumerate}
\end{RHP}
In particular, the extraction formula is 
\begin{equation} \label{unScaledExactSolution}
    p_k(s)= 2i \left(  [\boldsymbol{Z}_{-1}(s;k)]_{22} - \frac{[\boldsymbol{Z}_{-2}(s;k)]_{12}}{[\boldsymbol{Z}_{-1}(s;k)]_{12}}\right).
\end{equation}
We introduce the scalings
\begin{equation}\label{Scalings}
    z := \left( \frac{2}{k} \right)^{1/3} \zeta \hspace{.5 cm} \text{and} \hspace{.5 cm} x := \left( \frac{2}{k} \right)^{2/3} s
\end{equation}
and define
\begin{equation}
    \boldsymbol{M}(z;x,k) \equiv \boldsymbol{M}(z) := \left( \frac{k}{2}\right)^{-\frac{k}{3}\sigma_3}  \boldsymbol{Z}\left( \left( \frac{k}{2} \right)^{1/3} z ; \left( \frac{k}{2} \right)^{2/3} x ,k\right).
\end{equation}
$\boldsymbol{M}(z)$ solves the following Riemann-Hilbert problem:
\begin{RHP}\label{RHPMScaled}
    For each positive integer $k$, find $\boldsymbol{M}\left( z; x,k \right) \equiv \boldsymbol{M}(z)$ so that
    \begin{enumerate}
        \item \textbf{Analyticity}: $\boldsymbol{M}(z)$ is analytic in $z$ off the jumps contours shown in Figure \ref{RHPJumpsM}.
        \item \textbf{Jump Condition:}  $\boldsymbol{M}(z)$ can be continuously extended to the boundary and the boundary values taken by $\boldsymbol{M}(z)$ are related by the jump condition $\boldsymbol{M}_+(z)=\boldsymbol{M}_-(z)\boldsymbol{V}^{(\boldsymbol{M})}(z)$, where $\boldsymbol{V}^{(\boldsymbol{M})}(z)$ is as shown in Figure \ref{RHPJumpsM}, and
        \begin{equation}
            \theta(z:x) := \frac{4}{3} z^3 + xz.
        \end{equation}
        \begin{figure}[h]
            \centering
            \begin{tikzpicture}[]
                
                \draw[thick] (0,0) -- (30:1.5cm);
                \draw[->,  thick] (30:1.5cm) -- (30:3cm) node[right=0cm] {$\boldsymbol{L}_k(i\theta)$};
                \draw[thick] (0,0) -- (150:1.5cm);
                \draw[->,  thick] (150:1.5cm) -- (150: 3cm) node[left=0cm] {$\boldsymbol{L}^{-1}_k(i\theta)$};
                \draw[thick] (0,0) -- (210: 1.5 cm);
                \draw[->,  thick] (210: 1.5cm) -- (210: 3cm) node[left=0cm] {$\boldsymbol{U}_{k}(-i\theta)$};
                \draw[thick] (0,0) -- (330: 1.5cm);
                \draw[->, thick] (330: 1.5 cm) -- (330: 3cm) node[right=0cm] {$\boldsymbol{U}_{k}^{-1}(-i\theta$)};
                \node at (0,-.25) {$0$};           
            \end{tikzpicture}
            \caption{The jumps of $\boldsymbol{M}(z)$. Recall, $\boldsymbol{L}_{k}$ and $\boldsymbol{U}_k$ are defined \eqref{notationDefs}.}
            \label{RHPJumpsM}
        \end{figure} 
        
        \item \textbf{Normalization:}  There exist constant (in $z$) matrices $\boldsymbol{M}_{-1}(x;k) \equiv \boldsymbol{M}_{-1}$ and $\boldsymbol{M}_{-2}(x;k) \equiv \boldsymbol{M}_{-2}$ such that, for large $z$  
        \begin{equation}\label{unwindingLargeMExp}
            \boldsymbol{M}(z) z^{-k\sigma_3} = \I + \frac{\boldsymbol{M}_{-1}}{z} + \frac{\boldsymbol{M}_{-2}}{z^2} + \Oh\left(1/z^3\right).  
        \end{equation}
    \end{enumerate}
\end{RHP}
A direct calculation shows that under this scaling the extraction formula \eqref{unScaledExactSolution} becomes
\begin{equation}\label{goodRiddance}
    p_k\left( \left(\frac{k}{2}\right)^{2/3} x\right)= 2i\left(\frac{k}{2}\right)^{1/3} \left([\boldsymbol{M}_{-1}(x;k)]_{22}-\frac{[\boldsymbol{M}_{-2}(x;k)]_{12}}{[\boldsymbol{M}_{-1}(x;k)]_{12}}\right).
\end{equation}

\section{Analysis in the Pole-Free Region} \label{GenusZeroAnalysis}

In this section we perform the Riemann-Hilbert analysis in the pole-free region. Our approach is modeled from the methodology presented in Section 3 of \cite{RationalPainleveII}. In the first part, we introduce a technical but standard tool called the $g$-function. Once that is established, we will be able to define the pole-free region of the $x$-plane. In the subsequent section, we will use the $g$-function to reduce Riemann-Hilbert Problem \ref{RHPMScaled}  to an explicitly solvable Riemann-Hilbert problem, thereby giving us a rigorous way to approximate the left-hand side of  \eqref{unScaledExactSolution}. 
    
\subsection{The g-function} \label{the g-function section}
We start by constructing the $g$-function for the pole-free region. We  define the functions
    \begin{equation}\label{Delta a and b}
        \Delta(x) \equiv \Delta := \left( - \frac{4i}{S} \right)^{1/2}, \hspace{.18 in} a(x)\equiv a:= \frac{S-\Delta}{2}, \hspace{.18 in} b(x)\equiv b := \frac{S+\Delta}{2}, \hspace{.18 in} c(x)\equiv c:= -\frac{S}{2}. 
    \end{equation}
    The square root in \eqref{Delta a and b} is principal. This choice guarantees that  $\Real(a) \le \Real(b)$. Now, let $\Sigma$ denote a finite, simple, oriented arc that connects $a$ to $b$. (See Remark \ref{sigmaPlacementg0} on the exact placement of $\Sigma$). Given $\Sigma$, let $r(z;x)\equiv r(z)$ be the function analytic for $z \in \Ci \setminus \Sigma$ that satisfies the conditions
    \begin{equation}\label{defOfr(z)}
        r(z)^2 = (z-a)(z-b) \hspace{.33 in } \text{and} \hspace{.33 in}  r(z)=z+\Oh(1) \text{ as } z \to \infty.
    \end{equation}
    To define the $g$-function, we first define its derivative and integrate up. So,
    \begin{equation}\label{gprimeG0}
        g'(z;x) \equiv g'(z) := \frac{i\theta'(z)}{2} - i(S+2z)r(z),
    \end{equation}
    where $g'$ is cut on $\Sigma$. We find that the large-$z$ expansion of $g'(z)$ is
   \begin{equation}\label{gPrime in terms of a and b}
    g'(z) = i\frac{2S^2+2x+\Delta^2}{4} + i\frac{S\Delta^2}{4z}  + \Oh\left(1/z^2 \right).
   \end{equation}
    An equivalent way to write \eqref{cubicEqu} and $\Delta^2$ is
    \begin{equation}\label{system of S and Delta}
        2S^2 + 2x + \Delta^2= 0 \hspace{.66cm} \text{and} \hspace{.66cm} \Delta^2= -\frac{4i}{S}.
    \end{equation}
   Plugging \eqref{system of S and Delta} into \eqref{gPrime in terms of a and b}, we see that $g'(z)=1/z + \Oh(1/z^2)$ as $z \to \infty$. Thus, $g'(z)-1/(z-a)$ is integrable at $z=\infty$ and we can use the fundamental theorem of calculus to define $g(z)$ as an antiderivative of $g'(z)$. However, we must introduce a logarithmic cut. Let $L$ denote an oriented, unbounded arc from $z=\infty$ to $z=a$. Further, suppose that $L$ only intersects $\Sigma$ at $z=a$ and that, for large enough $z$, $L$ coincides with the negative real numbers. Now, we define $g(z)$ as an integral where the path of integration can be any path that does not cross $\Sigma$. Set
    \begin{equation}\label{gofzFormulag0}
        g(z;x) \equiv g(z) := \log(z-a) + \int_{\infty}^{z} g'(w) - \frac{1}{w-a} \dd w.
    \end{equation}

    For notational convenience, we write
   \begin{equation}\label{defhg0}
    h(z;x) \equiv h(z) := \frac{i\theta(z)}{2}- g(z), \hspace{.2 in} z \in \Ci \setminus \Sigma, \hspace{.2 in} x \in \Ci \setminus \Sigma_S.
   \end{equation}
   The behavior of $h(z)$ near the points $z=a$ and $z=b$ plays an important role in the Riemann-Hilbert analysis. Let $\D_a$ and $\D_b$ be sufficiently small disks around $z=a$ and $z=b$, respectively. Notice that $L \cup \Sigma$ divides $\D_a$ into two components. Let $\D_a^{\uparrow}$ denote the component that is on the plus side (or left side) of $L \cup \Sigma$, and let $\D_a^{\downarrow}$ denote the component that is on the minus side (or right side) of $L \cup \Sigma$.

   We now list some basic properties of $g$ and $h$ that we will use throughout the Riemann-Hilbert analysis.

   \begin{proposition}\label{properties of g and h g0}
        Given that $S,\Delta,a,b$ are defined, guaranteeing that $g(z)$ and $h(z)$ are well-defined by \eqref{gofzFormulag0} and \eqref{defhg0}, respectively, then for each $x \in \Ci \setminus \Sigma_S$ there exists a complex number $\lambda(x) \equiv \lambda$ such that for all $z \in \Sigma$ we have that
        \begin{equation} \label{jumpsofhg0}
            h_{+}(z) + h_{-}(z) = i\theta(z) - g_{+}(z) - g_{-}(z) = -\lambda.
        \end{equation}
        Further, if $x \ne 3e^{\pm 2\pi i /3}$, then under the right perturbation, $2h+\lambda$ locally behaves like a $3/2$-root function near the endpoints $a$ and $b$. More precisely, if we consider the functions
        \begin{equation}
            J_{a}(z) = \begin{cases}
                \left( 2h(z)+\lambda - 2\pi i \right) & \text{if }z \in \D_a^{\uparrow} \\
                \left( 2h(z)+\lambda + 2\pi i \right) & \text{if } z \in \D_a^{\downarrow}
            \end{cases}
        \end{equation}
        and
        \begin{equation}
            J_{b}(z) = 2h(z) +\lambda,
        \end{equation}
        then for each endpoint $p \in \{a,b\}$ there exists $\kappa_p \in \Ci \setminus \{0\}$ such that
        \begin{equation} \label{nearEndpointsG0}
            J_{p}(z) = \kappa_p (z-p)^{3/2} + \Oh\left( (z-p)^{5/2} \right).
        \end{equation}
        Additionally, when $x \in \R$, for each $z \in \Ci \setminus (\Sigma \cup L)$
        \begin{equation}\label{symOfhg0}
            \Real \left(2h\left(-\overline{z}; x \right) + \lambda(x) \right) =  \Real \left(2h(z;x) + \lambda(x)\right).
        \end{equation}
        Furthermore, for $x \in \R$,
        \begin{equation}\label{Re(2h+lambda) on the real line}
            \text{there are at most two real }z\text{-values} \text{ such that } Re\left( 2h(z)+\lambda \right) =0.
        \end{equation}
        Also,
        \begin{equation} \label{logJump}
            g_{+}(z) - g_{-}(z) = -2\pi i, \hspace{.2 cm} z \in L.
        \end{equation}
        Finally, as $z \to \infty$
        \begin{equation} \label{gfunNormg0}
            g(z) = \log(z) + \Oh(1/z).
        \end{equation}
   \end{proposition}
Most of these statements are standard and follow from direct computation, however we furnish a proof of statements \eqref{symOfhg0} and \eqref{Re(2h+lambda) on the real line}in Appendix \ref{proofApp}.
The function $h(z)$, and the constant $\lambda$ can both be explicitly expressed in a closed form. Consider
\begin{equation}\label{etaHelperFunDef}
    \eta(z) := \big((1-z)(1+z)\big)^{1/2} = i\big((z-1)(z+1)\big)^{1/2},
\end{equation}
where $\eta$ is cut on $[-1,1]$ and as $z \to \infty$, $\eta(z)=iz + \Oh(1/z).$ Note that on the cut $[-1,1]$, both the plus and minus boundary values of $\eta$ are real valued. Now, we let $\myarcsin$ denote 
\begin{equation}\label{ArcSinBranch}
    \myarcsin(z) := -i\log\left( iz + \eta(z)\right).
\end{equation}
Then,
\begin{equation} \label{hofzExact}
    h(z) =  \frac{i}{6} r(z) \left(4z^2+2Sz+2x\right) - 2i\myarcsin\left( \left( \frac{a-z}{a-b} \right)^{1/2} \right) + \frac{i x S}{6} + \log\left(\frac{-4}{\Delta} \right)+ \frac{1}{3}.
\end{equation}
Moreover, as $\lambda = -(h_{+} + h_{-})$,
\begin{equation}\label{defOfLambda}
    \lambda = -2\left( \frac{ixS}{6}+ \log\left( \frac{-4}{\Delta} \right) + \frac{1}{3}\right).
\end{equation}

\subsection{Defining the Pole-Free Region} \label{poleFreeRegSection}

We now have the necessary objects established to define the pole-free region. In particular, $x$-membership of the pole-free region will depend on the topology of the zero-level set of $\Real(2h+\lambda)$ in the $z$-plane (for motivation see Figure \ref{zeroLevelLinesGenusZeroAll}). The following definition is directly adapted from Definition 2 in \cite{RationalPainleveII}.

\begin{definition}\label{genus0Def}
    We say $x \in \Ci$ is in the pole-free region of the complex $x$-plane if the following five conditions hold:
\begin{enumerate}
    \item[\textbf{(a1)}] There are exactly three arcs of the zero-level set of $\Real(2h+\lambda)$ that terminate at the endpoint $z=a$.
    \item[\textbf{(a2)}] There exist two curves emanating from $z = a(x)$ that never cross the zero-level set of $\Real(2h+\lambda)$ and tend to infinity at distinct angles $5\pi/6$ and $-5\pi/6$. 
    \item[\textbf{(a3)}] There are exactly three arcs of the zero-level set of $\Real(2h+\lambda)$ that terminate at the endpoint $z=b$.
    \item[\textbf{(a4)}] There exist two curves emanating from $z = b$ that never cross the zero-level set of $\Real(2h+\lambda)$ and tend to infinity at distinct angles $\pi/6$ and $-\pi/6$. 
    \item[\textbf{(a5)}] There exists a finite curve that connects $z=a$ to $z=b$ that never crosses the zero-level set of $\Real(2h+\lambda)$. We further require that the curve is below $\Sigma$. That is, except for the endpoints $z=a$ and $z=b$, the curve is entirely to the right side of $\Sigma$.
\end{enumerate}
\end{definition}
 As $x$ varies, for $x$ to leave or enter the pole-free region it is necessary (but not sufficient) for the zero-level set of $\Real(2h(z;x)+\lambda(x))$ to undergo a sudden topological change. That is, the zero-level set must intersect itself. Further, the Cauchy-Riemann equations tell us this can only occur at a point $z_0$ where $h'(z_0;x)=0$ (here the derivative is with respect to $z$). However, by \eqref{gprimeG0} and \eqref{defhg0}, for each fixed $x \in \Ci \setminus \Sigma_S$, $h'(z;x)$ can only vanish at $z=a(x)$, $z=b(x)$, or $z=c(x)$. In fact, since $a(x)$ and $b(x)$ will always lie on the zero-level set (by construction) and $a(x) \ne b(x)$ (because $\Delta(x)$ never vanishes), our point of interest is $z=c(x)$. Specifically, we want the $x$-values where $c(x)$ collides with the zero-level set and so we consider 
\begin{equation}\label{auxBoundaryFunction}
    \mathfrak{c}(x) := \Real(2h(c(x);x)+\lambda(x)).
\end{equation}
Notice that $g(z)$ (and hence $h(z)$) is continuous at $z=b$. So, at $z=b$ \eqref{jumpsofhg0} reads $h_{+}(b) + h_{-}(b) = 2h(b)= -\lambda$. Thus, by the fundamental theorem of calculus we have 
\begin{equation}\label{integralFormulaofGothc}
    \mathfrak{c}(x) = 2\Real \int_{b}^{c} h'(w) \dd w,
\end{equation}
where the path of integration does not cross $\Sigma$. Notice, this implies that $\mathfrak{c}(x)$ is continuous. The zero-level set of $\mathfrak{c}(x)$ divides the $x$-plane into four connected, open components $\Omega_{\text{right}}$, $\Omega_{\text{up}}$, $\Omega_{\text{left}}$, and $\Omega_{\text{down}}$; see Figure \ref{boundaryandBifunCurve}.

\begin{figure}[h]
    \centering
    \scalebox{.55}{\includegraphics{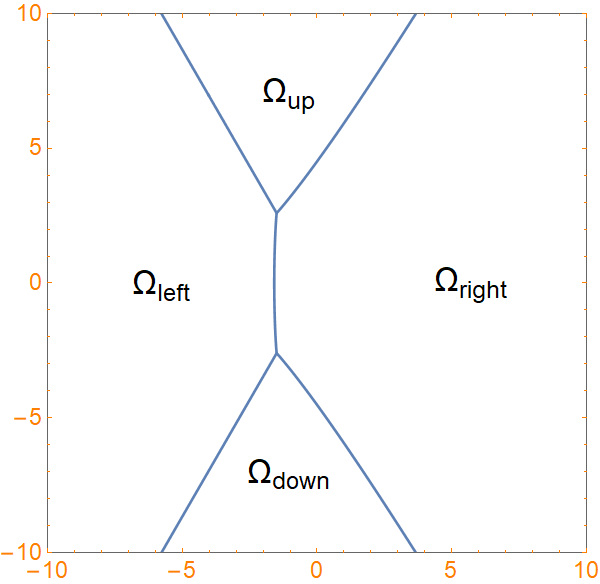}}
    \caption{The zero-level set of $\mathfrak{c}(x)$ in the $x$-plane.}
    \label{boundaryandBifunCurve}
\end{figure}

\begin{proposition}\label{defOfGenusZeroRegion}
    The pole-free region of the complex $x$-plane is exactly the union of $\Omega_{\text{left}}$, $\Omega_{\text{right}}$, and their shared boundary minus the apex points $x=3e^{\pm 2\pi i /3}$.
\end{proposition}
\begin{proof}
    Notice that conditions (\textbf{a1})--(\textbf{a5}) are preserved as $x$ varies continuously as long as $x$ does not cross the zero-level set of $\mathfrak{c}(x)$. So, in order to show that $\Omega_{\text{left}}$ (or $\Omega_{\text{right}}$) is contained in the pole-free region it suffices to show that at least one of its points is in the pole-free region.
    
    A direct calculation shows that for, $x \in \R$,
    \begin{equation} \label{x-asymptotics}
       \mathfrak{c}(x) = \begin{cases}
        \displaystyle \frac{\sqrt{2}}{3} (-x)^{3/2} + \Oh(1) & x \to - \infty 
         \\[10 pt] 
         \displaystyle -\frac{2}{3}x^{3/2} - \log\left( x^{3/2} \right) + \Oh(1)  & x \to +\infty.
        \end{cases}
    \end{equation}
    Given that $\Omega_{\text{left}}$ and $\Omega_{\text{right}}$ are defined by the zero-level set of $\mathfrak{c}(x)$, \eqref{x-asymptotics} tells us that $\mathfrak{c}(x)$ is positive in $\Omega_{\text{left}}$ and negative in $\Omega_{\text{right}}$.

    Now, we show that $\Omega_{\text{left}}$ is in the pole-free region. Fix $x_{\ell} \in \Omega_{\text{left}} \cap \R$. Notice, as $z \to \infty$ the $i\theta(z)$ part of $h(z)$ dominates because $g(z)= \log(z) + \Oh(1/z)$ (see \eqref{defhg0} and \eqref{gfunNormg0}). Thus, at $z=\infty$ we deduce that the zero-level set of $\Real(2h+\lambda)$ will have exactly six arcs at angles $0$, $\pi/3$, $2\pi/3$, $\pi$, $-2\pi/3$, and $-\pi/3$. Further, \eqref{nearEndpointsG0} shows that at $z=b$ there are exactly three arcs of the zero-level set of $\Real(2h+\lambda)$ emanating from it at angles separated by $2\pi /3$. Similarly, \eqref{nearEndpointsG0} shows that there are exactly three arcs of the zero-level set of $\Real(2h+\lambda)$ emanating from $z=a$ at angles separated by $2\pi /3$.

    As $\Real(2h(z) +\lambda)$ is harmonic in $z$ in $\Ci \setminus (\Sigma \cup L)$, the aforementioned arcs must connect pairwise  to each other. Further, once connected these six arcs comprise the entirety of the zero-level set of $\Real(2h+\lambda)$. In particular, it does not contain a finite arc, nor a closed loop. Indeed, the former case would violate the fact $\Real(2h+\lambda)$ is infinitely differentiable, while the latter case would violate the maximum modulus principle. So, our current task is to determine which configurations of connections are valid (i.e.\ connections that are consistent with Proposition \ref{properties of g and h g0}).

    We start by analyzing $2h'(z)$ on the imaginary axis in order to determine how many times the zero-level set intersects the imaginary axis. Towards that goal we consider the function $f :\R \to \Ci$ defined by 
    \begin{equation}\label{def of F aux function g0}
       f(v;x) \equiv f(v):= 2h(iv) + \lambda.
    \end{equation}
    A direct calculation shows that when $x$ is real, $a = - \overline{b}$, $c$ is purely imaginary, and $f'$ is a real-valued function with the sign chart shown in Figure \ref{signChartImaginaryAxishPrime}.
    \begin{figure}[h]
        \centering
        \begin{tikzpicture}
    
            \draw[->] (0,0) -- (6,0);
            \draw[->] (0,0) -- (-6,0);
    
            \draw[fill=black] (-1.5,0) circle (.075cm);
            \node at (-1.5,-.5) {$\rho$};
    
            \draw[fill=black] (1.5,0) circle (.075cm);
            \node at (1.5,-.5) {$\Imag(c)$};
    
            \node at (-3.75,.5) {{\LARGE $+$}};
    
            \node at (0,.5) {{\LARGE $-$}};
    
            \node at (3.75,.5) {{\LARGE $+$}};
    
            \node at (6.5,0) {$v$};
    
        \end{tikzpicture}
        \caption{A sign chart for $f'(v)$ (for any $x \in \R$), where $\rho$ is the modulus of the intersection point of $\Sigma$ and the imaginary axis, i.e.\ $i\rho \in \Sigma$.}
        \label{signChartImaginaryAxishPrime}
    \end{figure}
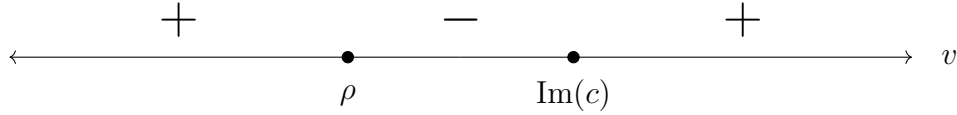

    Reading from Figure \ref{signChartImaginaryAxishPrime}, we see that on the interval $(\Imag(c),\infty)$ $\Real(f)$ has positive slope, while on $(\rho, \Imag(c))$ $F$ has negative slope. Since $\Real(f(\Imag c(x_{\ell})))=\mathfrak{c}(x_{\ell})>0$, it follows $\Real(f)$ does not vanish on $(\rho, \infty)$ and is positive. 
    
    At $v=\rho$, $f$ has a jump discontinuity because $i\rho \in \Sigma$ and $h$ has a jump discontinuity on $\Sigma$. In fact, \eqref{jumpsofhg0} implies that  $f_{+}(\rho) = - f_{-}(\rho)$. Thus, $\Real(f_{-}(\rho))<0$. Moreover, the fact $f$ has positive slope on $(-\infty, \rho)$ implies $\Real(f)$ does not vanish on $(-\infty,\rho)$. Therefore, $\Real(2h+\lambda)$ does not vanish on the imaginary axis. This means none of the zero-level arcs of $\Real(2h+\lambda)$ cross the imaginary axis.

    Next, since $x_{\ell} \in \R$,  the real part of $h(z)$ is symmetric about the imaginary axis, i.e.\ \eqref{symOfhg0} holds. Hence, we need only to consider the left half of the $z$-plane to ascertain how the zero-level arcs of $\Real(2h+\lambda)$ are connected.
    Now, label the zero-level arcs at $\infty$ with angles $2\pi/3$, $\pi$, and $-2\pi/3$ as $\textsc{arc}_1$, $\textsc{arc}_2$, and $\textsc{arc}_3$. For the arcs emanating from $z=a$, we label (in a counterclockwise fashion) $\textsc{arc}_{\alpha}$, $\textsc{arc}_{\beta}$, and $\textsc{arc}_{\gamma}$. Notice, if one of the arcs at $\infty$ (on the left side of the plane) did not connect to one of the arcs emanating from $z=a$, then two arcs at $z=a$ would be forced to connect to each other (because to connect to an arc in the right half of the plane would require it to cross the imaginary axis). However, this is not a valid connection because it forms a closed loop. This means we can assume without loss of generality that $\textsc{arc}_1$ is connected to $\textsc{arc}_{\alpha}$ (otherwise relabel the arcs emanating from $z=a$). Further, $\textsc{arc}_2$ is forced to connect to $\textsc{arc}_{\beta}$. Otherwise, the connection between $\textsc{arc}_3$ and $\textsc{arc}_{\beta}$ would intersect a zero-level arc; see Figure \ref{connectingLines}. Therefore, we have a complete picture of possible configurations of the zero-level arcs in the left half of the $z$-plane. Indeed, $\textsc{arc}_1$ is connected to $\textsc{arc}_{\alpha}$,  $\textsc{arc}_2$ is connected to $\textsc{arc}_{\beta}$, and  $\textsc{arc}_3$ is connected to $\textsc{arc}_{\gamma}$. Moreover, we can complete our picture of the entire $z$-plane using the symmetry \eqref{symOfhg0}; see Figure \ref{in Omega right}. Finally, we see that $x_{\ell}$ satisfies all of the axioms in the definition of the pole-free region and therefore $\Omega_{\text{left}}$ is completely contained in the pole-free region.  

    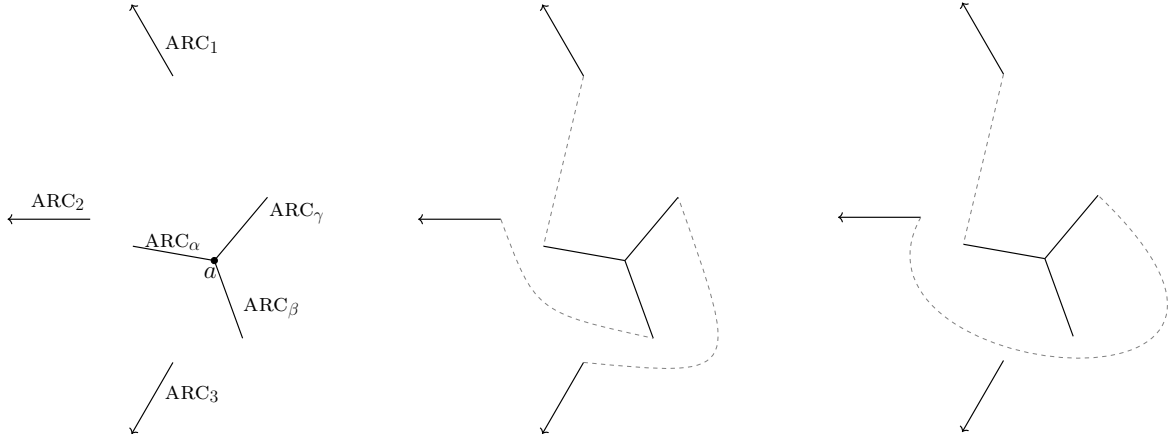
\begin{figure}
        \centering
        \begin{tabular}[h]{ccccc}
            \scalebox{.55}{\begin{tikzpicture}
        
                \draw[fill=black] (-1,-1) circle (.075cm);
                \draw [thick, xshift=-1cm, yshift=-1cm] (0,0) -- (50: 2cm);
                \node[ xshift=-1cm, yshift=-1cm] at (30: 2.3cm) {{\Large $\textsc{arc}_{\gamma}$}};
                \draw [thick, xshift=-1cm, yshift=-1cm] (0,0) -- (170: 2cm);
                \node[ xshift=-1cm, yshift=-1cm] at (156: 1.1cm) {{\Large $\textsc{arc}_{\alpha}$}};
                \draw [thick, xshift=-1cm, yshift=-1cm] (0,0) -- (290: 2cm);
                \node[ xshift=-1cm, yshift=-1cm] at (320: 1.8cm) {{\Large $\textsc{arc}_{\beta}$}};
                \node at (-1.1cm,-1.3cm) {{\Large $a$}};

                \draw[->, thick] (120 : 4cm) -- (120: 6cm);
                \node at (110: 4.5cm) {{\Large $\textsc{arc}_1$}};
        
                \draw[->,thick] (180: 4cm) -- (180: 6cm);
                \node at (175: 4.8cm) {{\Large $\textsc{arc}_2$}};
        
                \draw[->, thick] (-120: 4cm) -- (-120: 6cm);
                \node at (-110: 4.5cm) {{\Large $\textsc{arc}_3$}};
            \end{tikzpicture}}
            &\textcolor{white}{ fill}&
            \scalebox{.55}{\begin{tikzpicture}
                  
                \draw[->, thick] (120 : 4cm) -- (120: 6cm);
        
                \draw[->,thick] (180: 4cm) -- (180: 6cm);
                
                \draw[->, thick] (-120: 4cm) -- (-120: 6cm);

                \draw [thick, xshift=-1cm, yshift=-1cm]  (0,0) -- (170: 2cm);
                \draw [dashed,  gray] (120: 4cm) [xshift=-1cm, yshift=-1cm] -- (170: 2cm);
                
                \draw [thick, xshift=-1cm, yshift=-1cm] (0,0) -- (290: 2cm);
                \draw [dashed,  gray] (180: 4cm) [xshift=-1cm, yshift=-1cm].. controls (-150: 2.5cm) ..(290: 2cm);
        
                \draw [thick, xshift=-1cm, yshift=-1cm] (0,0) -- (50: 2cm);
                \draw [draw,  dashed, gray] (-120: 4cm) [xshift=-1cm, yshift=-1cm].. controls (-45: 4cm) ..(50: 2cm);
        
            \end{tikzpicture}}
            &\textcolor{white}{ fill}&
            \scalebox{.55}{\begin{tikzpicture}
    
                \draw[->, thick] (120 : 4cm) -- (120: 6cm);
        
                \draw[->,thick] (180: 4cm) -- (180: 6cm);
                
                \draw[->, thick] (-120: 4cm) -- (-120: 6cm);

                \draw [thick, xshift=-1cm, yshift=-1cm]  (0,0) -- (170: 2cm);
                \draw [dashed,  gray] (120: 4cm) [xshift=-1cm, yshift=-1cm] -- (170: 2cm);
                
                \draw [thick, xshift=-1cm, yshift=-1cm] (0,0) -- (290: 2cm);      
        
                \draw [thick, xshift=-1cm, yshift=-1cm] (0,0) -- (50: 2cm);
                \draw [dashed, gray] (180: 4cm) [xshift=-1cm, yshift=-1cm].. controls (-150: 6cm) and (-30: 8.5cm) ..(50: 2cm);

            \end{tikzpicture}}
    
        \end{tabular}
        \caption{Finding all of the possible configurations of the zero-level arcs of $\Real(2h+\lambda)$ in the left half of the $z$-plane, when $ x \in \Omega_{\text{left}}$. The left diagram shows the starting point, the middle diagram shows a valid configuration, while the right diagram shows an invalid configuration.}
        \label{connectingLines}
    \end{figure}

    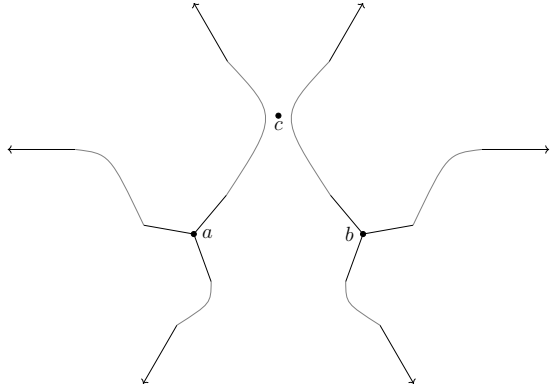
\begin{figure}
        \centering
        \scalebox{.45}{\begin{tikzpicture}
        
            \draw[fill=black] (-2.5,-2.5) circle (.075cm);
            \node at (-2.1cm,-2.5cm) {{\Large $a$}};
            \draw [thick, xshift=-2.5cm, yshift=-2.5cm] (0,0) -- (50: 1.5cm);
            \draw [thick, xshift=-2.5cm, yshift=-2.5cm] (0,0) -- (170: 1.5cm);
            \draw [thick, xshift=-2.5cm, yshift=-2.5cm] (0,0) -- (290: 1.5cm);

            \draw[fill=black] (2.5,-2.5) circle (.075cm);
            \node at (2.1cm,-2.5cm) {{\Large $b$}};
            \draw [thick, xshift=2.5cm, yshift=-2.5cm] (0,0) -- (130: 1.5cm);
            \draw [thick, xshift=2.5cm, yshift=-2.5cm] (0,0) -- (250: 1.5cm);
            \draw [thick, xshift=2.5cm, yshift=-2.5cm] (0,0) -- (10: 1.5cm);
    
            \draw[fill=black] (0,1) circle (.075cm);
            \node at (0,.7) {{\Large $c$}};

            \draw[->, thick] (120 : 3cm) -- (120: 5cm);
            
            \draw[thick,gray] (120: 3cm).. controls (0,1).. (-1.53582,-1.35093);
            
            \draw[->,thick] (180: 6cm) -- (180: 8cm);
    
            \draw[thick,gray] (180: 6cm).. controls (-4.98861, -1.11976+1) .. (-3.97721, -2.23953);
    
            \draw[->,thick] (240: 6cm) -- (240: 8cm);
    
            \draw[thick,gray] (240: 6cm).. controls (-2.49348+.5, -4.55285) .. (-1.98697, -3.90954);
    
            \draw[->,thick] (300: 6cm) -- (300: 8cm);
    
            \draw[thick,gray] (300: 6cm).. controls (2.49348-.5, -4.55285) .. (1.98697, -3.90954);
    
            \draw[->,thick] (0: 6cm) -- (0: 8cm);
    
            \draw[thick,gray] (6,0).. controls (4.98861, -1.11976+1) .. (3.97721, -2.23953);
    
            \draw[->, thick] (60: 3cm) -- (60: 5cm);
    
            \draw[thick,gray] (60: 3cm).. controls (0,1).. (1.53582, -1.35093);
        \end{tikzpicture}}
        \caption{Valid connections of the zero-level arcs of $\Real(2h+\lambda)$ for $x \in \Omega_{\text{left}}$.}
        \label{in Omega right}
    \end{figure} 

    We now show $\Omega_{\text{right}}$ is contained in the pole-free region. Fix $x_{r} \in \Omega_{\text{right}} \cap \R$. The same analysis gives us the same setup as in the previous case. That is, we know that the zero-level set of $\Real(2h+\lambda)$ will have exactly six unbounded arcs that tend to $\infty$ at angles $0$, $\pi/3$, $2\pi/3$, $\pi$, $-2\pi/3$, and $-\pi/3$, and  there are exactly three arcs of the zero-level set of $\Real(2h+\lambda)$ emanating from the points $z=a$ and $z=b$ at angles separated by $2\pi/3$. We now need to determine the valid connections between these arcs. Using the same the argument as in the previous case, we can use Figure \ref{signChartImaginaryAxishPrime} to find that the zero-level lines of $\Real(2h(z)+\lambda)$ cross the imaginary axis in the $z$-plane at two distinct points. 
    
    Notice, due to the symmetry condition \eqref{symOfhg0}, the topology of the zero-level lines differs from the previous case. Also, since $\mathfrak{c}(x)$ is continuous, \eqref{x-asymptotics} and the intermediate value theorem imply there exists $x_0 \in (x_{\ell}, x_{r})$ such that $\Real(2h(c;x_0)+\lambda(x_0))=0$. Here, we use the assumption that $\mathfrak{c}(x)$ only vanishes once on the real line, as indicated by the numerical plots. Indeed, plots suggest this map is monotone.
   
    To determine the valid connections of the zero-level arcs of $\Real(2h+\lambda)$ at $x_r$, we will use the fact that these arcs deform continuously as $x$ goes from $x_{\ell}$ to $x_r$ along the real line. As $x$ varies, we know that $\mathfrak{c}(x)$ vanishes at $x_0$. By the symmetry condition \eqref{symOfhg0} there are only two mechanisms for $\mathfrak{c}(x)$ to vanish on the real line. Namely, the arc that starts at $z=a$ that goes to $\infty$ at an angle of $2\pi/3$ can collide with the arc that starts at $z=b$ that goes to $\infty$ at an angle of $\pi/3$, or the arc that starts at $z=a$ that goes to $\infty$ at an angle of $-2\pi/3$ can collide with the arc that starts at $z=b$ that goes to $\infty$ at an angle of $-\pi/3$. Any other collision would cause $\mathfrak{c}(x)$ to vanish at least twice on $\R$. Further, since $\Imag(c)>0$ and the collision must occur at $z=c$, \eqref{Re(2h+lambda) on the real line} allows us to rule out the latter case. Therefore, the zero-level arcs of $\Real(2h+\lambda)$ when $x=x_0$ is topologically equivalent to Figure \ref{atTheBirfurPoint}.

    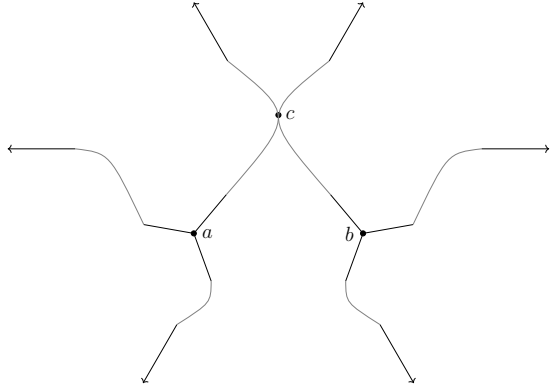
\begin{figure}
        \centering
        \scalebox{.45}{\begin{tikzpicture}    
            \draw[fill=black] (-2.5,-2.5) circle (.075cm);
            \node at (-2.1cm,-2.5cm) {{\Large $a$}};
            \draw [thick, xshift=-2.5cm, yshift=-2.5cm] (0,0) -- (50: 1.5cm);
            \draw [thick, xshift=-2.5cm, yshift=-2.5cm] (0,0) -- (170: 1.5cm);
            \draw [thick, xshift=-2.5cm, yshift=-2.5cm] (0,0) -- (290: 1.5cm);

            \draw[fill=black] (2.5,-2.5) circle (.075cm);
            \node at (2.1cm,-2.5cm) {{\Large $b$}};
            \draw [thick, xshift=2.5cm, yshift=-2.5cm] (0,0) -- (130: 1.5cm);
            \draw [thick, xshift=2.5cm, yshift=-2.5cm] (0,0) -- (250: 1.5cm);
            \draw [thick, xshift=2.5cm, yshift=-2.5cm] (0,0) -- (10: 1.5cm);

            \draw[fill=black] (0,1) circle (.075cm);
            \node at (.35,1) {{\Large $c$}};

            \draw[->, thick] (120 : 3cm) -- (120: 5cm);
        
            \draw[thick,gray] (120: 3cm).. controls (.5,1).. (-1.53582,-1.35093);
        
            \draw[->,thick] (180: 6cm) -- (180: 8cm);

            \draw[thick,gray] (180: 6cm).. controls (-4.98861, -1.11976+1) .. (-3.97721, -2.23953);

            \draw[->,thick] (240: 6cm) -- (240: 8cm);

            \draw[thick,gray] (240: 6cm).. controls (-2.49348+.5, -4.55285) .. (-1.98697, -3.90954);

            \draw[->,thick] (300: 6cm) -- (300: 8cm);

            \draw[thick,gray] (300: 6cm).. controls (2.49348-.5, -4.55285) .. (1.98697, -3.90954);

            \draw[->,thick] (0: 6cm) -- (0: 8cm);

            \draw[thick,gray] (6,0).. controls (4.98861, -1.11976+1) .. (3.97721, -2.23953);

            \draw[->, thick] (60: 3cm) -- (60: 5cm);

            \draw[thick,gray] (60: 3cm).. controls (-.5,1).. (1.53582, -1.35093);
        \end{tikzpicture}}
        \caption{The zero-level arcs of $\Real(2h+\lambda)$ when $x=x_0$.}
        \label{atTheBirfurPoint}
    \end{figure}

    When $x$ is pushed to the right of $x_0$ the arcs break apart so that there is a bounded arc that starts at $z=a$ and goes to $z=b$ and there is an unbounded arc that tends to $\infty$ in the left half of the $z$-plane at an angle of $2\pi/3$ and tends to $\infty$ in the right half at an angle of $\pi/3$. This breaking mechanism is forced because the zero-level set of $\Real(2h+\lambda)$ intersects the imaginary axis twice. So, we see that $x_r$ satisfies all of the axioms in the definition of the pole-free region; see Figure \ref{inOmegaRight}, and therefore  $\Omega_{\text{right}}$ is completely contained in the pole-free region.

    \begin{figure}[h]
        \centering
        \scalebox{.45}{\begin{tikzpicture}
        
            \draw[fill=black] (-2.5,-2.5) circle (.075cm);
            \node at (-2.1cm,-2.5cm) {{\Large $a$}};
            \draw [thick, xshift=-2.5cm, yshift=-2.5cm] (0,0) -- (50: 1.5cm);
            \draw [thick, xshift=-2.5cm, yshift=-2.5cm] (0,0) -- (170: 1.5cm);
            \draw [thick, xshift=-2.5cm, yshift=-2.5cm] (0,0) -- (290: 1.5cm);

            \draw[fill=black] (2.5,-2.5) circle (.075cm);
            \node at (2.1cm,-2.5cm) {{\Large $b$}};
            \draw [thick, xshift=2.5cm, yshift=-2.5cm] (0,0) -- (130: 1.5cm);
            \draw [thick, xshift=2.5cm, yshift=-2.5cm] (0,0) -- (250: 1.5cm);
            \draw [thick, xshift=2.5cm, yshift=-2.5cm] (0,0) -- (10: 1.5cm);
    
            \draw[fill=black] (0,1) circle (.075cm);
            \node at (.35,1) {{\Large $c$}};

            \draw[->, thick] (120 : 3cm) -- (120: 5cm);
            
            \draw[thick,gray] (120: 3cm).. controls (0,1).. (60: 3cm);;
            
            \draw[->,thick] (180: 6cm) -- (180: 8cm);
    
            \draw[thick,gray] (180: 6cm).. controls (-4.98861, -1.11976+1) .. (-3.97721, -2.23953);
    
            \draw[->,thick] (240: 6cm) -- (240: 8cm);
    
            \draw[thick,gray] (240: 6cm).. controls (-2.49348+.5, -4.55285) .. (-1.98697, -3.90954);
    
            \draw[->,thick] (300: 6cm) -- (300: 8cm);
    
            \draw[thick,gray] (300: 6cm).. controls (2.49348-.5, -4.55285) .. (1.98697, -3.90954);
    
            \draw[->,thick] (0: 6cm) -- (0: 8cm);
    
            \draw[thick,gray] (6,0).. controls (4.98861, -1.11976+1) .. (3.97721, -2.23953);
    
            \draw[->, thick] (60: 3cm) -- (60: 5cm);
    
            \draw[thick,gray]  (-1.53582,-1.35093).. controls (0,1).. (1.53582,-1.35093);
    
        \end{tikzpicture}}
        \caption{The zero-level arcs of $\Real(2h+\lambda)$ when $x\in \Omega_{\text{right}}$.}
        \label{inOmegaRight}
    \end{figure}
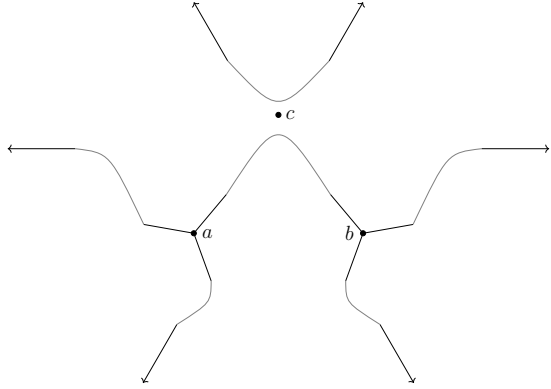
\end{proof}
\begin{remark}\label{sigmaPlacementg0}
    The zero-level set of $\Real(2h+\lambda)$ is independent of the placement of the band $\Sigma$. That is, the zero-level set of $\Real(2h+\lambda)$ is invariant with respect to the choice of $\Sigma$ as long as $\Sigma$ is a finite, simple curve that connects $a$ to $b$. Therefore, we are at liberty to pick $\Sigma$. For the ease of describing other objects in the analysis we choose $\Sigma$ such that
\begin{itemize}
    \item When $x \in \Omega_{\text{left}}$, we require that (i)  $\Sigma \cap \D_a$ is contained  in the zero-level arc of $\Real(2h+\lambda)$ that emanates from $z=a$ and goes to $\infty$ at an angle of $2\pi/3$ and (ii) $\Sigma \cap \D_b$ is contained in the zero-level arc of $\Real(2h+\lambda)$ that emanates from $z=b$ and goes to $\infty$ at an angle of $\pi/3$. 
    \item When $x \in \Omega_{\text{right}}$, $\Sigma$ coincides with the zero-level arc of $\Real(2h+\lambda)$ that connects $z=a$ to $z=b$.
\end{itemize}
\end{remark}
\begin{remark}
    The section of the boundary that borders $\Omega_{\text{left}}$ is the two straight rays that comprise $\Sigma_S$ (the branch cut of $S(x)$). The proof is straightforward and involves evaluating \eqref{hofzExact} on the boundary.
\end{remark}

\subsection[short]{Surveying the Pole-Free Region}
We now survey the pole-free region. See Table \ref{TableOfxValues}  for a list of all of the $x$-values we use to graph the zero-level lines of $\Real(2h+\lambda)$. 

\begin{table}[h]
    \centering
    \begin{tabular}{ccccc}
        $9e^{2\pi i /3} \approx -4.5 + 7.794i$ & & & & $ \approx 1.5+ 6.678i$ \\
        $ 3e^{2\pi i /3}-3 \approx -4.5 + 2.598i$ && $3e^{2\pi i /3} \approx -1.5 + 2.598$ && $3e^{2\pi i /3}+3 \approx 1.5 + 2.598i$ \\
        $ -4.5$ && $x_0 \approx -1.588$ && $1.5$ \\
        $3e^{-2\pi i /3}-3 \approx -4.5 - 2.598i $  && $3e^{-2\pi i /3} \approx -1.5 - 2.598$ && $3e^{-2\pi i /3}+3 \approx 1.5 - 2.598i$\\
        $9e^{-2\pi i /3} \approx -4.5 - 7.794i$ & & & & $\approx 1.5 - 6.678i$ \\
    \end{tabular}
    \caption{$x$-values corresponding to Figure \ref{zeroLevelLinesGenusZeroAll}}
    \label{TableOfxValues}
\end{table}

\begin{figure}
    \centering
    \begin{tabular}[h]{ccc}
        \scalebox{.3}{\includegraphics{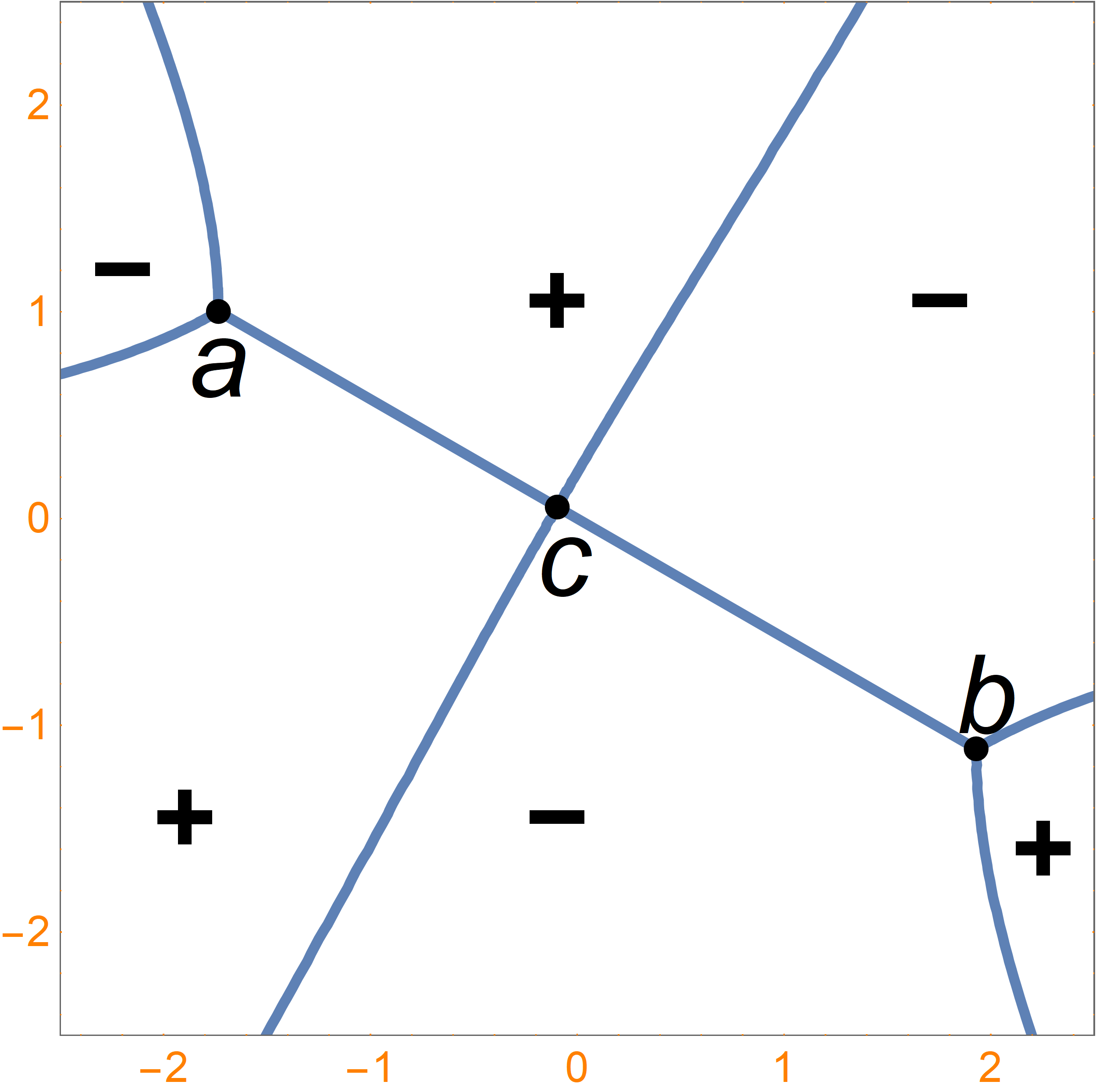}} & \scalebox{.3}{\includegraphics{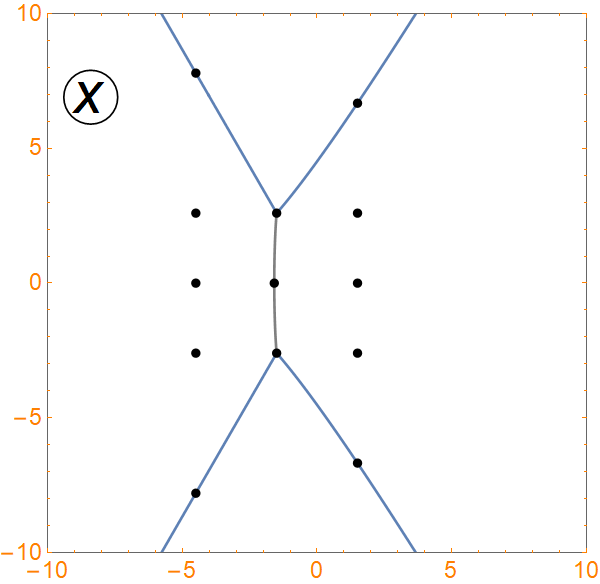}} & \scalebox{.3}{\includegraphics{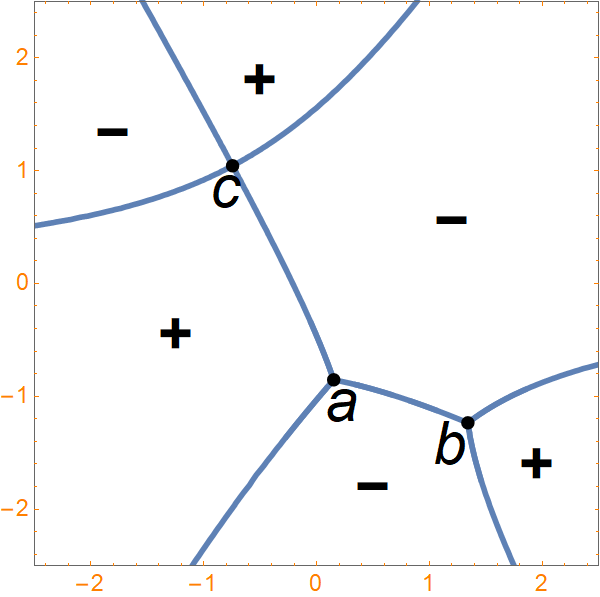}} \\
        \scalebox{.3}{\includegraphics{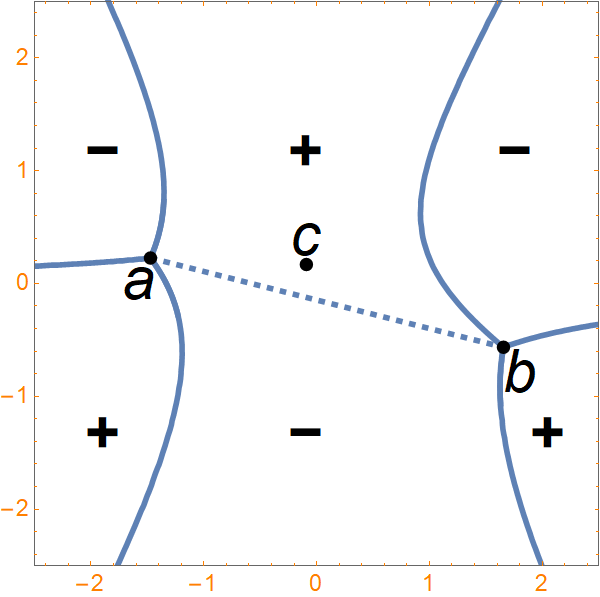}} & \scalebox{.3}{\includegraphics{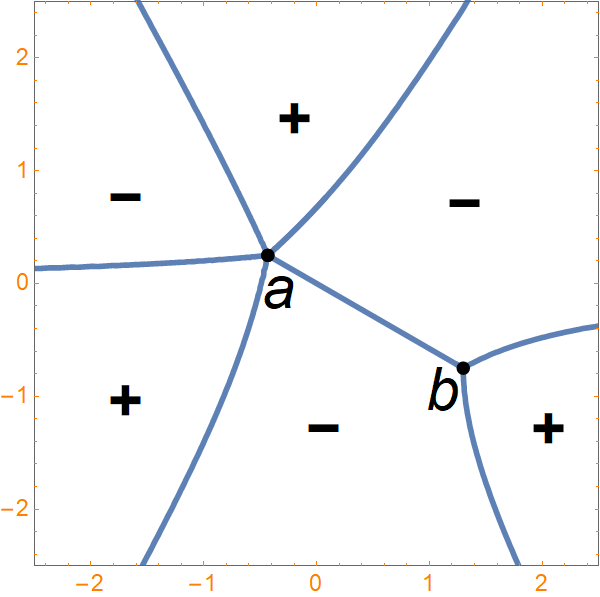}} & \scalebox{.3}{\includegraphics{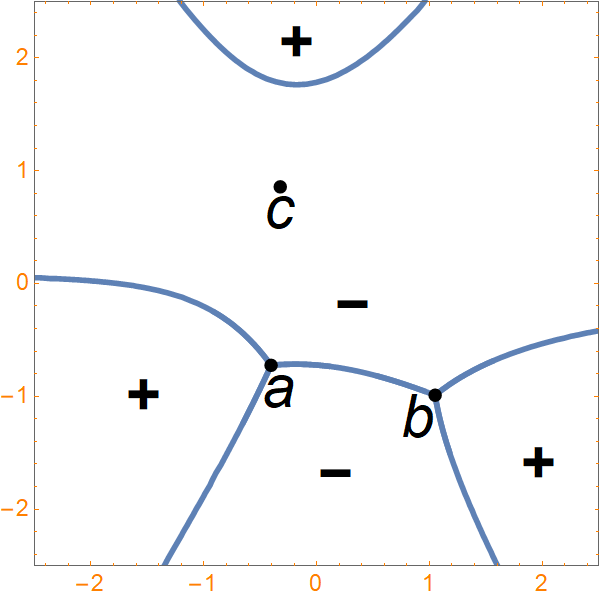}} \\
        \scalebox{.3}{\includegraphics{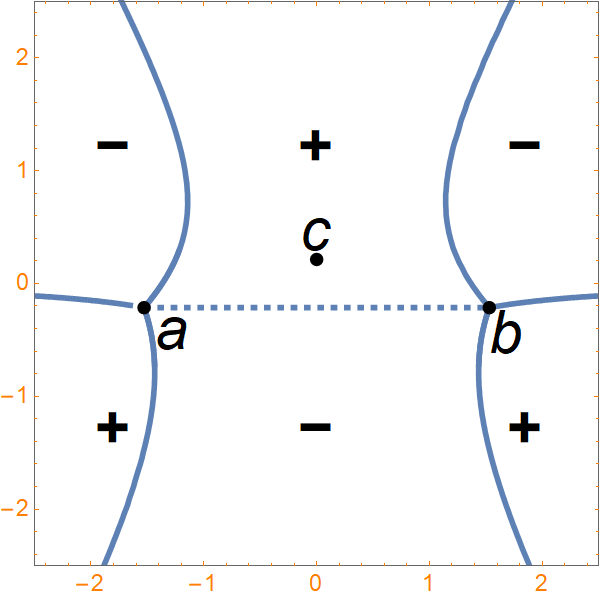}} & \scalebox{.3}{\includegraphics{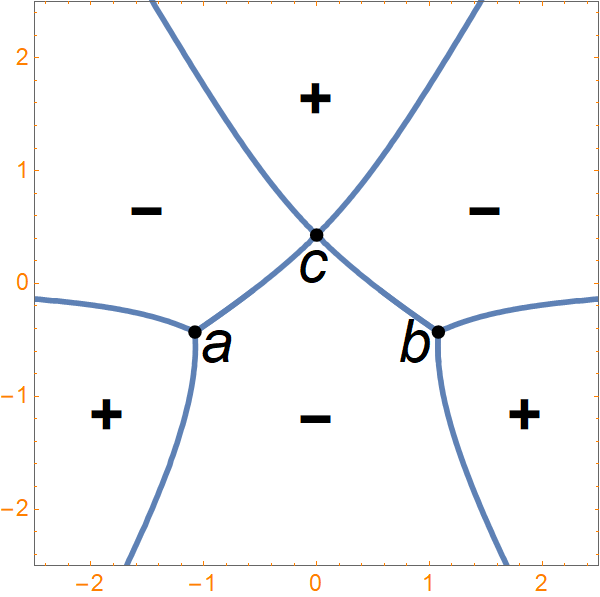}} & \scalebox{.3}{\includegraphics{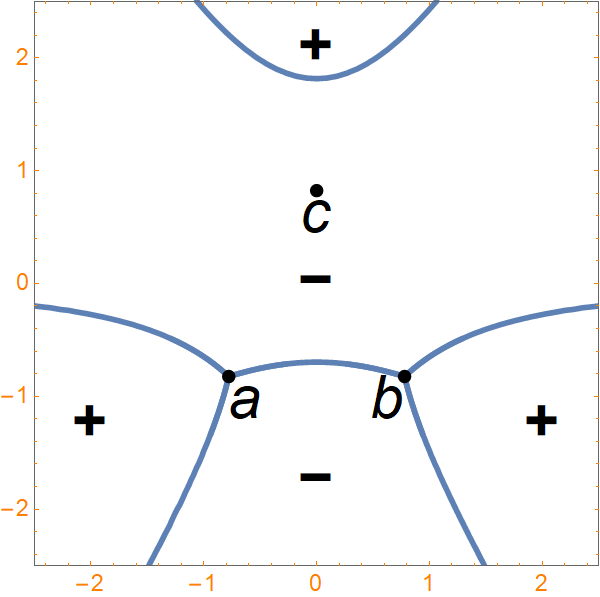}} \\
        \scalebox{.3}{\includegraphics{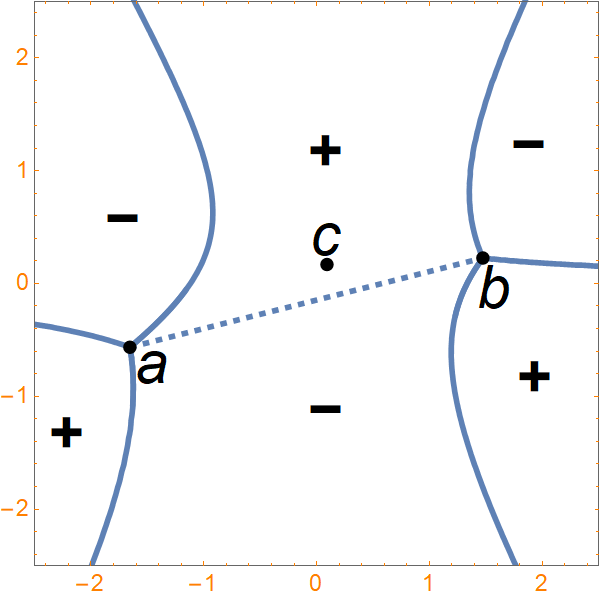}} & \scalebox{.3}{\includegraphics{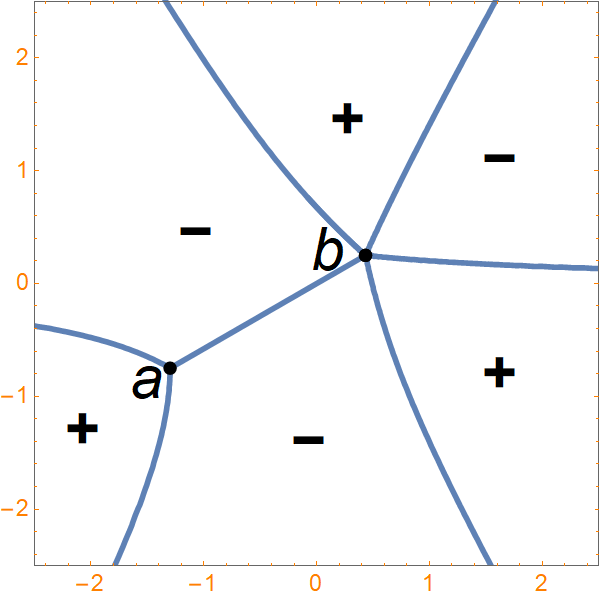}} & \scalebox{.3}{\includegraphics{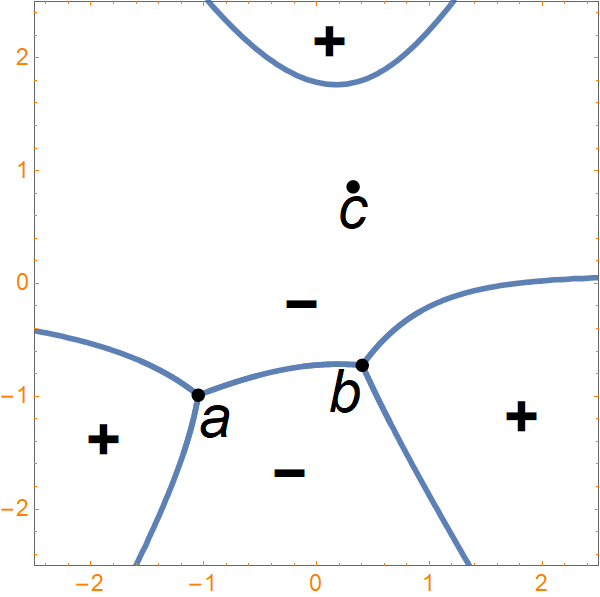}} \\
        \scalebox{.3}{\includegraphics{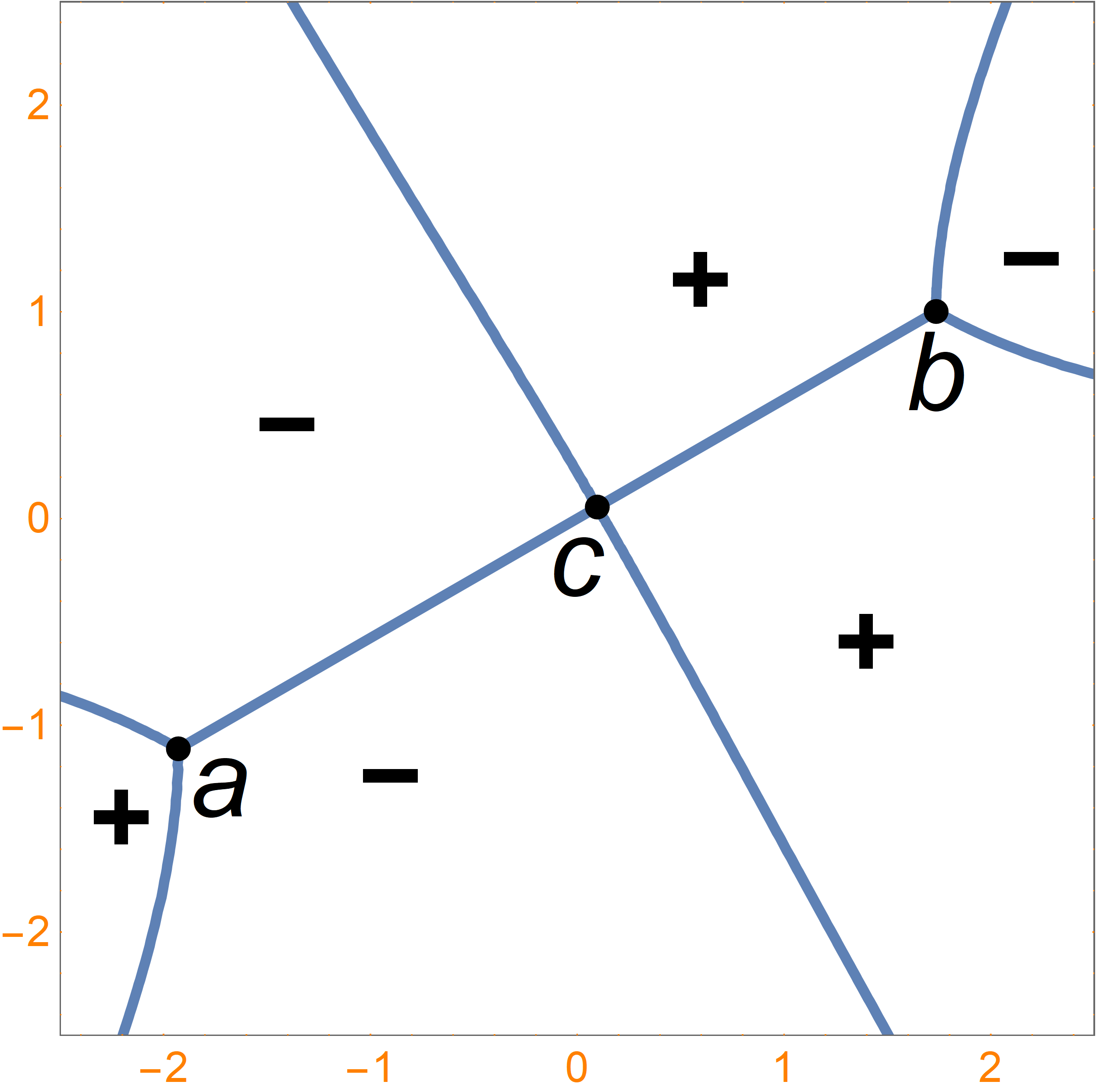}} && \scalebox{.3}{\includegraphics{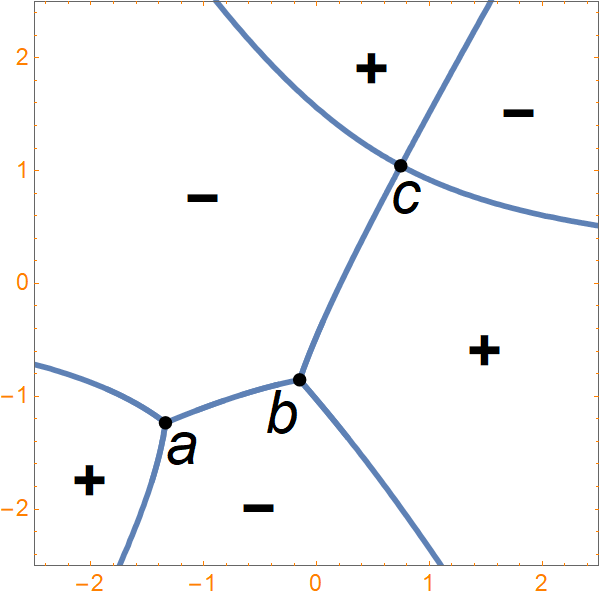}}
    \end{tabular}
    \caption{The zero-level arcs of $\Real(2h + \lambda)$ in the $z$-plane for various  $x$-values as shown in Table \ref{TableOfxValues}. Here, the dotted lines refers to branch cuts that do not coincide with the zero-level set.}
    \label{zeroLevelLinesGenusZeroAll}
\end{figure} 

Notice the breaking mechanism differs between the left and right side of the $x$-plane. On one hand, when $x$ leaves the pole-free region to the left of the vertical line $\Real(x)=-1.5$, two distinct zero-level arcs collide with each other at $z=c$. In this case, $\Real(c)$ will be between $\Real(a)$ and $\Real(b)$. The resulting configuration has a zero-level arc that separates $a$ and $b$. Thus, such a configuration violates $(\textbf{a5})$. In Figure \ref{zeroLevelLinesGenusZeroAll} this corresponds to when $x=9e^{2\pi i /3}$ or $x = 9e^{-2\pi /3}$. Notice that we have an exact formula for when this phenomenon occurs, $x = \rho e^{\pm 2\pi i /3}$ with $\rho >3$. On the other hand, when $x$ leaves the pole-free region to the right of $\Real(x)=-1.5$, two distinct zero-level arcs collide with each other at $z=c$, but $\Real(c)$ will not be between $\Real(a)$ and $\Real(b)$. The resulting configuration violates $(\textbf{a2})$ or $(\textbf{a4})$ depending on if $x$ is in the upper or lower half of the plane, respectively. The right boundary curves are defined implicitly by those $x$ satisfying $\Real\left( 2h(c(x);x)+\lambda(x) \right)=0$.

There are two special $x$-values, $x=3e^{\pm 2\pi i /3}$, when an endpoint collides with $z=c$. These are the apexes of the boundary. When $x= 3 e^{2\pi i /3}$, then $a=c$ or, in other words, $3a+b=0$. In fact, $2h(a)+\lambda$ has a zero of order $5/2$. Hence, \textbf{(a1)} is not satisfied. Similarly, when $x=3e^{-2\pi i /3}$, \textbf{(a3)} is not satisfied. Finally, we can see that a topological change is not a sufficient condition to be on the boundary. Indeed, at $x_0$  we see that the zero-level lines undergo a sudden topological change but of all the axioms are still satisfied. Notice, this section of the zero-level set of $\mathfrak{c}(x)$ is curved and is not contained in the vertical line $\Real(x)=-3/2$.

\subsection{Preliminary Deformations in the Riemann-Hilbert Analysis} 

For the remainder of Section \ref{GenusZeroAnalysis} we fix $x$ in the pole-free region as we perform the Deift-Zhou nonlinear steepest-descent method. The general line of attack is to apply a sequence of deformations to $\boldsymbol{M}(z)$ to obtain a corresponding Riemann-Hilbert problem that can be approximated by a solvable Riemann-Hilbert problem. Our first step is to collapse the corresponding jump contours of $\boldsymbol{M}(z)$ in an open bounded set. In order to specify the exact placement of the deformed jump contours, we need to construct two conformal maps $\tau_a$ and $\tau_b$. 

\subsubsection[]{Conformal Map Constructions}\label{conformalMapSection}
We start with the simpler case and construct $\tau_b$. Recall, $\D_b$ is a small disk around $z=b$. Notice, one can write the jump condition \eqref{jumpsofhg0} as
\begin{equation}\label{2h+lambda sqrt jump}
    2h_{+}(z) + \lambda = -(2h_{-}(z) +\lambda), \hspace{.33 cm} z \in \Sigma.
\end{equation}
So, \eqref{nearEndpointsG0} and \eqref{2h+lambda sqrt jump} imply that in $\D_b$, the map $z \mapsto (2h(z)+\lambda)^2$ is analytic and nonvanishing, except for a third-order zero at $z=b$. Also, since $\Sigma$ was placed on the zero-level set of $\Real(2h+\lambda)$, we have that for each $z \in \Sigma \cap \D_b$, $(2h_{\pm}(z)+\lambda)^2 <0$. With these considerations in mind, we see that there exists a conformal map $\tau_b$ such that
\begin{itemize}
    \item For each $z \in \D_b$, $\tau_b(z)^3 = (2h(z)+\lambda)^2$,
    \item $\tau_b$ maps $\D_b$ to a small neighborhood of the origin,
    \item $\tau_b$ maps $\D_b \cap \Sigma$ to the negative real line.
\end{itemize}

We need to construct a similar conformal mapping of $\D_a$. However, in this case we need to deal with the logarithmic cut $L$. Notice, one can write \eqref{logJump} in terms of $h$ as
\begin{equation}\label{helpfulLogJump}
    2h_{+}(z)-2\pi i = 2h_{-}(z)+2\pi i, \hspace{.33 cm} z \in L.
\end{equation}
Consider the map 
\begin{equation}\label{aux}
    z \mapsto \begin{cases}
        \left( 2h(z)+\lambda - 2\pi i \right)^{2} & \text{if }z \in \D_a^{\uparrow}, \\
        \left( 2h(z)+\lambda + 2\pi i \right)^{2} & \text{if } z \in \D_a^{\downarrow}.
    \end{cases}
\end{equation}
Notice, \eqref{2h+lambda sqrt jump} and \eqref{helpfulLogJump} imply that both the upper and lower functions of the above map have the same continuous extension on $(L \cup \Sigma) \cap \D_a$. Therefore, the map \eqref{aux} is analytic in $\D_a$. Moreover, \eqref{nearEndpointsG0} shows that this map is nonvanishing in $\D_a$ except for a third-order zero at $z=a$. Thus, there exists a conformal map $\tau_a$ that maps $\D_a$ to a neighborhood of the origin, that satisfies
\begin{equation}
    (\tau_a(z))^3 = \begin{cases}
        \left( 2h(z)+\lambda - 2\pi i \right)^{2} & \text{if }z \in \D_a^{\uparrow}, \\
        \left( 2h(z)+\lambda + 2\pi i \right)^{2} & \text{if } z \in \D_a^{\downarrow},
    \end{cases}
\end{equation}
and that maps $\Sigma \cap \D_a$ to the negative real numbers.

\subsubsection{The Collapsing Deformation}
We deform jump contours of Riemann-Hilbert problems by the standard sectionally analytic substitution method (for more details see Section 4.3 in \cite{movingContours}). Our first deformation is to collapse the four jump rays of $\boldsymbol{M}(z)$ onto the line segment with endpoints $z=a$ and $z=b$. The resulting jump contour will be comprised of four unbounded rays and one line segment from $z=a$ to $z=b$. Two of the unbounded jump rays emanate from $z=a$ 
and approach $\infty$ at angles $\pm 2\pi  /3$, while the other  two rays emanate  from $z=b$ and approach $\infty$ at angles $\pm \pi  /3$. The jump along the deformed rays will be the same jump along the parent ray, while the new jump along the line segment is simply the product 
\begin{equation}
    \begin{bmatrix}
        1 & -e^{-ik\theta} \\
        0 & 1
    \end{bmatrix}
    \begin{bmatrix}
        1 & 0 \\
        e^{ik\theta} & 1
    \end{bmatrix}
    =
    \begin{bmatrix}
        0 & -e^{-ik\theta} \\
        e^{ik\theta} & 1
    \end{bmatrix}.
\end{equation}
Since the deformation occurred on a bounded set, the new resulting Riemann-Hilbert problem has the same normalization condition as Riemann-Hilbert \ref{RHPMScaled}. In fact, the new Riemann-Hilbert problem is:

\begin{RHP}\label{RHPN}
    For each $k \in \N$, find $\boldsymbol{N}(z;x,k)\equiv\boldsymbol{N}(z)$ so that 
    \begin{enumerate} 
        \item \textbf{Analyticity}: $\boldsymbol{N}(z)$ is analytic in $z$ except along the jump contours shown in Figure \ref{jumpsOfNGenusZero}.
    
        \item \textbf{Jump Condition:} $\boldsymbol{N}(z)$ can be continuously extended to the boundary and the boundary values taken by $\boldsymbol{N}(z)$ are related by the jump condition $\boldsymbol{N}_+(z) = \boldsymbol{N}_-(z)\boldsymbol{V}^{(\boldsymbol{N})}(z)$ as shown in Figure \ref{jumpsOfNGenusZero}.
        \begin{figure}[h]
            \centering
        \scalebox{.85}{
            \begin{tikzpicture}[]
            
                \draw[thick] (2.25,0) -- (30:  3.75cm);
                \draw[thick] (30: 3.75cm) -- (30: 7cm);
                \draw[very thick,->] (30: 3.9cm) -- (30: 3.92);
                \node(jumpq1) at (42: 4.7cm) {$\boldsymbol{L}_k(i\theta)$};

                \draw[thick] (-2.25,0) -- (150: 3.75cm);
                \draw[thick] (150: 3.75cm) -- (150: 7cm);
                \draw[very thick,->] (150: 3.9cm) -- (150: 3.92cm);
                \node(jumpq2) at (138: 4.7cm) {$\boldsymbol{L}_k^{-1}(i\theta)$};
    
                \draw[thick] (-2.25,0) -- (210: 3.75cm);
                \draw[thick] (210: 3.75cm)--(210: 7cm);
                \draw[very thick, ->] (210: 3.9cm) -- (210: 3.92cm);
                \node(jumpq3) at (-138: 4.7cm) {$\boldsymbol{U}_{k}(-i\theta)$};
                
                \draw[thick] (2.25,0) -- (330: 3.75cm);
                \draw[thick] (330: 3.75cm) -- (330: 7cm);
                \draw[very thick, ->] (330: 3.9cm) -- (330: 3.92cm);
                \node(jumpq4) at (-43: 4.7cm) {$\boldsymbol{U}_{k}^{-1}(-i\theta)$};
    
                \draw[thick] (-2.25,0) -- (2.25,0);
                \draw[very thick, ->] (0,0) -- (.01,0);
                \node(jumpatBand) at (0,1) {$\begin{bmatrix}
                    0 & -e^{-ik\theta} \\
                    e^{ik\theta} & 1
                \end{bmatrix}$};
    
                 \draw[fill=black] (-2.25,0) circle (.075cm);
                 \node(aLable) at (-2.5,0) {{\Large $a$}};
                 \draw[fill=black] (2.25,0) circle (.075cm);
                 \node(bLable) at (2.5,0) {{\Large$b$}};    
    
            \end{tikzpicture}
        }
        \caption{The jumps of $\boldsymbol{N}(z)$.}
        \label{jumpsOfNGenusZero}
        \end{figure}
    
        \item \textbf{Normalization:} As $z \to \infty$, the matrix $\boldsymbol{N}(z)z^{-k\sigma_3} = \I + \Oh\left(1/z\right)$.
    
    \end{enumerate}
\end{RHP}

Although there is some degree of freedom where we place the deformed jump rays, we require the following criteria:
\begin{itemize}
    \item In $\Ci \setminus (\D_a \cup \D_b)$,  the four unbounded jump contours avoid the zero-level set of $\Real(2h+\lambda)$. Such a deformation is guaranteed to exist by Axioms (\textbf{a2}) and (\textbf{a4}),
    \item In $\D_a$ the jump contour that has an upper-triangular jump is contained in $\tau_a^{-1}\left( \left( 0, \infty\right) \right)$,
    \item In $\D_a$ the jump contour that has a lower-triangular jump is contained in $\tau_a^{-1}\left(\rho e^{-2i/3}\right)$, $\rho >0$,
    \item In $\D_b$ the jump contour that has an upper-triangular jump is contained $\tau_b^{-1}\left( \left( 0, \infty\right) \right)$, 
    \item In $\D_b$ the jump contour that has a lower-triangular jump is contained in $\tau_b^{-1}\left(\rho e^{2i/3}  \right)$,  $\rho >0$.
\end{itemize}

\subsection{Opening of the Lens} The jump on the  line segment with endpoints  $z=a$ and $z=b$ has the following factorization:
\begin{equation} \label{lenFactorizationab}
    \begin{bmatrix}
        0 & -e^{-ik\theta} \\
        e^{ik\theta} & 1
    \end{bmatrix} = 
    \begin{bmatrix}
        1 & 0 \\
        -e^{ik\theta} & 1
    \end{bmatrix}
    \begin{bmatrix}
        0 & -e^{-ik\theta} \\
        e^{ik\theta} & 0
    \end{bmatrix}.
\end{equation} 
As such, we use sectionally analytic substitutions to open a lens. The upper boundary of the lens, i.e.\ the jump contour with an off-diagonal jump, is placed on $\Sigma$. For the lower boundary of the lens, near the endpoints $z=a$ and $z=b$ we have the following requirements:
\begin{itemize}
    \item In $\D_a$, the lower boundary of the lens, i.e.\ the jump contour with a lower-triangular jump, is contained in $\tau_a^{-1}\left( \rho e^{2\pi i /3} \right)$ for $\rho>0$,
    \item In $\D_b$, the lower boundary of the lens is contained in $\tau_b^{-1}\left( \rho e^{-2\pi i /3} \right)$ for $\rho>0$.
\end{itemize}
Outside of $\D_a \cup \D_b$ we only require the lower  boundary of the lens does not intersect the zero-level set of $\Real\left( 2h+\lambda \right)$. The existence of such a lens is guaranteed by Axiom \textbf{(a5)}. The opening lens deformation results in another matrix-valued function $\boldsymbol{O}(z)$. Since this deformation is on a bounded set, $\boldsymbol{O}(z)$ and $\boldsymbol{N}(z)$ satisfy the same normalization condition. Thus, $\boldsymbol{O}(z)$ solves the following Riemann-Hilbert problem:

\begin{RHP}
    For each $k \in \N$, find $\boldsymbol{O}(z;x,k)\equiv \boldsymbol{O}(z)$ so that 
    \begin{enumerate}
        \item \textbf{Analyticity}:
            $\boldsymbol{O}(z)$ is analytic in $z$ except along the jump contours shown in Figure  \ref{jumpsOfOGenusZero}.
        \item \textbf{Jump Condition:}
            $\boldsymbol{O}(z)$ can be continuously extended to the boundary and the boundary values taken by $\boldsymbol{O}(z)$ are related by the jump condition $\boldsymbol{O}_+(z)=\boldsymbol{O}_-(z)\boldsymbol{V}^{(\boldsymbol{O})}(z)$ as shown in Figure \ref{jumpsOfOGenusZero}.
            
            \begin{figure}[h]
                \centering
                \scalebox{1}{
                    \begin{tikzpicture}
            
                        \draw[thick] (2.25,0) -- (30:  3.75cm);
                        \draw[thick] (30: 3.75cm) -- (30: 7cm);
                        \draw[very thick,->] (30: 3.9cm) -- (30: 3.92);
                        \node(jumpq1) at (42: 4.7cm) {$\boldsymbol{L}_k(i\theta)$};

                        \draw[thick] (-2.25,0) -- (150: 3.75cm);
                        \draw[thick] (150: 3.75cm) -- (150: 7cm);
                        \draw[very thick,->] (150: 3.9cm) -- (150: 3.92cm);
                        \node(jumpq2) at (138: 4.7cm) {$\boldsymbol{L}_k^{-1}(i\theta)$};
            
                        \draw[thick] (-2.25,0) -- (210: 3.75cm);
                        \draw[thick] (210: 3.75cm)--(210: 7cm);
                        \draw[very thick, ->] (210: 3.9cm) -- (210: 3.92cm);
                        \node(jumpq3) at (-138: 4.7cm) {$\boldsymbol{U}_{k}(-i\theta)$};
                        
                        \draw[thick] (2.25,0) -- (330: 3.75cm);
                        \draw[thick] (330: 3.75cm) -- (330: 7cm);
                        \draw[very thick, ->] (330: 3.9cm) -- (330: 3.92cm);
                        \node(jumpq4) at (-43: 4.7cm) {$\boldsymbol{U}_{k}^{-1}(-i\theta)$};
            
                        \draw[thick] (-2.25, 0).. controls (-1.55,1.25) and (1.55,1.25) .. (2.25,0);
                        \draw[very thick, ->] (0,.95) -- (.001,.95);            
                        \node(offDia) at (90: 1.75cm) {$\begin{bmatrix}
                            0 & -e^{-ik\theta} \\
                            e^{ik\theta} & 0
                        \end{bmatrix}$};
            
                        \draw[thick] (-2.25,0).. controls (-.1,-.4) and (.1,-.4) .. (2.25,0);
                        \draw[very thick,->] (0,-.3) -- (.001, -.3);
                        \node(lowerLenJump) at (-90: .7cm) {$\boldsymbol{L}_{k}^{-1}(i\theta)$};
            
                         \draw[fill=black] (-2.25,0) circle (.075cm);
                         \node(aLable) at (-2.5,0) {$a$};
                         \draw[fill=black] (2.25,0) circle (.075cm);
                         \node(bLable) at (2.5,0) {$b$};    
            
                    \end{tikzpicture}
            }
            \caption{The jumps of $\boldsymbol{O}(z)$.}
            \label{jumpsOfOGenusZero}
            \end{figure}
        \item \textbf{Normalization:} As $z \to \infty$, the matrix $\boldsymbol{O}(z)z^{-k\sigma_3}= \I + \Oh\left(1/z\right)$.    
    \end{enumerate}
\end{RHP}

\subsection{Applying the g-function}

We now apply the $g$-function defined in Section \ref{the g-function section} (see \eqref{gofzFormulag0}) to the function $\boldsymbol{O}(z)$. We consider
\begin{equation}\label{defOfPkgenusZero}
    \boldsymbol{P}(z;x,k) \equiv \boldsymbol{P}(z)   :=
        e^{-\frac{k\lambda}{2}\sigma_3} 
        \boldsymbol{O}(z)
        e^{-k g(z)\sigma_3}
        e^{\frac{k \lambda}{2}\sigma_3},
\end{equation}
where $\lambda$ is defined in \eqref{defOfLambda}. Introducing the $g$-function simplifies both the jump and normalization conditions of the corresponding Riemann-Hilbert problem. Since $e^{g(z)}$ is analytic everywhere except on $\Sigma$, the deformation $\boldsymbol{O}(z) \to \boldsymbol{P}(z)$ does not introduce new jump contours or move any. It does, however, simplify the jump matrices. A formula for jump matrices of $\boldsymbol{P}(z)$, given the jump matrices of $\boldsymbol{O}(z)$, is
    \begin{equation}\label{jumpRecipeForPG0}
        \boldsymbol{V}^{(\boldsymbol{O})} =
        \begin{bmatrix}
            v_{11}^{(\boldsymbol{O})} & v_{12}^{(\boldsymbol{O})} \\
            v_{21}^{(\boldsymbol{O})} & v_{22}^{(\boldsymbol{O})}
        \end{bmatrix} 
        \Longrightarrow \boldsymbol{V}^{(\boldsymbol{P})} =
        \begin{bmatrix}
            v_{11}^{(\boldsymbol{O})} e^{-k(g_{+}-g_{-})} & 
            v_{12}^{(\boldsymbol{O})} e^{k(g_{+}+g_{-} -\lambda)} \\
            v_{21}^{(\boldsymbol{O})}e^{-k(g_{+}+g_{-}- \lambda)} & 
            v_{22}^{(\boldsymbol{O})} e^{k(g_{+}-g_{-})}
        \end{bmatrix}.
    \end{equation}
On one hand, $\boldsymbol{V}^{(\boldsymbol{O})}$ is highly oscillatory along $\Sigma$. On the other hand, due to the jump of $g$ on $\Sigma$ (see \eqref{jumpsofhg0}) we see that $\boldsymbol{V}^{(\boldsymbol{P})}$ is a constant off-diagonal matrix. For all of the other jump matrices, their corresponding jump contour is placed so that $\boldsymbol{V}^{(\boldsymbol{P})}$ decays to the identity matrix as $k \to \infty$. Indeed, if  $\boldsymbol{V}^{(\boldsymbol{P})}$ is an upper-triangular matrix, then the exponent of its $12$-entry is $-(2h+\lambda)$ and its jump contour is completely contained in a region of the $z$-plane where $\Real(2h+\lambda)>0$. Likewise, if  $\boldsymbol{V}^{(\boldsymbol{P})}$ is a lower-triangular matrix, then the exponent of its $21$-entry is $2h+\lambda$ and  its jump contour is completely contained in a region of the $z$-plane where $\Real(2h+\lambda)<0$ (see Figure \eqref{zeroLevelLinesGenusZeroAll} and \eqref{genusZeroJumpsP}). Finally, the normalization \eqref{gfunNormg0} of $g(z)$ simplifies the large-$z$ behavior of $\boldsymbol{P}(z)$. Therefore, $\boldsymbol{P}(z)$ is the solution to the following Riemann-Hilbert problem:
\begin{RHP}\label{RHPPGenus0}
    For each positive integer $k$, find $\boldsymbol{P}\left( z;x,k \right) \equiv \boldsymbol{P}(z)$ so that
    \begin{enumerate}
        \item \textbf{Analyticity}: $\boldsymbol{P}(z)$ is analytic in $z$ off the jumps contours shown in Figure \ref{genusZeroJumpsP}.
        \item \textbf{Jump Condition:} $\boldsymbol{P}(z)$ can be continuously extended to the boundary and the boundary values taken by $\boldsymbol{P}(z)$ are related by the jump condition $\boldsymbol{P}_+(z)=\boldsymbol{P}_-(z)\boldsymbol{V}^{(\boldsymbol{P})}(z)$, where $\boldsymbol{V}^{(\boldsymbol{P})}(z)$ is as shown in Figure \ref{genusZeroJumpsP}.
        \begin{figure}[h]
            \centering
            \scalebox{.875}{
                \begin{tikzpicture}[]
        
                    \draw[thick] (2.25,0) -- (30:  3.75cm);
                    \draw[thick] (30: 3.75cm) -- (30: 7cm);
                    \draw[very thick,->] (30: 3.9cm) -- (30: 3.92);
                    \node(jumpq1) at (42: 4.7cm) {$\boldsymbol{L}_k(2h+\lambda)$};

                    \draw[thick] (-2.25,0) -- (150: 3.75cm);
                    \draw[thick] (150: 3.75cm) -- (150: 7cm);
                    \draw[very thick,->] (150: 3.9cm) -- (150: 3.92cm);
                    \node(jumpq2) at (138: 4.7cm) {$\boldsymbol{L}_k^{-1}(2h+\lambda)$};
        
                    \draw[thick] (-2.25,0) -- (210: 3.75cm);
                    \draw[thick] (210: 3.75cm)--(210: 7cm);
                    \draw[very thick, ->] (210: 3.9cm) -- (210: 3.92cm);
                    \node(jumpq3) at (-138: 4.7cm) {$\boldsymbol{U}_{k}(-(2h+\lambda))$};
                    
                    \draw[thick] (2.25,0) -- (330: 3.75cm);
                    \draw[thick] (330: 3.75cm) -- (330: 7cm);
                    \draw[very thick, ->] (330: 3.9cm) -- (330: 3.92cm);
                    \node(jumpq4) at (-43: 4.7cm) {$\boldsymbol{U}_{k}^{-1}(-(2h+\lambda))$};
        
                    \draw[thick] (-2.25, 0).. controls (-1.55,1.25) and (1.55,1.25) .. (2.25,0);
                    \draw[very thick, ->] (0,.95) -- (.001,.95);            
                    \node(offDia) at (90: 1.75cm) {$\begin{bmatrix}
                        0 & -1 \\
                        1 & 0
                    \end{bmatrix}$};
        
                    \draw[thick] (-2.25,0).. controls (-.1,-.4) and (.1,-.4) .. (2.25,0);
                    \draw[very thick,->] (0,-.3) -- (.001, -.3);
                    \node(lowerLenJump) at (-90: .7cm) {$\boldsymbol{L}_{k}^{-1}(2h+\lambda)$};
        
                     \draw[fill=black] (-2.25,0) circle (.075cm);
                     \node(aLable) at (-2.5,0) {$a$};
                     \draw[fill=black] (2.25,0) circle (.075cm);
                     \node(bLable) at (2.5,0) {$b$};    
        
                \end{tikzpicture} }
            \caption{Jumps for $\boldsymbol{P}(z)$.}
            \label{genusZeroJumpsP}
        \end{figure}
        \item \textbf{Normalization:} As $z \to \infty$, the matrix  $\boldsymbol{P}(z) = \I + \Oh\left(1/z\right)$.
    \end{enumerate}
\end{RHP}
\subsection{The Outer Parametrix}

The phase function was constructed so that if $k$ is large, all of the the jump matrices (except for the one along $\Sigma$) of $\boldsymbol{P}(z)$ are close to the identity matrix. Ignoring these jumps we are led to the following Riemann-Hilbert problem:

\begin{RHP}\label{genusZeroModelRHP}
    Determine the matrix-valued function $\dotbold{P}^\out(z;x)\equiv \dotbold{P}^{\out}(z)$ with the following properties:
    \begin{enumerate}
        \item\textbf{Analyticity:} $\dotbold{P}^\out(z)$ is analytic in $z$ except on $\Sigma$. At the endpoints $a$ and $b$, negative one-fourth power singularities are allowed.
        \item\textbf{Jump Condition:} $\dotbold{P}^{\out}(z)$ can be continuously extended to the interior of $\Sigma$ and the boundary values taken by $\dotbold{P}^\out(z)$ are related by the jump condition 
        \begin{equation}\label{OdotJump}
            \dotbold{P}^\out_+(z) = \dotbold{P}^\out_-(z) \begin{bmatrix}
                    0 & -1 \\
                    1 & 0 
            \end{bmatrix}.
        \end{equation}
        \item \textbf{Normalization:} As $z \to \infty$, the matrix $\dotbold{P}^\out(z) = \I + \Oh\left(1/z\right)$.
    \end{enumerate}.    
\end{RHP}
It is well known that the unique solution to this Riemann-Hilbert problem is
\begin{equation}\label{outerParametrixSolution}
    \dotbold{P}^\out(z) = 
    \begin{bmatrix}
        \frac{1}{2}\Big(\beta(z) + \frac{1}{\beta(z)}\Big) & \frac{i}{2}\Big(\beta(z) - \frac{1}{\beta(z)}\Big) \\
        - \frac{i}{2}\Big(\beta(z) - \frac{1}{\beta(z)}\Big) & \frac{1}{2} \Big(\beta(z) + \frac{1}{\beta(z)}\Big)
    \end{bmatrix},
\end{equation}
where $\beta(z;x) \equiv \beta(z)$ denotes the function analytic for $z \in \Ci \setminus \Sigma$ that satisfies the conditions
\begin{equation}
    \beta(z)^4 = \left(\frac{z-b}{z-a}\right) \hspace{.5 cm} \text{and} \hspace{.5 cm} \beta(z) = 1 + \Oh(1/z) \text{ as } z \to \infty. 
\end{equation}
See, for example, Section 3.3 of \cite{Buckingham2014LargedegreeAO}.

\subsection{The Inner Parametrices}
The deformation $\boldsymbol{P}(z) \to \dotbold{P}^{\out}(z)$ was motivated by the fact that the jump contours we ignored decayed to the identity as $k \to \infty$. In light of the theory of small-norm Riemann-Hilbert problems, we need this decay to be uniform (with respect to $z$). However, near $z=a$ and $z=b$ the convergence is not uniform. We also see that $\dotbold{P}^{\out}(z)$ has quarter-root singularities at $a$ and $b$. Therefore, $\dotbold{P}^{\out}(z)$ is not a good approximation of $\boldsymbol{P}(z)$ near $z=a$ and $z=b$, and so we need to approximate Riemann-Hilbert Problem \ref{RHPPGenus0} in a different manner. In particular, in $\D_a$ and $\D_b$ we construct  Airy parametrices that satisfy the same jump conditions that $\boldsymbol{P}(z)$ satisfies. In Section \ref{innerParaSection} we give a systematic way to construct Airy parametrices and provide full details for one example.
\subsubsection{Local Parametrices at the Endpoints $z=a$ and $z=b$.}

Set
\begin{equation}
    \boldsymbol{H}^{(a)}(z) := \dotbold{P}^{\out}(z) e^{-k(h(z)+\lambda/2)\sigma_3}e^{\frac{k}{2}\tau_{a}^{3/2}\sigma_3}\boldsymbol{V}^{-1}\tau_{a}^{-  \sigma_3 /4}
\end{equation}
and
\begin{equation}
    \dotbold{P}^{(a)}:= \boldsymbol{H}^{(a)}(z)k^{-\sigma_3/ 6} \boldsymbol{A}\left( k^{2/3} \tau_a(z) \right) e^{-\frac{k}{2}\tau_{a}^{3/2}\sigma_3} e^{k(h(z)+\lambda/3)\sigma_3},
\end{equation}
where $\boldsymbol{A}(\cdot)$ is the well-known solution to Riemann-Hilbert Problem \ref{AiryRHP} in Appendix \ref{AiryAppendix} and $\boldsymbol{V}$ is defined in \eqref{VMatDef}.
In $\D_b$ set
\begin{equation}
    \boldsymbol{H}^{(b)}(z) := \dotbold{P}^{\out}(z) e^{-k(h(z)+\lambda/2)\sigma_3} e^{\frac{i\pi}{2}\sigma_3} e^{\frac{k}{2}\tau_{b}(z)^{3/2}\sigma_3}\boldsymbol{V}^{-1} \tau_{b}(z)^{-\sigma_3/4}
\end{equation}
and
\begin{equation}
    \dotbold{P}^{(b)}(z) := \boldsymbol{H}^{(b)}(z)k^{-\sigma_{3}/6}\boldsymbol{A}\left( k^{2/3} \tau_{b}(z)\right) e^{-\frac{k}{2} \tau_{b}(z)^{3/2}\sigma_3} e^{-\frac{i\pi}{2}\sigma_3} e^{k(h(z)+\lambda/2)\sigma_3}.
\end{equation}
For $p \in \left\{ a,b\right\}$, for each $z \in \partial \D_p$ as $k \to \infty$ the mismatch jumps satisfy
\begin{equation}\label{poleFreeMisMatch}
   \dotbold{P}^{(p)}(z) \dotbold{P}^{\out}(z)^{-1} = \I +\begin{bmatrix}
        \Oh\left( k^{-2} \right) & \Oh\left( k^{-1} \right) \\
        \Oh\left( k^{-1} \right) & \Oh\left( k^{-2} \right)
    \end{bmatrix}.
\end{equation}

\subsection{Error Analysis}

In this section we will quantify the error introduced by ignoring the jumps of $\boldsymbol{P}(z)$ that decayed to the identity as $k \to \infty$ by analyzing a small-norm Riemann-Hilbert problem. First, we will compose the global parametrix out of the local and outer parametrices. Then, we will set up the Riemann-Hilbert problem that characterizes the error between the global parametrix and $\boldsymbol{P}(z)$, the actual solution. Finally, using the theory of small-norm Riemann-Hilbert problems, we will show that our approximation of the right-hand side of \eqref{goodRiddance} is of order $1/k$.
\subsubsection{The Global Parametrix}

Now that we have constructed a parametrix in all parts of the $z$-plane, we are ready to define the global parametrix. Naturally, we set 
\begin{equation}\label{genusOneGlobalParametrix}
    \dotbold{P}(z;x,k) \equiv \dotbold{P}(z) = 
    \begin{cases}
        \dotbold{P}^{(a)}(z) & z \in \D_a,\\
        \dotbold{P}^{(b)}(z) & z \in \D_b,\\
        \dotbold{P}^\out(z)  & z \in \Ci \setminus \left(  \Sigma \cup \overline{\D_a \cup \D_b}    \right).
    \end{cases}
\end{equation}
Set the (unknown) error matrix to be the ratio between $\boldsymbol{P}(z)$ and $\dotbold{P}(z)$. That is,
\begin{equation}
    \boldsymbol{E}(z;x,k) \equiv \boldsymbol{E}(z) := \boldsymbol{P}(z) \dotbold{P}^{-1}(z).
\end{equation}
Notice, when $z$ is not on the jump contours of $\boldsymbol{P}(z)$ or $\dotbold{P}(z)$ then $\boldsymbol{E}(z)$ is a product of analytic functions, and hence is analytic. Thus, $\boldsymbol{E}(z)$ will have a jump discontinuity if either $z$ is on a jump contour of $\boldsymbol{P}(z)$ that is not $\Sigma$ or $z$ is on the boundary of $\D_a \cup \D_b$; see Figure \ref{JumpsofE}. Also, since both $\boldsymbol{P}(z)$ and $\dotbold{P}(z)$ converge to the identity as $z \to \infty$, $\boldsymbol{E}(z)$ converges to the identity as $z \to \infty$. Thus, we find that $\boldsymbol{E}(z)$ is the unique solution to the following Riemann-Hilbert problem:

\begin{RHP}\label{lazyGuy}
    For each $k \in \N$, find $\boldsymbol{E}(z)$ so that 
    \begin{enumerate}
        \item \textbf{Analyticity}:
            $\boldsymbol{E}(z)$ is analytic in $z$ except along the jump contours shown in Figure \ref{JumpsofE}.
        \item \textbf{Jump Condition:}
            $\boldsymbol{E}(z)$ can be continuously extended to the boundary and the boundary values taken by $\boldsymbol{E}(z)$ are related by the jump condition $\boldsymbol{E}_+(z)=\boldsymbol{E}_-(z)\boldsymbol{V}^{(\boldsymbol{E})}(z)$ as shown in Figure \ref{JumpsofE}.
            \begin{figure}[h]
                \centering
                \begin{tikzpicture}[]
                    
                    \draw[thick] (2.71971,.8828) -- (30:  3.75cm);
                    \draw[thick] (30: 3.75cm) -- (30: 7cm);
                    \draw[very thick,->] (30: 3.9cm) -- (30: 3.92);
                    \node(jumpq1) at (42: 4.7cm) {$\dotbold{P}\boldsymbol{L}_k(2h+\lambda)\dotbold{P}^{-1}$};

                    \draw[thick] (-2.71971,.8828)-- (150: 3.75cm);
                    \draw[thick] (150: 3.75cm) -- (150: 7cm);
                    \draw[very thick,->] (150: 3.9cm) -- (150: 3.92cm);
                    \node(jumpq2) at (136: 4.7cm) {$\dotbold{P}\boldsymbol{L}_k^{-1}(2h + \lambda)\dotbold{P}^{-1}$};
            
                    \draw[thick] (-2.71971,-.8828) -- (210: 3.75cm);
                    \draw[thick] (210: 3.75cm)--(210: 7cm);
                    \draw[very thick, ->] (210: 3.9cm) -- (210: 3.92cm);
                    \node(jumpq3) at (-135: 4.7cm) {$\dotbold{P}\boldsymbol{U}_{k}(-(2h+\lambda))\dotbold{P}^{-1}$};
                    
                    \draw[thick] (2.71971,-.8828)-- (330: 3.75cm);
                    \draw[thick] (330: 3.75cm) -- (330: 7cm);
                    \draw[very thick, ->] (330: 3.9cm) -- (330: 3.92cm);
                    \node(jumpq4) at (-47: 4.7cm) {$\dotbold{P}\boldsymbol{U}_{k}^{-1}(-(2h+\lambda))\dotbold{P}^{-1}$};

                     \node(Da) at (-2.25,0) {$\mathbb{D}_a$};
                     \draw[thick] (-2.25, 0) circle (1);
                     \node(DaJump) at (-1,.8) {$\boldsymbol{V}_{a}^{(\boldsymbol{E})}$};
                     \draw[very thick,->] (2.25,1) -- (2.25001,1); 
                     
                     \node(Db) at (2.25,0) {$\mathbb{D}_b$};
                     \draw[thick] (2.25,0) circle (1cm);
                     \node(DaJump) at (1.1,.75) {$\boldsymbol{V}_{b}^{(\boldsymbol{E})}$};
                     \draw[very thick,->] (-2.25,1) -- (-2.2499,1);

                     \draw[thick] (1.28*-.9941,1.28*-.1091).. controls (-.05,3.5*-.1091) and (.05,3.5*-.1091)..(-1.28*-.9941,1.28*-.1091);
                     \draw[very thick,->] (0,-.315) -- (.001, -.315);
                     \node(lowerLenJump) at (-91: 1.1cm) {$\dotbold{P}\boldsymbol{L}_{k}^{-1}(2h+\lambda)\dotbold{P}^{-1}$};
                \end{tikzpicture}
                \caption{The jumps of $\boldsymbol{E}(z)$.}
                \label{JumpsofE}
            \end{figure}
        \item \textbf{Normalization:} As $z \to \infty$, the matrix $\boldsymbol{E}(z)= \I + \Oh\left(1/z\right)$.    
    \end{enumerate}
\end{RHP}

Notice, not only does $\boldsymbol{E}(z)$ normalize to the identity matrix, but all of its jumps are close to the identity for large $k$. Indeed, due to the sign of $\Real(2h+\lambda)$, all of the triangular matrices in Figure \ref{JumpsofE} are exponentially close to the identity (and clearly conjugating by the bounded matrix $\dotbold{P}$ does not cause growth). Further, the mismatch jumps on the boundaries of $\D_a$ and $\D_b$ have been shown to go like $\I + \Oh\left( 1/k \right)$ as $k \to \infty$; see \eqref{poleFreeMisMatch}. The upshot is that, according to the theory of small-norm Riemann-Hilbert problems, the following proposition is true. 

\begin{proposition}
    Given that $x$ is in the pole-free region of the $x$-plane, as $k \to \infty$
    \begin{equation}\label{Error Result}
        \boldsymbol{E}(z;x,k) = \I + \Oh(1/k).
    \end{equation}
\end{proposition}

\subsubsection{Tracking the Error Term in the Extraction Formula }
We recall the transformations and deformations we employed to arrive at $\dotbold{P}$.
\begin{equation}
    \begin{tikzcd}[ampersand replacement=\&]
        \boldsymbol{M}(z) \arrow[r] \& \boldsymbol{N}(z) \arrow[l] \arrow[r]\& \boldsymbol{O}(z) \arrow[l] \arrow[r] \& \boldsymbol{P}(z) \arrow[l] \arrow[r] \& \dotbold{P}(z)
        \end{tikzcd}
\end{equation}
We want to approximate $\boldsymbol{M}(z)$ in terms of the explicit function $\dotbold{P}(z)$. Towards this goal we start off by observing that the transformations $\boldsymbol{M}(z) \to \boldsymbol{N}(z) \to \boldsymbol{O}(z)$ all occur in a bounded subset of the $z$-plane. Hence, for large enough $z$, we have $\boldsymbol{M}(z) = \boldsymbol{N}(z) = \boldsymbol{O}(z)$. Thus, inverting  \eqref{defOfPkgenusZero}, we find for large $z$,
\begin{equation}
    \boldsymbol{M}(z) = e^{\frac{k\lambda}{2}\sigma_3}\boldsymbol{P}(z)e^{kg(z)\sigma_3} e^{-\frac{k\lambda}{2}\sigma_3} = e^{\frac{k\lambda}{2}\sigma_3}\boldsymbol{E}(z) \dotbold{P}(z)e^{kg(z)\sigma_3} e^{-\frac{k\lambda}{2}\sigma_3} .
\end{equation}
Now, write the large-$z$ expansions for $\boldsymbol{E}(z)$ and $\dotbold{P}(z)$ as
\begin{equation}\label{jiveSamba}
    \boldsymbol{E}(z)= \I + \frac{\boldsymbol{E}_{-1}}{z} + \frac{\boldsymbol{E}_{-2}}{z^2} + \Oh\left(\frac{1}{z^3}\right) \hspace{.25cm} \text{and} \hspace{.25cm} \dotbold{P}(z) = \I + \frac{\dotbold{P}_{-1}}{z} + \frac{\dotbold{P}_{-2}}{z^2} + \Oh\left(\frac{1}{z^3}\right).
\end{equation}
Also, using \eqref{gfunNormg0} we find  there exists some $ g_{-1}(x) \equiv g_{-1} \in \Ci$ such that as, $z \to \infty$,
\begin{equation}
    g(z) = \log(z) + \frac{g_{-1}}{z} + \Oh\left( 1/z^2 \right).
\end{equation} 
This implies as, $z \to \infty$,
\begin{equation}\label{gfunLargezExpan}
    e^{kg(z)\sigma_3} z^{-k\sigma_3} = \I + \frac{kg_{-1}\sigma_3}{z} +\begin{bmatrix}
        \Oh\left( 1/z^2 \right) & \Oh\left( 1/z^3 \right) \\
        \Oh\left( 1/z^3 \right) & \Oh\left( 1/z^2 \right)
    \end{bmatrix}.
\end{equation}
In view of \eqref{goodRiddance}, we extract $[\boldsymbol{M}_{-1}]_{22}$, $[\boldsymbol{M}_{-1}]_{12}$, and $[\boldsymbol{M}_{-2}]_{12}$ and find
\begin{align}
    \begin{split}
        [\boldsymbol{M}_{-1}]_{22} &= [\boldsymbol{E}_{-1}]_{22} + [\dotbold{P}_{-1}]_{22} - kg_{-1}, \\
        [\boldsymbol{M}_{-1}]_{12} &= \left( [\boldsymbol{E}_{-1}]_{12} + [\dotbold{P}_{-1}]_{12} \right) e^{k\lambda}, \\
        [\boldsymbol{M}_{-2}]_{12} &= e^{k\lambda} \Big([\boldsymbol{E}_{-1}]_{11}[\dotbold{P}_{-1}]_{12} + [\boldsymbol{E}_{-1}]_{12}[\dotbold{P}_{-1}]_{22} - \\
        & \hspace{2.25 cm} \left( [\boldsymbol{E}_{-1}]_{12} +[\dotbold{P}_{-1}]_{12}\right)kg_{-1} + [\boldsymbol{E}_{-2}]_{12} + [\dotbold{P}_{-2}]_{22}\Big).
    \end{split}
\end{align}
Therefore,
\begin{multline}\label{pk asymptotic Formula in g0}
    p_k\left( \left(\frac{k}{2}\right)^{2/3} x\right) = 2i \left( \frac{k}{2} \right)^{1/3} \Bigg( [\boldsymbol{E}_{-1}]_{22} + [\dotbold{P}_{-1}]_{22} - \\ 
    \frac{[\boldsymbol{E}_{-1}]_{11}[\dotbold{P}_{-1}]_{12} + [\boldsymbol{E}_{-1}]_{12}[\dotbold{P}_{-1}]_{22}+[\boldsymbol{E}_{-2}]_{12} + [\dotbold{P}_{-2}]_{12}}{[\boldsymbol{E}_{-1}]_{12} + [\dotbold{P}_{-1}]_{12}} \Bigg).
\end{multline}
A direct calculation finds that
\begin{equation}\label{P22Genus0}
    [\dotbold{P}(z)]_{22} = \frac{1}{2} \left(\beta(z)+ \frac{1}{\beta(z)}\right) = 1 + \Oh\left(z^{-2}\right)
\end{equation}
and
\begin{equation}\label{P12Genus0}
    [\dotbold{P}(z)]_{12} = \frac{i}{2} \left(\beta(z)-\frac{1}{\beta(z)}\right) = \frac{i(a-b)}{4 z} + \frac{i(a^2-b^2)}{8 z^2} + \Oh\left(z^{-3}\right).
\end{equation}
We are now ready to prove Theorem \ref{genus0 theorem}.

\begin{proof}[Proof of Theorem \ref{genus0 theorem}]
    Using \eqref{Error Result}, we can write \eqref{pk asymptotic Formula in g0} as
    \begin{equation}
        p_k\left( \left(\frac{k}{2}\right)^{2/3} x\right) = 2i \left( \frac{k}{2} \right)^{1/3} \left( [\dotbold{P}_{-1}(x)]_{22} - \frac{[\dotbold{P}_{-2}(x)]_{12}+\Oh\left( 1/k \right)}{[\dotbold{P}_{-1}(x)]_{12}+ \Oh\left( 1/k \right) } \right).
    \end{equation}
    Further, looking at \eqref{P12Genus0} we have that
    \begin{equation}\label{fractionalFormulapk}
        [\dotbold{P}_{-1}(x)]_{12} = \frac{i(a(x)-b(x))}{4} = - \frac{i\Delta(x)}{4}.
    \end{equation}
    Notice that $\Delta(x)= \sqrt{-4i/S(x)}$ never vanishes in the pole-free region of the $x$-plane because $S(x)$ is analytic (and hence pole free) in the pole-free region. Therefore, the quotient in \eqref{fractionalFormulapk} never blows up and we can write
    \begin{equation}\label{noidea1}
        p_k\left( \left(\frac{k}{2}\right)^{2/3} x\right) = 2i \left( \frac{k}{2} \right)^{1/3} \left( [\dotbold{P}_{-1}(x)]_{22} - \frac{[\dotbold{P}_{-2}(x)]_{12}}{[\dotbold{P}_{-1}(x)]_{12}} + \Oh(1/k) \right). 
    \end{equation}
    Now, using \eqref{P22Genus0} and \eqref{P12Genus0} to express the right-hand side of \eqref{noidea1} in terms of $S$ we find
    \begin{equation}
        p_k\left( \left(\frac{k}{2}\right)^{2/3} x\right) = -2i \left( \frac{k}{2} \right)^{1/3} \left( \frac{S(x)}{2} + \Oh(1/k) \right).
    \end{equation}
    In view of \eqref{AnnoyingScalingFormula} and \eqref{Scalings} we obtain,
    \begin{equation}
        -(2k)^{-1/3}u^{(\alpha)}_{\HM}\left( - \frac{k^{2/3}}{2^{1/3}}x\right)= -i \frac{S(x)}{2} + \Oh\left( 1/k \right),
    \end{equation}
    our desired result.
\end{proof}

\section{Analysis in the Pole Region} \label{GenusOneAnalysis}
We now begin our analysis in the pole region.
\subsection[]{The G-function} \label{The G-function pole section}
    Motivated by the bifurcation of the zero-level set  of $\Real(2h+\lambda)$ as $x$ leaves the pole-free region (see Figure \ref{zeroLevelLinesGenusZeroAll}), we aim to deform Riemann-Hilbert Problem \ref{RHPMScaled} so that its jump matrices decay to the identity as $k \to \infty$ along all of its jump contours except for two bands and one gap (instead of one band as in the pole-free case). As we will have two bands, we will have four endpoints: $A(x)\equiv A$, $B(x)\equiv B$, $C(x) \equiv C$, and $D(x) \equiv D$. The normalization of the soon-to-be-defined $G$-function will be dictated by the endpoints. In order to enforce our desired normalization we require
\begin{align}
    \begin{split}\label{g1MomentConditions}
    \s_1 &:= A+B+C+D = 0, \\
    \s_2 &:= AB+AC+AD+BC+BD+CD =x/2, \\
    \s_3 &:= ABC+ACD+BCD =-i. 
    \end{split}
\end{align}
These conditions give us six real equations to determine the endpoints. The remaining two real equations will be derived from what are commonly known as the Boutroux conditions. To state these conditions we need to introduce the function $R(z;x)$. Temporarily set $\Sigma_1$ as the line segment with endpoints $A$ and $B$, $\Gamma$ as the line segment with endpoints $B$ and $C$, and $\Sigma_2$ as the line segment with endpoints $C$ and $D$. Now, let $R(z;x) \equiv R(z)$ denote the function analytic for $z \in \Ci \setminus (\Sigma_1 \cup \Sigma_2)$ that satisfies the conditions
\begin{equation}\label{defOfRGenusOne}
    R(z)^{2} = (z-A)(z-B)(z-C)(z-D) \hspace{.33cm} \text{and} \hspace{.33cm} R(z) = z^2 + \Oh(z) \text{ as } z \to \infty.
\end{equation}
Also, the moment conditions \eqref{g1MomentConditions} imply $R(z)$ has the large-$z$ expansion 
\begin{equation}
    R(z) = z^2 + \frac{x}{4} + \frac{i}{2z} + \Oh\left( 1/z^2\right).
\end{equation}
The Boutroux conditions are
\begin{equation} \label{Boutroux}
    \text{Im}\left(\int\limits_{\Sigma_1}R_{+}(w) \dd w\right) = 0 \hspace{.5cm} \text{and} \hspace{.5cm} \Imag\left(\int\limits_{\Gamma}R(w) \dd w\right) = 0.
\end{equation}

\subsubsection{Definition and Properties of the G-function}
As in the pole-free case, we will define the $G$-function as an antiderivative. Let
\begin{equation}\label{G1Gprimeintegral}
    G'(z;x) \equiv G'(z) = \frac{R(z)}{2\pi i} \int\limits_{\Sigma_1\cup \Sigma_2}  \frac{i \theta'(w)}{R_+(w)(w-z)}\dd w,
\end{equation}
where the path of integration on $\Sigma_1 \cup \Sigma_2$ goes from $A$ to $B$ and then from $C$ to $D$. The moment conditions \eqref{g1MomentConditions}  guarantee that $G'(z) = 1/z + \Oh\left( 1/z^2 \right)$ as $z \to \infty$. Thus, $G'(z)- 1/(z-A)$ is integrable at $z=\infty$. Integrating the latter expression introduces a logarithmic cut. Let $L$ denote an oriented, unbounded arc from  $z=\infty$ to $z=A$. We further require that $L$ only intersects $\Sigma_1$ at $z=A$, avoids $\Gamma \cup \Sigma_2$, and coincides with the negative real numbers for large enough $z$. We now define $G(z)$ as an integral where the path of integration can be any path that does not cross $\Sigma_1 \cup \Gamma \cup \Sigma_2$. Set
\begin{equation}\label{defofG pole}
    G(z;x) \equiv G(z) := \log(z-A) + \int_{\infty}^{z} G'(w) - \frac{1}{w - A} \dd w.
\end{equation}
For notational convenience, we write
\begin{equation}\label{defofH pole}
    H(z;x) \equiv H(z) := \frac{i\theta(z)}{2}- G(z).
\end{equation}

Analogously to the pole-free case, the zero-level lines of $\Real\left( 2H(z)+\Lambda \right)$ are independent of the placement of the curves $\Sigma_1$, $\Gamma$, and $\Sigma_2$ given the three curves are finite and simple, and $\Sigma_1$ connects $A$ to $B$, $\Gamma$ connects $B$ to $C$, and $\Sigma_2$ connects $C$ to $D$. It is convenient to make the following selections:
\begin{itemize}
    \item Let $\Sigma_1$ coincide with the zero-level arc of $\Real\left( 2H+\Lambda \right)$ that connects $A$ to $B$.
    \item Let $\Gamma$ coincide with the zero-level arc of $\Real\left( 2H+\Lambda \right)$ that connects $B$ to $C$.
    \item Let $\Sigma_2$ coincide with the zero-level arc of $\Real\left( 2H+\Lambda \right)$ that connects $C$ to $D$.
\end{itemize} 

In the same vein as in the pole-free case, the leading error in our approximation will propagate from the endpoints. As such, we will need to do local analysis at each endpoint. For each $p \in \{A,B,C,D\}$, let $\D_p$ denote a sufficiently small disk around $z=p$.
For the points $A,B,C$ we will need to further specify certain regions of $\D_p$. In particular, we make the following partitions (also see Figure \ref{AlltheCircles}):
\begin{itemize}
    \item For $\D_{A}$, notice $L \cup \Sigma_1$ divides $\D_A$ into two open components. Let $\D^{\uparrow}_{A}$ denote the open component that is on the left side of $L \cup \Sigma_1$, and set $\D_A^{\downarrow}:= \D_A \setminus \overline{\D^{\uparrow}_{A}}$.
    \item For $\D_{B}$, notice $\Sigma_1 \cup \Gamma$ divides $\D_B$ into two open components (recall that $\Sigma_1$ and $\widetilde{\Sigma}_1$ are the same contours with opposite orientations). Let $\D^{\uparrow}_{B}$ denote the open component that is on the left side of $\Sigma_1 \cup \Gamma$, and set $\D^{\downarrow}_{B}:=\D_{B} \setminus \overline{\D^{\uparrow}_{B}}$.
    \item For $\D_{C}$, notice $\Gamma \cup \Sigma_2$ divides $\D_C$ into two components. Let $\D_{C}^{\uparrow}$ denote the open component that is on the left side of $\Gamma \cup \Sigma_2$, and set $\D_{C}^{\downarrow} :=\D_C \setminus \overline{\D^{\uparrow}_{C}}$.    
    
\end{itemize}

We now list the basic properties of $G$ and $H$ that we will use throughout the Riemann-Hilbert analysis.
\begin{proposition}\label{G and H fun prop}
    For each $x$ in the pole region of the $x$-plane, given that the moment conditions \eqref{g1MomentConditions} and the Boutroux conditions \eqref{Boutroux} are satisfied, guaranteeing that $G(z)$ and $H(z)$ are well-defined by \eqref{defofG pole} and \eqref{defofH pole} respectively, then the following statements hold:
    \begin{itemize}
        \item There exists a complex number $\Lambda(x) \equiv \Lambda$ such that 
        \begin{equation} \label{sigma1Jump}
        H_+(z) + H_{-}(z) = i\theta(z) - G_{+}(z) - G_{-}(z) = - \Lambda \hspace{.5 cm} \textit{for each } z \in \Sigma_1.    
        \end{equation}
        \item There exists a real number $\omega(x) \equiv \omega \in \R$ such that 
        \begin{equation}\label{gammaJump}
            H_{+} - H_{-} = - (G_{+}-G_{-}) = -i\omega \hspace{.5 cm} \textit{for each } z \in \Gamma.  
        \end{equation}
        \item  There exists a real number $\Omega(x) \equiv \Omega \in \R$ such that 
        \begin{equation}\label{sigma2Jump}
            H_+(z) + H_{-}(z) = i\theta(z) - G_{+}(z) - G_{-}(z) = - \Lambda -i \Omega \hspace{.5 cm} \textit{for each } z \in \Sigma_2.    
        \end{equation}
        \item On the logarithmic cut $L$,
        \begin{equation}
            G_{+}(z) - G_{-}(z) = -2\pi i.
        \end{equation}
        \item  As $z \to \infty$,
        \begin{equation}
            G(z) = \log(z) + \Oh(1/z).
        \end{equation}
        \item At each endpoint, under the right perturbation, $2H(z)+i\Lambda$ locally looks like a $3/2$-root function. More precisely, if we consider the functions
        \begin{equation}
            J_{A}(z):=   \begin{cases}
                2H(z)+\Lambda - 2\pi i, & z \in \D_A^{\uparrow}, \\
                2H(z)+\Lambda + 2\pi i, & z \in \D_A^{\downarrow},
            \end{cases}
        \end{equation}
        \begin{equation}
            J_{B}(z) :=  \begin{cases}
                2H(z)+\Lambda + i\omega, & z \in \D^{\uparrow}_{B}, \\
                2H(z)+\Lambda-i\omega, & z \in \D^{\downarrow}_{B},
               \end{cases}
        \end{equation}
        \begin{equation}
            J_{C}(z) := \begin{cases}
                2H(z) + \Lambda  +i\omega + i\Omega, & z \in \D^{\uparrow}_C,\\
                2H(z) + \Lambda  -i\omega + i\Omega & z \in \D^{\downarrow}_C,
            \end{cases}
        \end{equation}
        and
        \begin{equation}
            J_{D}(z):= 2H(z) + \Lambda + i\Omega,
        \end{equation}
        then for each endpoint $p \in \left\{ A,B,C,D \right\}$ there exists $\kappa_{p} \in \Ci$ such that
        \begin{equation}\label{2H near endpoint}
           J_{p}(z) = \kappa_{p}(z-p)^{3/2} + \Oh\left( (z-p)^{5/2}\right) \hspace{.33 in} \text{for } z \in \D_{p}.
        \end{equation}
    \end{itemize} 
\end{proposition}
These statements are standard and can be proved by direct computation.
\subsection{Preliminary Deformations in the Riemann-Hilbert Analysis}
As in the pole-free analysis, we deform Riemann-Hilbert Problem \ref{RHPMScaled} in preparation  for introducing the $G$-function. These deformations are motivated by the sign chart of $\Real(2H+\Lambda)$; see Figure \ref{GenusOneZeroLevelLines}. To specify the exact placement of the deformed jump contours, we will construct conformal maps at each endpoint. Then, we will collapse contours and open lenses.

\subsubsection{Conformal Map Constructions} Using the same argument, as in Section \ref{conformalMapSection} and Proposition \ref{G and H fun prop}, we are able to construct conformal maps $\tau_A$, $\tau_B$, $\tau_C$, and $\tau_D$ according to the table below.

\begin{table}[h]
    \centering
    \begin{tabular}{|c|c|c|}
        \hline
        $p$ & Conformal Map $\tau_p: \D_p \to \Ci$ &  $\tau_p^{-1}((-\infty,0]) \cap \D_p$ \\
        \hline
        $z=A$ & $
        \tau_A(z)^3 = \begin{cases}
            \left( 2H(z)+\Lambda - 2\pi i \right)^{2}, & z \in \D_A^{\uparrow}, \\
            \left( 2H(z)+\Lambda + 2\pi i \right)^{2}, & z \in \D_A^{\downarrow},
        \end{cases}
        $& $\Sigma_1 \cap \D_A$ \\
        \hline
        $z=B$ & $ \tau_{B}(z)^3 = \begin{cases}
         (2H(z)+\Lambda + i\omega)^2, & z \in \D^{\uparrow}_{B}, \\
        (2H(z)+\Lambda-i\omega)^2, & z \in \D^{\downarrow}_{B},
        \end{cases}$
        & $\Sigma_1 \cap \D_B$ \\
        \hline
        $z=C$ & $
        \tau_C(z)^3 = \begin{cases}
            (2H(z) + \Lambda  +i\omega + i\Omega)^2, & z \in \D^{\uparrow}_C,\\
            (2H(z) + \Lambda  -i\omega + i\Omega)^2 & z \in \D^{\downarrow}_C,
        \end{cases}
        $ 
        & $\Gamma \cap \D_C$ \\
        \hline
        $z=D$ & $\tau_D(z)^3= (2H(z)+\Lambda+i\Omega)^2$ & $\Sigma_2 \cap \D_D$ \\
        \hline
    \end{tabular}
    \caption{Conformal maps used in the pole region.}
\end{table}

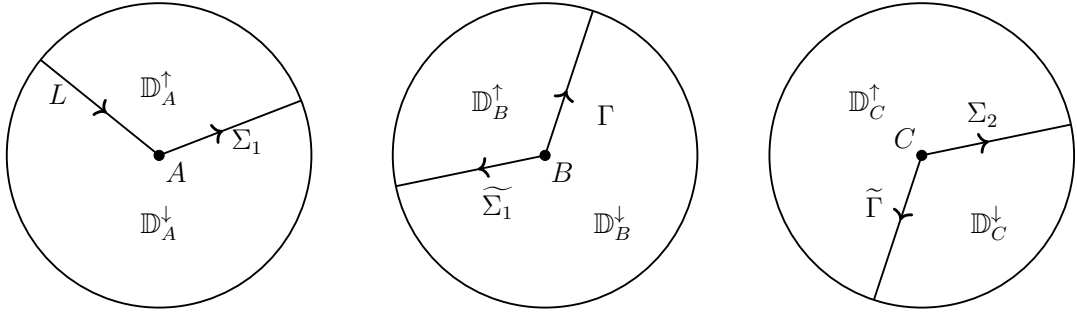
\begin{figure}[ht]
    \centering
    \begin{tabular}{ccccc}
        \scalebox{.9}{\begin{tikzpicture}
            \draw[thick] (0,0) circle (2.25cm);
            \draw[rotate=21.1139,thick] (0,0) -- (0:2.25cm);
            \draw[rotate=21.1139, very thick,->] (0:1cm) -- (0:1.01cm);
            \node at (1.3,.2) {$\Sigma_1$};
            \draw[rotate=21.1139,thick] (0,0) -- (120:2.25cm);
            \draw[rotate=21.1139, very thick,->] (120:1.01cm) -- (120:1.cm);
            \node at (-1.5,.9) {$L$};
            \draw[fill=black] (0,0) circle (.075cm);
            \node at (.25,-.25) {$A$};
            \node at (0,-1.) {$\D_{A}^{\downarrow}$};
            \node at (0,1) {$\D_{A}^{\uparrow}$};
        \end{tikzpicture}}
        & \textcolor{white}{fl}&
       \scalebox{.9}{ \begin{tikzpicture}
        \draw[thick] (0,0) circle (2.25cm);
        \draw[rotate=71.6867,thick] (0,0) -- (0:2.25cm);
        \draw[rotate=71.6867, very thick,->] (0:1cm) -- (0:1.01cm);
        \node at (-.7,-.7) {$\widetilde{\Sigma_1}$};
        \draw[rotate=71.6867,thick] (0,0) -- (120:2.25cm);
        \draw[rotate=71.6867, very thick,->] (120:1cm) -- (120:1.01cm);
        \node at (.9,.6) {$\Gamma$};
        \draw[fill=black] (0,0) circle (.075cm);
        \node at (.25,-.25) {$B$};
        \node at (1.,-1.) {$\D_{B}^{\downarrow}$};
        \node at (-.8,.8) {$\D_{B}^{\uparrow}$};
    \end{tikzpicture}}
        & \textcolor{white}{fl}&
        \scalebox{.9}{\begin{tikzpicture}
            \draw[thick] (0,0) circle (2.25cm);
            \draw[rotate=-108.313,thick] (0,0) -- (0:2.25cm);
            \draw[ rotate=-108.313, very thick,->] (0:1cm) -- (0:1.01cm);
            \node at (-.7,-.7) {$\widetilde{\Gamma}$};
            \draw[rotate=-108.313,thick] (0,0) -- (120:2.25cm);
            \draw[ rotate=-108.313, very thick,->] (120:1cm) -- (120:1.01cm);
            \node at (.9,.6) {$\Sigma_2$};
            \draw[fill=black] (0,0) circle (.075cm);
            \node at (-.25,.25) {$C$};
            \node at (1.,-1.) {$\D_{C}^{\downarrow}$};
            \node at (-.8,.8) {$\D_{C}^{\uparrow}$};
        \end{tikzpicture}}
    \end{tabular}
    \caption{The partitions of $\D_A$, $\D_B$, and $\D_C$. Here, $\widetilde{\Sigma_1}$ denotes the same contour $\Sigma_1$ but with reverse orientation. $\widetilde{\Gamma}$ is defined analogously.}\label{AlltheCircles}
\end{figure}

\subsubsection{The Collapsing Deformation} We collapse the jump contours of $\boldsymbol{M}(z)$ according to Figure \ref{jumpsOfNGenusOne}. We make the following specifications:
\begin{itemize}
    \item The jump contour that connects $A$ to $B$ coincides with $\Sigma_1$.
    \item In $\D_A$, the jump contour where the jump is upper triangular is contained in $\tau_{A}^{-1}((0,\infty))$.
    \item In $\D_A$, the jump contour where the jump is lower triangular is contained in $\tau_{B}^{-1}(\rho e^{2\pi i /3})$ for $\rho>0$.
    \item In $\D_B$, the jump contour where the jump is upper triangular is contained in $\tau_{B}^{-1}((0,\infty))$ for $\rho>0$.
    \item In $\D_B$, the jump contour that connects $B$ to $C$ is contained in $\tau_{B}^{-1}(\rho e^{3\pi i /2})$ for $\rho>0$.
    \item The jump contour that connects $C$ to $D$ coincides with $\Sigma_2$.
    \item In $\D_C$, the jump contour that connects $B$ to $C$ is contained in $\tau_C^{-1}(\rho e^{2\pi i /3})$ for $\rho>0$.
    \item In $\D_D$, the jump contour that emanates from $D$ and goes to $\infty$ is contained $\tau_D^{-1}((0,\infty))$.
    \item For all other jumps contours, we only require that they avoid the zero-level set of $\Real(2H+\Lambda)$ and do not self-intersect.
\end{itemize}
Let $\boldsymbol{N}(z)$ denote the solution of the  Riemann-Hilbert problem obtained from the collapsing deformation. More precisely:

\begin{RHP}\label{RHPNGenus0}
    For each $k \in \N$, find $\boldsymbol{N}(z;x,k)\equiv\boldsymbol{N}(z)$ so that 
    \begin{enumerate} 
        \item \textbf{Analyticity}: $\boldsymbol{N}(z)$ is analytic in $z$ except along the jump contours shown in Figure \ref{jumpsOfNGenusOne}.
    
        \item \textbf{Jump Condition:} $\boldsymbol{N}(z)$ can be continuously extended to the boundary and the boundary values taken by $\boldsymbol{N}(z)$ are related by the jump condition $\boldsymbol{N}_+(z) = \boldsymbol{N}_-(z)\boldsymbol{V}^{(\boldsymbol{N})}(z)$ as shown in Figure \ref{jumpsOfNGenusOne}.
        \begin{figure}[ht]
            \centering
            \scalebox{1}{
            \begin{tikzpicture} 
                \draw[thick, blue, xshift=-4.1289cm, yshift=-3.07927cm, rotate=21.1139](0,0).. controls (2,.075) and (2.33, .15) ..(4.06248,0);
                \draw[very thick, ->,blue, xshift=-4.1289cm, yshift=-3.07927cm, rotate=21.1139] (2,.1) --(2.001,.1);
                \draw[thick,->,black](-4.1289,-3.07927) .. controls (-5.5,-3.65) and (-5.75,-3.4) .. (-6.7117, -3.875);
                \draw[thick,->](-4.1289,-3.07927) .. controls (-3.5,.5) and (-4.75, 3.075) ..  (150:7.75cm);
                
                \draw[thick,xshift=-0.339153cm, yshift=-1.61587cm, rotate=71.6867](0,0)..controls (1,1.75) and (2.6, 1.75)..(0:3.60408cm);
                \draw[very thick,->,xshift=-0.339153cm, yshift=-1.61587cm, rotate=71.6867] (1.7,1.3)--(1.7001,1.3);
                \draw[thick,->](-0.339153,-1.61587) .. controls (2.25,-.875) and (3,-1.75) ..  (-30:7.75cm);
        
                \draw[thick, xshift=3.67476cm, yshift=2.88945cm, rotate=-159.388,blue](0,0)--(0:3.07853cm);
                \draw[very thick,blue,->, xshift=3.67476cm, yshift=2.88945cm, rotate=-159.388](1.501,0)--(1.5,0);
                \draw[thick,->](3.67476,2.88945) .. controls (4.6,3.375) and (5.5, 3.3) .. (6.7117,3.875);

                \draw[fill=black, xshift=-4.1289cm, yshift=-3.07927cm, rotate=21.1139](0,0) circle (.075cm);
                \node at (-4.2,-3.4) {$A$};
                \draw[fill=black, xshift=-0.339153cm, yshift=-1.61587cm, rotate=71.6867](0,0) circle (.075cm);
                \node at (-.1,-1.8) {$B$};
                \draw[fill=black, xshift=0.793299cm, yshift=1.80568cm, rotate=-108.313](0,0) circle (.075cm);
                \node at (.5,2) {$C$};
                \draw[fill=black, xshift=3.67476cm, yshift=2.88945cm, rotate=-159.388](0,0) circle (.075cm);
                \node at (3.9,3.2) {$D$};
        
                \node at (-5,0) {$\boldsymbol{L}_{k}^{-1}(i\theta)$};
                \node at (-5.5,-4) {$\boldsymbol{U}_{k}(-i\theta)$};
                \node at (-1.85,-3.5) {$\begin{bmatrix}
                     0 & -e^{-ik\theta} \\
                     e^{ik\theta} & 1
                \end{bmatrix}$};
                 \node at (3,-2.5) {$\boldsymbol{U}_{k}^{-1}(-i\theta)$};
                 
                 \node at (3,.25) {$\boldsymbol{L}_{k}(i\theta)$};
                 \draw[thick,dotted,->] (2.4,.25) -- (-.2,.25);
                 \draw[thick,dotted,->] (2.6,.5) -- (2.6, 2);
                 \draw[thick, dotted,->] (2.8, .5) -- (4.5,2.9);
                 \node at (-2.25,-1.9) {$\textcolor{blue}{\Sigma_1}$};
                 \node at (2,2.75) {$\textcolor{blue}{\Sigma_2}$};

        
        \end{tikzpicture}}
    \caption{The jumps of $\boldsymbol{N}(z)$ in the pole case.}
    \label{jumpsOfNGenusOne}
    \end{figure}
    
        \item \textbf{Normalization:} As $z \to \infty$, the matrix $\boldsymbol{N}(z)z^{-k\sigma_3} = \I + \Oh\left(1/z\right)$.
    
    \end{enumerate}
\end{RHP}

\begin{figure}[ht]
    \centering
    \scalebox{.66}{\includegraphics{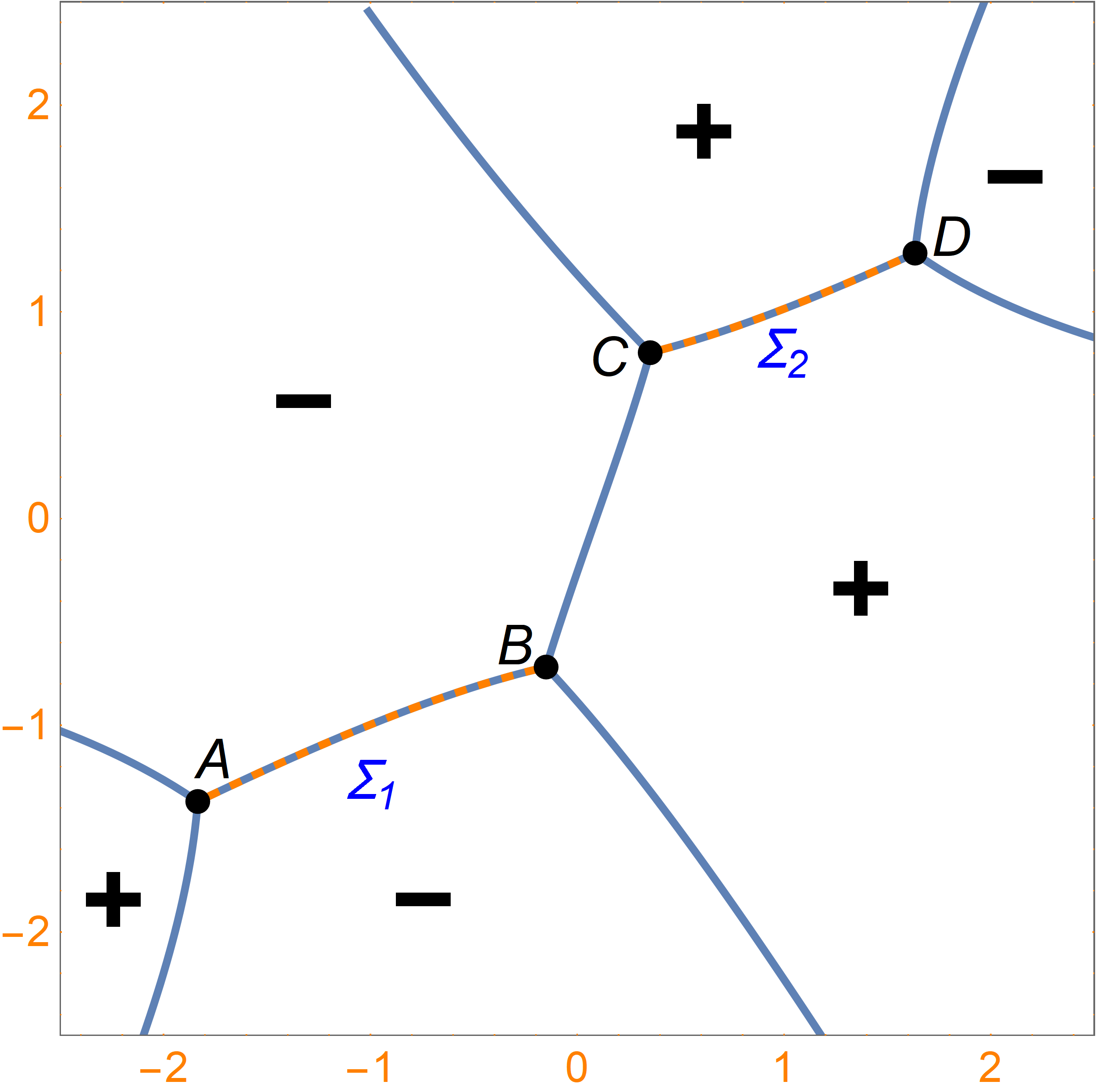}}
    \caption{The sign chart for $\Real(2H+\Lambda)$ when $x=-3/2 - 10i$. The dashed blue-and-orange lines indicate that a zero-level line is also a cut of $\Real(2H+\Lambda)$.}
    \label{GenusOneZeroLevelLines}
\end{figure}

\subsubsection{Opening Lenses} We now open lenses on $\Sigma_1$ and $\Sigma_2$ guided by the sign chart for $\Real\left( 2H+\Lambda \right)$ (see Figure \ref{GenusOneZeroLevelLines}).

On $\Sigma_1$, we use the same factorization we used in the pole-free analysis; see \eqref{lenFactorizationab}. The jump matrix on $\Sigma_2$ has the factorization
\begin{equation}
    \LT{e^{ik\theta}} = \UT{e^{-ik\theta}} 
    \begin{bmatrix}
        0 & -e^{-ik\theta} \\
        e^{ik\theta} & 0
    \end{bmatrix}
    \UT{e^{-ik\theta}}.
\end{equation} 
We open three lenses according to Figure \ref{jumps of O in pole}. We make the following specifications:
\begin{itemize}
    \item In $\D_A$, the boundary of the lens that is on the right side of $\Sigma_1$ is contained in $\tau^{-1}_{A}(\rho e^{2\pi i /3})$ for $\rho>0$.    
    \item In $\D_B$, the boundary of the lens that is on the right side of $\Sigma_1$ is contained in $\tau_{B}^{-1}(\rho e^{-2\pi i /3})$ for $\rho>0$.
    \item In $\D_C$, the boundary of the lens that is on the right side of $\Sigma_2$ is contained in $\tau_C^{-1}(\rho e^{- 2\pi i /3})$ for $\rho>0$.
    \item In $\D_C$, the boundary of the lens that is on the left side of $\Sigma_2$ is contained in $\tau_C^{-1}((0,\infty))$. 
    \item In $\D_D$, the boundary of the lens that is on the right side of $\Sigma_2$ is contained in $\tau_D^{-1}(\rho e^{- 2\pi i /3})$ for $\rho>0$.
    \item In $\D_D$, the  boundary of the lens that is on the left side of $\Sigma_2$ is contained in $\tau_D^{-1}(\rho e^{ 2\pi i /3})$ for $\rho>0$.
\end{itemize}
Let $\boldsymbol{O}(z)$ denote the solution of the Riemann-Hilbert problem obtained from the opening lenses deformation. Specifically:

\begin{RHP}
    Determine the $2 \times 2$ matrix $\boldsymbol{O}(z;x,k)\equiv\boldsymbol{O}(z)$ with the following properties:\begin{enumerate}
        \item\textbf{Analyticity:} $\boldsymbol{O}(z)$ is analytic in $z$ except on the jump contours shown in Figure \ref{jumps of O in pole}.
        \item\textbf{Jump Condition:} The boundary values taken by $\boldsymbol{O}(z)$ are related by the jump condition $\boldsymbol{O}_+(z)=\boldsymbol{O}_-(z)\boldsymbol{V}^{(\boldsymbol{O})}(z)$ as shown in Figure \ref{jumps of O in pole}.
        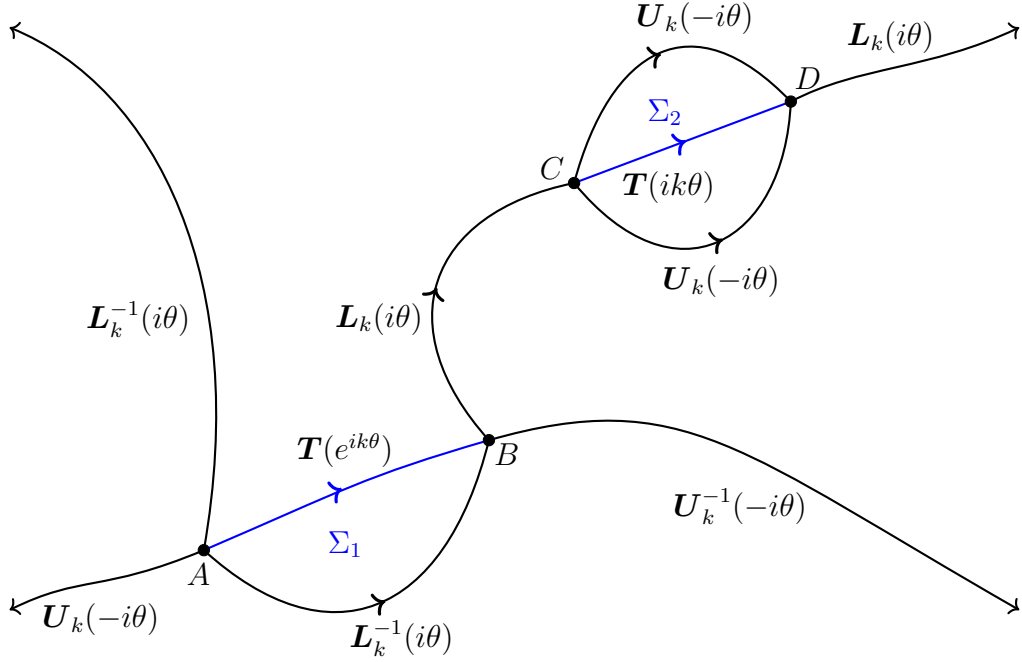
\begin{figure}[ht]
            \centering
            \scalebox{1}{
                \begin{tikzpicture} 
        \draw[thick, blue, xshift=-4.1289cm, yshift=-3.07927cm, rotate=21.1139](0,0).. controls (2,.075) and (2.33, .15) ..(4.06248,0);
        \draw[very thick, ->,blue, xshift=-4.1289cm, yshift=-3.07927cm, rotate=21.1139] (2,.1) --(2.001,.1);
        \draw[thick,xshift=-4.1289cm, yshift=-3.07927cm, rotate=21.1139](0,0).. controls (1.,-2) and (2.66, -2) ..(4.06248,0);
        \draw[very thick, ->, xshift=-4.1289cm, yshift=-3.07927cm, rotate=21.1139] (2,-1.5) --(2.001,-1.5);
        \draw[thick,->,black](-4.1289,-3.07927) .. controls (-5.5,-3.65) and (-5.75,-3.4) .. (-6.7117, -3.875);
        \draw[thick,->](-4.1289,-3.07927) .. controls (-3.5,.5) and (-4.75, 3.075) ..  (150:7.75cm);
        
        
        \draw[thick, xshift=-0.339153cm, yshift=-1.61587cm, rotate=71.6867](0,0)..controls (1,1.75) and (2.6, 1.75)..(0:3.60408cm);
        \draw[very thick,->,xshift=-0.339153cm, yshift=-1.61587cm, rotate=71.6867] (1.7,1.3)--(1.7001,1.3);
        \draw[thick,->](-0.339153,-1.61587) .. controls (2.25,-.875) and (3,-1.75) ..  (-30:7.75cm);
        
        
        
        
        \draw[thick, blue, xshift=3.67476cm, yshift=2.88945cm, rotate=-159.388](0,0)--(0:3.07853cm);
        \draw[very thick, blue, ->, xshift=3.67476cm, yshift=2.88945cm, rotate=-159.388](1.501,0)--(1.5,0);
        \draw[thick, xshift=3.67476cm, yshift=2.88945cm, rotate=-159.388](0,0).. controls (.75,1.75) and (2.4, 2) ..(0:3.07853cm);
        \draw[very thick, ->, xshift=3.67476cm, yshift=2.88945cm, rotate=-159.388](1.501,1.4)--(1.5,1.4);
        \draw[thick, xshift=3.67476cm, yshift=2.88945cm, rotate=-159.388](0,0).. controls (.75,-1.75) and (2, -1.5) ..(0:3.07853cm);
        \draw[very thick, ->, xshift=3.67476cm, yshift=2.88945cm, rotate=-159.388](1.401,-1.2)--(1.4,-1.2);
        \draw[thick,->](3.67476,2.88945) .. controls (4.6,3.375) and (5.5, 3.3) .. (6.7117,3.875);
        

        \draw[fill=black, xshift=-4.1289cm, yshift=-3.07927cm, rotate=21.1139](0,0) circle (.075cm);
        \node at (-4.2,-3.4) {$A$};
        \draw[fill=black, xshift=-0.339153cm, yshift=-1.61587cm, rotate=71.6867](0,0) circle (.075cm);
        \node at (-.1,-1.8) {$B$};
        \draw[fill=black, xshift=0.793299cm, yshift=1.80568cm, rotate=-108.313](0,0) circle (.075cm);
        \node at (.5,2) {$C$};
        \draw[fill=black, xshift=3.67476cm, yshift=2.88945cm, rotate=-159.388](0,0) circle (.075cm);
        \node at (3.9,3.2) {$D$};
        
        \node at (-5,0) {$\boldsymbol{L}_{k}^{-1}(i\theta)$};
        \node at (-5.5,-4) {$\boldsymbol{U}_{k}(-i\theta)$};
        \node at (-1.5,-4.25) {$\boldsymbol{L}_{k}^{-1}(i\theta)$};
        \node at (-2.25, -1.75) {$\boldsymbol{T}(e^{ik\theta})$};
        \node at (-2.25,-3) {$\textcolor{blue}{\Sigma_1}$};
        \node at (3,-2.5) {$\boldsymbol{U}_{k}^{-1}(-i\theta)$};
        \node at (-1.8,0) {$\boldsymbol{L}_{k}(i\theta)$};
        \node at (5,3.75) {$\boldsymbol{L}_{k}(i\theta)$};
        \node at (2,2.75) {$\textcolor{blue}{\Sigma_2}$};
        \node at (2.05,1.75) {$\boldsymbol{T}(ik\theta)$};
        \node at (2.4,4) {$\boldsymbol{U}_{k}(-i\theta)$};
        \node at (2.75,.5) {$\boldsymbol{U}_{k}(-i\theta)$};

                \end{tikzpicture}
            }
            \caption{The jumps of $\boldsymbol{O}(z)$ in the pole case.}
            \label{jumps of O in pole}
        \end{figure}
        \item \textbf{Normalization:} As $z \to \infty$, the matrix $\boldsymbol{O}(z)z^{-k\sigma_3} = \I + \Oh\left(1/z\right)$.
    \end{enumerate}.    
\end{RHP}

\subsection{Applying the G-function.} Analogously to the pole-free case, we apply the $G$-function defined in Section \ref{The G-function pole section} (see \eqref{defofG pole}) to the function $\boldsymbol{O}(z)$. We consider
\begin{equation} \label{defOfPGenusOne}
    \boldsymbol{P}(z;x,k) \equiv \boldsymbol{P}(z) := e^{-\frac{k\Lambda}{2}\sigma_3} \boldsymbol{O}(z) e^{-kG(z)\sigma_3} e^{\frac{k\Lambda}{2}\sigma_3}.
\end{equation}
Then, $\boldsymbol{P}(z)$ is analytic where both $\boldsymbol{O}(z)$ and $e^{G(z)}$ are (that is, for $z \in \Ci \setminus (\Sigma_1 \cup \Gamma \cup \Sigma_2)$). Hence, we have introduced a jump on $\Gamma$. In particular, $\boldsymbol{V}^{(\boldsymbol{P})}_{\Gamma} = e^{-ki\omega\sigma_3}$, where $\omega$ was introduced in \eqref{gammaJump}. For compensation, both the normalization and jump conditions for $\boldsymbol{P}(z)$ are now simplified. Indeed, for large $z$, $\boldsymbol{P}(z)$ converges to $\I$ instead of $z^{k\sigma_3}$.

The jump recipe \eqref{jumpRecipeForPG0} is valid after replacing the lowercase symbols `$g$' and `$\lambda$' with their corresponding uppercase counterparts. From this, we glean that the nontrivial exponential entry of the triangular jump matrices is transformed by the rule $e^{\pm ik\theta} \Rightarrow e^{\pm k (2H+\Lambda)}$. Further, the sign chart of $\Real(2H+\Lambda)$ implies that all of the triangular jump matrices for $\boldsymbol{P}(z)$ decay to the identity as $k \to \infty$. Moreover, the twist jump matrices on $\Sigma_1$  and $\Sigma_2$  become constant. Using the jump recipe,
\eqref{sigma1Jump}, and \eqref{sigma2Jump}, we find
\begin{equation}
    \boldsymbol{V}^{(\boldsymbol{P})}_{\Sigma_1} = \begin{bmatrix}
        0 & -1 \\
        1 & 0
    \end{bmatrix} \hspace{1.33cm} \text{ and } \hspace{1.33cm} \boldsymbol{V}^{(\boldsymbol{P})}_{\Sigma_2} = \begin{bmatrix}
        0 & -e^{ki\Omega} \\
        e^{-ki\Omega} & 0
    \end{bmatrix}.
\end{equation}
Therefore the Riemann-Hilbert problem $\boldsymbol{P}(z)$ solves is: 

\begin{RHP}\label{G1PRHP}
    Determine the $2 \times 2$ matrix $\boldsymbol{P}(z;x,k)\equiv \boldsymbol{P}(z)$ with the following properties:\begin{enumerate}
        \item\textbf{Analyticity:} $\boldsymbol{P}(z)$ is analytic in $z$ except on the jump contours shown in Figure \ref{jumpsOfPGenusOne}.
        \item\textbf{Jump Condition:} The boundary values taken by $\boldsymbol{P}(z)$ are related by the jump condition $\boldsymbol{P}_+(z)=\boldsymbol{P}_-(z)\boldsymbol{V}^{(\boldsymbol{P})}(z)$ as shown in Figure \ref{jumpsOfPGenusOne}.
        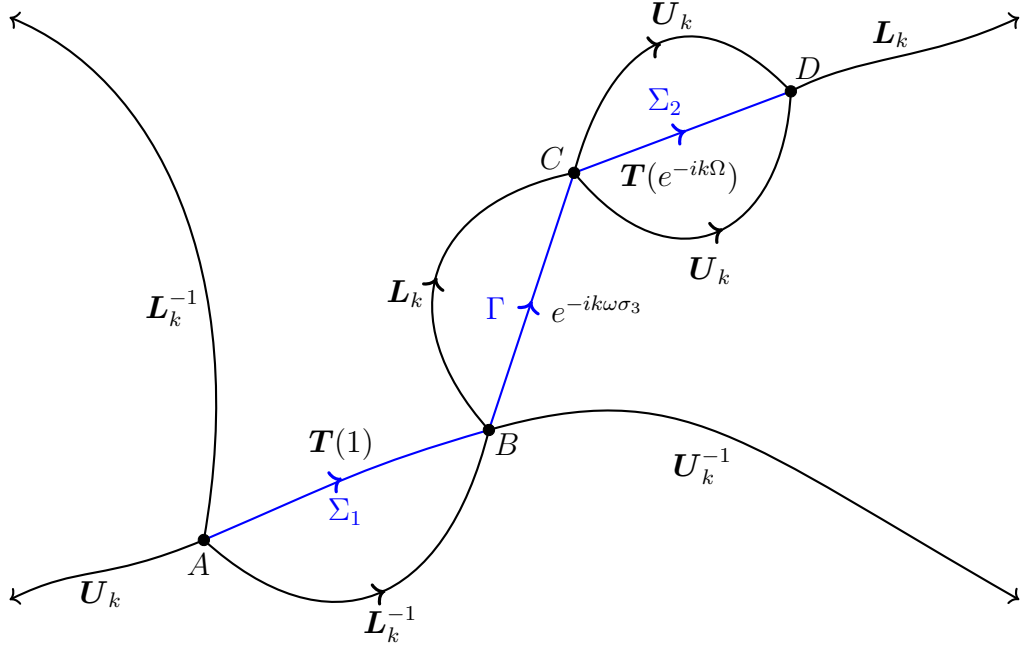
\begin{figure}[h]
            \centering
            \scalebox{1}{
                \begin{tikzpicture} 
        \draw[thick, blue, xshift=-4.1289cm, yshift=-3.07927cm, rotate=21.1139](0,0).. controls (2,.075) and (2.33, .15) ..(4.06248,0);
        \draw[very thick, ->,blue, xshift=-4.1289cm, yshift=-3.07927cm, rotate=21.1139] (2,.1) --(2.001,.1);
        \draw[thick,xshift=-4.1289cm, yshift=-3.07927cm, rotate=21.1139](0,0).. controls (1.,-2) and (2.66, -2) ..(4.06248,0);
        \draw[very thick, ->, xshift=-4.1289cm, yshift=-3.07927cm, rotate=21.1139] (2,-1.5) --(2.001,-1.5);
        \draw[thick,->,black](-4.1289,-3.07927) .. controls (-5.5,-3.65) and (-5.75,-3.4) .. (-6.7117, -3.875);
        \draw[thick,->](-4.1289,-3.07927) .. controls (-3.5,.5) and (-4.75, 3.075) ..  (150:7.75cm);

        \draw[thick, blue , xshift=-0.339153cm, yshift=-1.61587cm, rotate=71.6867](0,0)--(0:3.60408cm);
        \draw[very thick, blue,-> ,xshift=-0.339153cm, yshift=-1.61587cm, rotate=71.6867](1.8,0)--(1.801,0);
        \draw[thick, xshift=-0.339153cm, yshift=-1.61587cm, rotate=71.6867](0,0)..controls (1,1.75) and (2.6, 1.75)..(0:3.60408cm);
        \draw[very thick,->, xshift=-0.339153cm, yshift=-1.61587cm, rotate=71.6867] (1.7,1.3)--(1.7001,1.3);
        \draw[thick,->](-0.339153,-1.61587) .. controls (2.25,-.875) and (3,-1.75) ..  (-30:7.75cm);
        
        
        
        \draw[thick, blue, xshift=3.67476cm, yshift=2.88945cm, rotate=-159.388](0,0)--(0:3.07853cm);
        \draw[very thick, blue, ->, xshift=3.67476cm, yshift=2.88945cm, rotate=-159.388](1.501,0)--(1.5,0);
        \draw[thick, xshift=3.67476cm, yshift=2.88945cm, rotate=-159.388](0,0).. controls (.75,1.75) and (2.4, 2) ..(0:3.07853cm);
        \draw[very thick, ->, xshift=3.67476cm, yshift=2.88945cm, rotate=-159.388](1.501,1.4)--(1.5,1.4);
        \draw[thick, xshift=3.67476cm, yshift=2.88945cm, rotate=-159.388](0,0).. controls (.75,-1.75) and (2, -1.5) ..(0:3.07853cm);
        \draw[very thick, ->, xshift=3.67476cm, yshift=2.88945cm, rotate=-159.388](1.401,-1.2)--(1.4,-1.2);
        \draw[thick,->](3.67476,2.88945) .. controls (4.6,3.375) and (5.5, 3.3) .. (6.7117,3.875);
        

        \draw[fill=black, xshift=-4.1289cm, yshift=-3.07927cm, rotate=21.1139](0,0) circle (.075cm);
        \node at (-4.2,-3.4) {$A$};
        \draw[fill=black, xshift=-0.339153cm, yshift=-1.61587cm, rotate=71.6867](0,0) circle (.075cm);
        \node at (-.1,-1.8) {$B$};
        \draw[fill=black, xshift=0.793299cm, yshift=1.80568cm, rotate=-108.313](0,0) circle (.075cm);
        \node at (.5,2) {$C$};
        \draw[fill=black, xshift=3.67476cm, yshift=2.88945cm, rotate=-159.388](0,0) circle (.075cm);
        \node at (3.9,3.2) {$D$};
        
        \node at (-4.55,0) {$\boldsymbol{L}_{k}^{-1}$};
        \node at (-5.5,-3.8) {$\boldsymbol{U}_{k}$};
        \node at (-1.65,-4.2) {$\boldsymbol{L}^{-1}_{k}$};
        \node at (-2.3, -1.8) {$\boldsymbol{T}(1)$};
        \node at (-2.25,-2.7) {$\textcolor{blue}{\Sigma_1}$};
        \node at (2.5,-2.1) {$\boldsymbol{U}_{k}^{-1}$};
        \node at (-1.45,.2) {$\boldsymbol{L}_{k}$};
        \node at (5,3.65) {$\boldsymbol{L}_{k}$};
        \node at (2,2.75) {$\textcolor{blue}{\Sigma_2}$};
        \node at (2.2,1.75) {$\boldsymbol{T}(e^{-ik\Omega})$};
        \node at (2.1,3.9) {$\boldsymbol{U}_{k}$};
        \node at (2.6,.5) {$\boldsymbol{U}_{k}$};
        \node at (-.25,0) {$\textcolor{blue}{\Gamma}$};
        \node at (1.1,0) {$e^{-ik\omega\sigma_3}$};
                \end{tikzpicture}
            }
            \caption{The jumps of $\boldsymbol{P}(z)$ for $x$ in the pole region. Here, if the jump matrix is upper triangular then its exponent is $-(2H+\Lambda)$, while if the jump matrix is lower triangular then its exponent is $2H+\Lambda$.}
            \label{jumpsOfPGenusOne}
        \end{figure}
        \item \textbf{Normalization:} As $z \to \infty$, the matrix $\boldsymbol{P}(z) = \I + \Oh\left(1/z\right)$.
    \end{enumerate}    
\end{RHP}

\subsection{Constructing the Outer Parametrix} 
In this section we explicitly construct the outer parametrix for the pole region. See \cite{FarFieldAsym} for  a similar construction. We start by introducing the Riemann-Hilbert problem corresponding to the outer parametrix. Next, we deform the Riemann-Hilbert problem so that both of its jumps on $\Sigma_1$ and $\Sigma_2$ are the twist matrix $\boldsymbol{T}(1)$. Then, we recall and implement the necessary machinery (the Abel map, the Riemann  $\Theta$-function, and  holomorphic differentials) in order to construct a $2 \times 2$ matrix-valued function whose jumps on $\Sigma_1 \cup \Sigma_2$ ultimately amounts to switching the columns. Finally, we use algebraic functions to complete our construction of the outer parametrix. 

\subsubsection{The Riemann-Hilbert Problem for the Outer Parametrix} Disregarding the jumps of $\boldsymbol{P}(z)$ that decay to the identity matrix we obtain the following Riemann-Hilbert problem:

\begin{RHP}\label{PdotRHP}
    For each $k \in \N$, find the $2 \times 2$ matrix-valued function $\dotbold{P}^{\out}(z;x,k) \equiv \dotbold{P}^{\out}(z)$ that satisfies the following properties:
    \begin{enumerate}
        \item \textbf{Analyticity}:
            $\dotbold{P}^{\out}(z)$ is analytic in the complex $z$-plane except on $\Sigma_1 \cup \Gamma \cup \Sigma_2$.
        \item \textbf{Jump Condition:}
            The jumps of $\dotbold{P}^{\out}(z)$ are topologically equivalent to those shown in Figure \ref{topologicallyEqu}.
            \begin{figure}[h]
                \centering
            \scalebox{1}
            {
                \begin{tikzpicture}
                    \draw[->, very thick] (-4.5cm,0) -- (-3cm,0) node[above=0cm] {$\begin{bmatrix}
                        0 & -1 \\
                        1 & 0
                    \end{bmatrix}$};
                    \draw[very thick] (-3cm,0) -- (-1.5cm,0) node[below=0cm]{$B$};
                    \draw[->, very thick] (-1.5,0) -- (0,0) node[above=0cm]{$\begin{bmatrix}
                        e^{-ik\omega} & 0 \\
                        0 & e^{i k\omega}
                    \end{bmatrix}$};
                    \draw[very thick] (0, 0) -- (1.5cm,0) node[below=0cm]{$C$};
                    \draw[->, very thick] (1.5cm,0) -- (3cm,0) node[above=0cm]{$\begin{bmatrix}
                        0 & - e^{ik\Omega} \\
                        e^{-ik\Omega} & 0
                    \end{bmatrix}$};
                    \draw[very thick] (3cm,0) -- (4.5cm,0);
                    \draw[fill=black] (-4.5,0) circle (.1cm) node[below=0cm]{$A$};
                    \draw[fill=black] (-1.5,0) circle (.1cm);
                    \draw[fill=black] (1.5,0) circle (.1cm);
                    \draw[fill=black] (4.5,0) circle (.1cm) node[below=0cm]{$D$};
                \end{tikzpicture}
            }
                \caption{A topologically equivalent version of the jump contours for $\dotbold{P}^{\out}(z)$. Here, $\omega,\Omega \in \R$, and $\Sigma_1$ corresponds to the line segment $AB$, $\Gamma$ corresponds to the line segment $BC$, and $\Sigma_2$ corresponds to the line segment $CD$.}
                \label{topologicallyEqu}
            \end{figure}
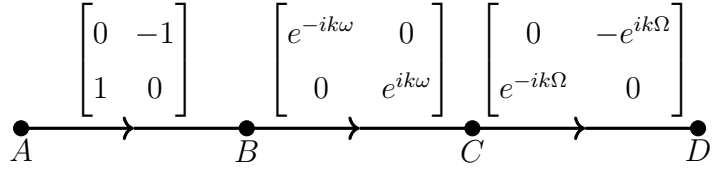
        \item \textbf{Normalization:} As $z \to \infty$, $\dotbold{P}^{\out}(z) = \I + \Oh(1/z)$. 
    \end{enumerate}
\end{RHP} 
Our next step is to move the dependence on $\omega$ and $\Omega$ from the jump matrices to the normalization condition. We consider
    \begin{equation}
        F(z;x) \equiv F(z) := \frac{R(z)}{2\pi i} \left( - \int\limits_{\Gamma}\frac{i\omega}{R(w)(w-z)} \dd w-\int\limits_{\Sigma_2}\frac{i\Omega}{R_{+}(w)(w-z)} \dd w \right).
    \end{equation}
Notice $F$ satisfies the following jumps:
\begin{itemize}
    \item For $z \in \Sigma_1$, $F_+(z) + F_-(z) =0$.
    \item For $z \in \Gamma$, $F_{+}(z) - F_{-}(z)=-i\omega$.  
    \item For $z \in \Sigma_2$, $F_{+}(z)+F_{-}(z)=-i\Omega$.
\end{itemize} 
Also, as $z \to \infty$, $F(z)= F_1z + F_0 + \Oh(1/z)$, where
\begin{equation}\label{F1Def}
    F_1 := \frac{1}{2\pi i} \left( \int\limits_{\Gamma}\frac{i\omega}{R(w)}\dd w +  \int\limits_{\Sigma_2}\frac{i\Omega}{R_{+}(w)} \dd w \right)
\end{equation}
and
\begin{equation}
    F_0 := \frac{1}{2\pi i} \left(\int\limits_{\Gamma}\frac{i\omega w}{R(w)} \dd w +  \int\limits_{\Sigma_2}\frac{i\Omega w}{R_{+}(w)} \dd w \right).
\end{equation}
Set
\begin{equation}\label{defOfQdot}
  \dotbold{Q}(z;x,k) \equiv \dotbold{Q}(z) := e^{kF_0\sigma_3}\dotbold{P}^{\out}(z) e^{-kF(z)\sigma_3}.
\end{equation}
Then $\dotbold{Q}$ solves the following Riemann-Hilbert problem:

\begin{RHP}\label{QdotRHP} For each $k \in \N$, find the $2 \times 2$ matrix-valued function $\dotbold{Q}(z;x,k) \equiv \dotbold{Q}(z)$ that satisfies the following properties:
        \begin{enumerate}
            \item \textbf{Analyticity}: 
                $\dotbold{Q}(z)$ is analytic in the complex $z$-plane except on $\Sigma_1 \cup \Sigma_2$.    
            \item \textbf{Jump Condition:} 
                The boundary values taken by $\dotbold{Q}$ on $\Sigma_1 \cup \Sigma_2$ are related by the jump condition $\dotbold{Q}_{+}= \dotbold{Q}_{-} \begin{bmatrix}
                    0 & -1 \\
                    1 & 0
                \end{bmatrix}$.
            \item \textbf{Normalization:} 
                As $z \to \infty$,\begin{equation}\label{QdotNormalization}
                    \dotbold{Q}(z)e^{kF_1z\sigma_3} = \I + \frac{\dotbold{Q}_{-1}}{z} + \frac{\dotbold{Q}_{-2}}{z^2} + \Oh(1/z^3).
                \end{equation}
        \end{enumerate}
\end{RHP}
Notice we eliminated the jump on $\Gamma$.
\subsubsection{The Baker-Akhiezer Map}
Recall the scalar function $R(z)$ is cut on $\Sigma_1 \cup \Sigma_2$. We consider $\T$, the genus-one Riemann surface defined by $R(z)$. For our purposes it is convenient to view $\T$ as two copies of the complex $z$-plane cut on $\Sigma_1 \cup \Sigma_2$. 
Given a complex point $z$, we relate the two copies or sheets by the identity $R\left(z^{+}\right) = -R\left(z^{-}\right)$. We say $z^+$ lives on the plus sheet of $\T$, while $z^{-}$ lives on the minus sheet of $\T$. In our discussion unless otherwise specified we assume the inputs of $\A(\cdot)$ are on the plus sheet. We now introduce a basis of homology cycles $\{\aloop, \bloop\}$ on $\T$ as shown in Figure \ref{homologyBasis}. Here integration on the second sheet is accomplished by replacing $R(z)$ by $-R(z)$.

\begin{figure}[h]
    \centering
\scalebox{1}
{
    \begin{tikzpicture}[]
        \draw[] (-4.5cm,0) -- (-3cm,0);
        \draw[] (-3cm,0) -- (-1.5cm,0) node[below=0cm]{$B$};
        
        \draw[] (1.5cm,0) -- (3cm,0);
        \draw[] (3cm,0) -- (4.5cm,0);
        \draw[fill=black] (-4.5,0) circle (.05cm) node[below=0cm]{$A$};
        \draw[fill=black] (-1.5,0) circle (.05cm);
        \draw[fill=black] (1.5,0) circle (.05cm) node[below=0cm]{$C$};
        \draw[fill=black] (4.5,0) circle (.05cm) node[below=0cm]{$D$};
        \draw[thick, dotted, red] (2.25cm,0) arc (0:180:2.25cm and 1.75cm);
        \draw[red, thick] (2.25cm,0) arc (0:-88:2.25cm and 1.75cm);
        \draw[red, thick] (-2.25cm,0) arc (180:269.5:2.25cm and 1.75cm);
        \draw[->, very thick, red] (.01cm, -1.75cm) -- (-.01cm,-1.75cm);
        \draw[->, very thick, red] (-.01cm, 1.75cm) -- (.01cm,1.75cm);

        \draw[thick, blue] (0,0) arc (0:88: 3cm and 1.75cm);
        \draw[thick, blue] (-6cm,0) arc (180:91: 3cm and 1.75cm);
        \draw[thick, blue] (0,0) arc (0: -89: 3cm and 1.75cm);
        \draw[thick, blue] (-6cm,0) arc (180: 268: 3cm and 1.75cm);
        \draw[<-,very thick, blue] (-3.05cm, 1.75cm) -- (-2.99cm, 1.75cm);
        \draw[->, very thick, blue] (-3.01cm,-1.75) -- (-2.95cm, -1.75cm);

        \node() at (-6.25cm, 0cm) {\Large$\aloop$};
        \node() at (2.4cm, -.35cm) {\Large $\bloop$};
        
    \end{tikzpicture}
}
    \caption{ The homology basis $\left\{ \aloop, \bloop \right\}$ of $\T$. Thick solid curves lie on the plus sheet while the dotted curve lies on the minus sheet.}
    \label{homologyBasis}
\end{figure}
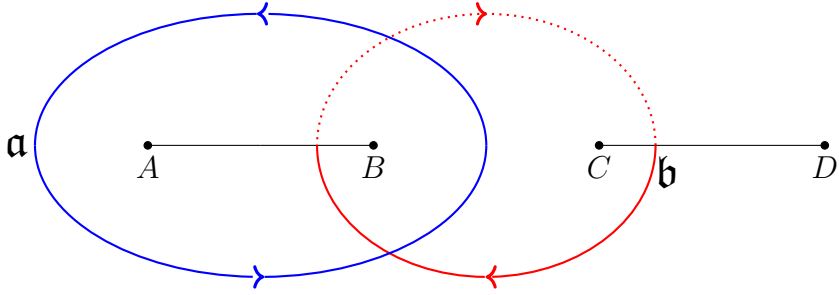

As $\{\aloop, \bloop\}$ is a homology basis, we see that any map defined via path integration on this genus-one surface will be well-defined modulo the periods defined by integrating along $\aloop$ and $\bloop$. In particular, in order for the Abel map,
\begin{equation}\label{AbelMapDef}
   \A(z;x) \equiv \A(z) := \frac{2\pi i}{\ointctrclockwise_\aloop \frac{\dd w}{R(w)}} \int_A^z \frac{\dd w}{R(w)},
\end{equation} to be well defined we mod out the value obtained from integrating along $\aloop$,
\begin{equation}
    \frac{2\pi i}{\ointctrclockwise_\aloop \frac{\dd w}{R(w)}} \ointctrclockwise_\aloop \frac{\dd w}{R(w)} = 2\pi i,
\end{equation}
and  the value obtained from integrating along $\bloop$,
\begin{equation}\label{defOfScriptedB}
    \B(x) \equiv \B:= \frac{2\pi i}{\ointctrclockwise_\aloop \frac{\dd w}{R(w)}} \ointclockwise_\bloop \frac{\dd w}{R(w)}.
\end{equation}
We now list the basic properties of $\A(z)$.
\begin{proposition}\label{jumps of Abel map prop}
    Given the lattice $\mathcal{L}:=\left\{j_{1} 2\pi i + j_{2} \B : j_{1},j_{2} \in \Z  \right\}$, the boundary values taken by the Abel map, $\A(z)$, on the cuts on the plus sheet are related by the jump conditions
    \begin{equation}
        \A_{+}\left( z^{+} \right) - \A_{-}\left( z^{+} \right)= 0 \mod \mathcal{L} \hspace{1 cm} \text{for } z^{+} \in \Sigma_1,
    \end{equation}
    \begin{equation}
        \A_{+}\left( z^{+} \right) - \A_{-}\left( z^{+} \right) = -2\pi i  \mod \mathcal{L} \hspace{1 cm} \text{for } z^{+} \in \Gamma,
    \end{equation}
    and
    \begin{equation}
        \A_{+}\left( z^{+} \right) - \A_{-}\left( z^{+} \right) =  - \B \mod \mathcal{L} \hspace{1 cm} \text{for } z^{+} \in \Sigma_2.
    \end{equation}
    Also, as $z \to \infty$ on the plus sheet,
    \begin{equation}\label{LargeZAbelMap}
        \A(z) = \A(\infty) + \frac{\A_{-1}}{z} + \Oh\left( 1/z^2 \right), \hspace{1cm} \text{where } \A_{-1}(x) \equiv \A_{-1} := - \frac{2\pi i}{\oint_{\aloop}\frac{\dd w}{R(w)}}.
    \end{equation}

\end{proposition}

Next, we use $\A(z)$ to construct a function (the Baker-Akhiezer function), defined on $\T$, that is path independent without modding. To do this, we employ Riemann $\Theta$-functions. Recall, given $M \in \Ci$, with $\Real(M)<0$, then for each $z \in \Ci$,
\begin{equation}\label{defOfRiemannThetaFunction}
    \Theta(z;M) \equiv \Theta(z):= \sum\limits_{k\in \Z} \exp\left(kz+\frac{1}{2}Mk^2 \right).
\end{equation}
It is well known (see \cite{NIST}) that $\Theta(z;M)$ satisfies the following properties:
\begin{itemize}
    \item 
        \textit{Periodicity:} $\Theta(z+2\pi i) = \Theta(z)$.
    \item 
        \textit{Quasiperiodicity:} $\Theta(z+M;M)= e^{-\frac{M}{2}} e^{-z} \Theta(z;M)$. More generally, for $n \in \Z$ $\Theta(z+nM;M)=e^{-Mn^2/2}e^{-nz}\Theta(z;M)$.
    \item 
        \textit{Evenness:} $\Theta(-z)=\Theta(z)$.
    \item
        All of the zeros of $\Theta$ are simple, and  $\Theta(z_0) =0$ if and only if there exist integers $j_1, j_2$  such that  
        \begin{equation}
            z_0= \left( \pi i +M/2 \right) + j_1 2\pi i + j_2 M.
        \end{equation}
\end{itemize}
We see that the $\Theta$-function gives us a handle on our goal of constructing a well-defined function on $\T$ from $\A(z)$. Naturally, we consider $\Theta(z;\B)$. Notice by using the $2\pi i$-periodicity of the $\Theta$-function, $\Theta\left(\A(z)\right)$ is independent of the number of $\aloop$-cycles used in the path that connects $A$ to $z$ (or any cycle in the same homology class as $\aloop$). We now develop the machinery that allows us to use the $\B$-quasiperiodicity of $\Theta$ to construct a function that is also independent of the number of $\bloop$-cycles used in the path that connects $A$ to $z$. We start with the observation that the $\B$-quasiperiodicity of $\Theta$ allows us to construct a function whose $\bloop$-cycle dependence is multiplicative, not additive. Indeed, for complex constants $C_1, C_2 \in \Ci$ and $n \in \Z$, notice
\begin{align}\label{quo}\begin{split}
    \frac{\Theta\left(\A(z)+C_1+C_2+n\B\right)}{\Theta\left(\A(z)+C_1+n\B\right)} &=\frac{\Theta\left(\A(z)+C_1+C_2\right)}{\Theta\left(\A(z)+C_1\right)} e^{-n C_2}.
\end{split}
\end{align}
We see that the above quotient ``counts'' (in the negative direction in the exponent with weight $C_2$) the number of $\bloop$-cycles used to connect $A$ to $z$. So, if we multiply it with a function that ``counts'' in the positive direction, then the product will be independent of the number of $\bloop$-cycles.

We define the differentials
\begin{equation}
    \Phi := \frac{2\pi i}{\ointctrclockwise_{\aloop}\frac{\dd w}{R(w)}} \frac{\dd w}{R(w)}, \hspace{.66 cm} \Upsilon_0 := \frac{w^2}{R(w)} \dd w, \hspace{.66cm} \text{ and } \hspace{.66cm} \Upsilon:= \Upsilon_0 - \left( \frac{1}{2\pi i} \ointctrclockwise_{\aloop} \Upsilon_0 \right)\Phi.
\end{equation}
Notice that $\Upsilon$ vanishes when it is integrated against an $\aloop$-cycle. We set the constant 
\begin{equation}\label{Udef}
    U(x) \equiv U := \ointclockwise_{\bloop} \Upsilon.
\end{equation}
We now list some of the basic properties of the differential $\Upsilon$. 
\begin{proposition}\label{jumps of diff prop}
    The boundary values taken by the map
    \begin{equation} \label{diffMap} 
        z \mapsto \int_{A}^{z} \Upsilon
    \end{equation}
    on $\Sigma_1\cup \Sigma_2$ in the plus-sheet are related by the jump conditions
    \begin{equation}
        \left( \int_{A}^{z^{+}} \Upsilon\right)_{+} + \left( \int_{A}^{z^{+}} \Upsilon \right)_{-} = 0 \mod U\Z \hspace{1cm} \text{for } z^{+} \in \Sigma_1
    \end{equation}
    and
    \begin{equation}
        \left( \int_{A}^{z^{+}} \Upsilon\right)_{+} + \left( \int_{A}^{z^{+}} \Upsilon \right)_{-} = -U \mod U\Z \hspace{1cm} \text{for } z^{+} \in \Sigma_2.
    \end{equation}
    Also, for some constants $\Upsilon_{0}, \Upsilon_{-1} \in \Ci$, as $z \to \infty$
    \begin{equation}\label{LargeZupsilon}
        z - \int_{A}^{z} \Upsilon = \Upsilon_{0} + \frac{\Upsilon_{-1}}{z} + \Oh\left( 1/z^2 \right).
    \end{equation}
\end{proposition}

Let $P$ be a path that connects $A$ to $z$ using exactly $n$ $\bloop$-cycles, for some $n \in \Z$. Notice
\begin{align}\label{diffMagic}
    \begin{split}
        \exp\left( kF_1 \int_{P}\Upsilon \right) &= \exp\left(kF_1 n \ointclockwise_{\bloop}\Upsilon \right) \exp\left( kF_1 \int_{A}^{z}\Upsilon \right) \\
        &= \exp(nkF_1U) \exp\left( kF_1 \int_{A}^{z}\Upsilon \right).
    \end{split}
\end{align}
Looking at \eqref{quo} and \eqref{diffMagic}, we see that if we choose $C_2 =kF_1U$ and if the paths in the Abel map and the differential are consistent, then their product is path independent on $\T$. 

We also need to specify the constant $C_1$. As we are trying to solve a Riemann-Hilbert problem with no poles, we need to control where the denominator of \eqref{quo} vanishes. Thus, we set
\begin{equation}\label{RiemannConstant}
    \K(x) \equiv \K := i\pi + \frac{\B}{2}
\end{equation}
to be the \textit{Riemann constant} and
\begin{equation}
    C_1 := - \A(Q) - \K,
\end{equation} 
where $Q \in \Ci$ is specified in \eqref{zero of fDia}. Now, the Baker-Akhiezer function is given as
\begin{equation} \label{genusOneBakerAkhiezer} 
    q(z;x,k) \equiv q(z) :=
\frac{\Theta\left(\A(z) -\A(Q) - \K - kF_1U
\right)}{\Theta\left(\A(z)-\A(Q) - \K \right)} e^{-kF_1 \int_A^z\Upsilon}.
\end{equation}
Notice $q$ has a simple pole at $z=Q$. 

\subsubsection{The Solution to the Model Riemann-Hilbert Problem}

We start with the observation that by permuting the signs of $C_1$ and $F_1U$ in $q(z)$ we can construct a matrix-valued function that is analytic in $\Ci \setminus (\Sigma_1 \cup \Sigma_2)$ with a simple jump discontinuity. Define
\begin{multline}
    \dotbold{R}(z;x,k) \equiv \dotbold{R}(z) := \\
    \begin{bmatrix}
        \displaystyle \frac{\Theta(\A(z)+\A(Q)+\K-kF_1U)}{\Theta(\A(z)+\A(Q)+\K)}e^{-kF_1\int_{A}^{z}\Upsilon} & \displaystyle \frac{\Theta(\A(z)-\A(Q)-\K+kF_1U)}{\Theta(\A(z)-\A(Q)-\K)}e^{kF_1\int_{A}^{z}\Upsilon} \\        
        \displaystyle \frac{\Theta(\A(z)-\A(Q)-\K-kF_1U)}{\Theta(\A(z)-\A(Q)-\K)}e^{-kF_1\int_{A}^{z}\Upsilon} & \displaystyle \frac{\Theta(\A(z)+\A(Q)+\K+kF_1U)}{\Theta(\A(z)+\A(Q)+\K)}e^{kF_1\int_{A}^{z}\Upsilon}           
        \end{bmatrix}.
\end{multline}
The  jump conditions in Propositions \ref{jumps of Abel map prop} and \ref{jumps of diff prop} imply 
\begin{equation}
    \dotbold{R}_{+}(z) = \dotbold{R}_{-}(z) \begin{bmatrix}
        0 & 1 \\
        1 & 0 
    \end{bmatrix} \hspace{1cm} \text{ for } z \in \Sigma_1 \cup \Sigma_2.
\end{equation}

We will construct $\dotbold{Q}(z)$ out of $\dotbold{R}(z)$. Our first step is to introduce a sign change into the jump of $\dotbold{R}$. In order to do so, we define
\begin{equation}
    \gamma(z) := \left(\frac{(z-A)(z-C)}{(z-B)(z-D)}\right)^{1/4},
\end{equation}
where $\gamma$ is cut on $\Sigma_1 \cup \Sigma_2$ and as $z \to \infty$, $\gamma(z)=1 + \Oh(1/z)$. We also define
\begin{equation}
    f^{\text{D}}(z) := \frac{\gamma(z)+\gamma^{-1}(z)}{2} \hspace{1.75cm} \text{ and } \hspace{1.75 cm} f^{OD}(z):= \frac{\gamma(z)-\gamma^{-1}(z)}{2i}.
\end{equation}
A direct calculation shows that $\gamma_{+}=\gamma_{-}e^{-i\pi/2}$ on $\Sigma_1 \cup \Sigma_2$. Hence,
\begin{equation}
    f^{D}_{+}(z) = f^{OD}_{-}(z) \hspace{1 cm} \text{and} \hspace{1 cm} f^{OD}_{+}(z) = - f^{D}_{-} (z) \hspace{1 cm} \text{ for } z \in \Sigma_1 \cup \Sigma_2.
\end{equation}
Also, a straightforward computation shows that
\begin{equation}\label{zero of fDia}
    Q(x)\equiv Q := \frac{BD-AC}{B+D-A-C}
\end{equation}
is the unique solution to the equation $f^{\text{D}}(z)f^{\text{OD}}(z)=0$. Further, it is a simple zero and either $f^{\text{D}}(Q)=0$ or $f^{\text{OD}}(Q)=0$. It has been shown in \cite{RationalPainleveIII} that while the construction of the parametrix does depend on whether $Q$ is a zero of $f^{\text{D}}$ or $f^{\text{OD}}$, the final answer (i.e.\ the formula we extract from $\dotbold{P}^{\out}(z)$; see \eqref{genusOneResult}) is the same regardless of which function vanishes at $Q$. Numerically, we see that $f^{\text{OD}}(Q)=0$ for various $x$-values in the pole region. Thus, we assume that $Q$ is the unique simple zero of $f^{\text{OD}}(z)$, and we proceed accordingly.

We now modify $\dotbold{R}(z)$ by setting
\begin{equation}
    \dotbold{Q}(z;x,k) \equiv \dotbold{Q}(z) =
    \begin{bmatrix}
        C_{11} & 0 \\
        0 & C_{22}
    \end{bmatrix} 
    \begin{bmatrix}
        f^{\text{D}}(z) [\dotbold{R}(z)]_{11}&
        f^{\text{OD}}(z) [\dotbold{R}(z)]_{12} \\
        -f^{\text{OD}}(z) [\dotbold{R}(z)]_{21}&
        f^{\text{D}}(z) [\dotbold{R}(z)]_{22}
    \end{bmatrix},
\end{equation}
where $C_{11}$ and $C_{22}$ are normalizing constants specified in \eqref{defOfC11} and \eqref{defOfC22}. Now, $\dotbold{Q}(z)$ satisfies the jump condition $\dotbold{Q}_{+}(z) = \dotbold{Q}_{-}(z) \boldsymbol{T}(1)$ for $z \in z \in \Sigma_1 \cup \Sigma_2$. 

It is well known that the Abel map $\A(z)$ with domain $\T$ is injective
(see for example the discussion in Section 4.5.2 in \cite{RationalPainleveII}). Since $\A(Q^{-}) + \A(Q) =0$ (here $Q$ is on the plus sheet and $Q^{-}$ is the equivalent point on the minus sheet), it follows that for each $z$ on the plus sheet, $\A(z)+\A(Q) \ne 0$, otherwise $\A(z)$ would not be injective.
Hence, the diagonal entries of $\dotbold{R}(z)$ are analytic in $\Ci \setminus (\Sigma_1 \cup \Sigma_2)$. Further, $Q$ was selected so that in the off-diagonal entries of $\dotbold{R}(z)$ the simple pole $z=Q$ is canceled by the simple zero of $f^{\text{OD}}(z)$. Therefore, the off-diagonal entries of $\dotbold{Q}(z)$ are analytic in $\Ci \setminus (\Sigma_1 \cup \Sigma_2)$. 

Finally, to deduce the values of $C_{11}, C_{22}$, we notice as  $z \to \infty$, $f^{\text{D}}(z) = 1 + \Oh(1/z)$ and $f^{\text{OD}}(z)= \Oh(1/z)$, thus we find the limits of $[\dotbold{R}]_{11}$ and $[\dotbold{R}]_{22}$ as $z \to \infty$ and take the reciprocal to find
\begin{equation}\label{defOfC11}
    C_{11}(x;k) \equiv C_{11} := \frac{\Theta\left(\A(\infty)+\A(Q)+\K\right)}{\Theta\left(\A(\infty)+\A(Q)+\K-kF_1U\right)} e^{-kF_1\Upsilon_{0}}
\end{equation}
and 
\begin{equation}\label{defOfC22}
    C_{22}(x;k) \equiv C_{22} := \frac{\Theta\left(\A(\infty)+\A(Q)+\K\right)}{\Theta\left(\A(\infty)+\A(Q)+\K+kF_1U\right)} e^{kF_1\Upsilon_{0}}.
\end{equation}
Here, $C_{11}$ and $C_{22}$ are not guaranteed to exist. Indeed, their denominators can vanish. In other words, $C_{11}$ and $C_{22}$ do not exist if and only if
\begin{equation}\label{poleCondition}
   \A(\infty) + \A(Q) + \K= \pm kF_{-1}U \mod \left\{ \K + j_{1}2\pi i + j_{2} \B : j_{1}, j_{2} \in \Z \right\}.
\end{equation} 
When $k$ is large enough, if the latter equality holds for some $x$ in the pole region, then the left-hand side of \eqref{genus0 theorem} has a pole near $- k^{2/3} 2^{1/3} x$.

For each $k \in \N$, let $\mathscr{P}_k$ denote all of the solutions in the $x$-plane of \eqref{poleCondition}. Notice, $\mathscr{P}_k$ is the set of exceptional points for which Riemann-Hilbert Problem \ref{PdotRHP} does not admit a solution. Now, let $\delta>0$ be a fixed small number. For each $k \in \N$ let
\begin{equation}\label{SwissCheese}
    \mathscr{S}_{k} := \left\{ x \text{ in the pole region } : \dist(x,\mathscr{P}_k)> \frac{\delta}{k^{2/3}} \right\}.
\end{equation}
That is, $\mathscr{S}_{k}$ is the result of excising a neighborhood of each exceptional point from the pole region. 

We have completed the construction of $\dotbold{Q}(z)$, and thus $\dotbold{P}^{\out}(z)$ with the help of \eqref{defOfQdot}.

\subsection{Constructing the Inner Parametrices}\label{innerParaSection}
In this section, we construct the inner parametrices around the points $A$, $B$, $C$, and $D$. In particular, we go over in detail the construction of the local parametrix at $z=C$. For the other three we just give the resulting formulae.  Our approach is modeled from the methodology presented in Section 7.6.2 in \cite{RationalPainleveIV}.

\subsubsection{Local Parametrix at $z=C$} 
Our first step is to massage the jumps of $\boldsymbol{P}(z)$ so that they match the jumps of $\boldsymbol{A}\left( k^{2/3} \tau_{C}(z) \right)$, where $\boldsymbol{A}(\cdot)$ is the solution to Riemann-Hilbert Problem \ref{AiryRHP}. We give a systematic approach to do this.

\begin{enumerate}
    \item \textit{Clear the $z$-dependence.} Multiply (on the right) by $e^{-k(H(z)+\Lambda/2)\sigma_3}$ to obtain the ``core'' jumps of $\boldsymbol{P}(z)$.
    \item \textit{Mold the core jumps into the shape of the Airy jump matrices.} In this case, we move (by sectionally analytic substitutions) $\Sigma_2$ onto $\Gamma$. This deformation flips the triangularity of the jump matrix whose jump contour is contained in $\D_{C}^{\downarrow}$.
    \item \textit{Insert $\tau$ dependence.} Multiply by $e^{\frac{k}{2}\tau_{C}^{3/2}\sigma_3}$. Here, $\tau_{C}(z)^{3/2}$ is cut on $\widetilde{\Gamma}$ with $\left( \tau_{C}(z)^{3/2} \right)_{+}= - \left( \tau_{C}(z)^{3/2} \right)_{-}$.
\end{enumerate}
In Figure \ref{trackingAiry}  we track the jumps of $\boldsymbol{P}(z)$ as we apply the above procedure. Following it yields
\begin{equation}\label{AiryKey}
    \boldsymbol{Q}^{(C)}(z;x,k) \equiv \boldsymbol{Q}^{(C)}(z) := 
    \displaystyle\begin{cases}
        \boldsymbol{P}(z) e^{-k(H(z)+\Lambda/2)\sigma_3}  e^{\frac{k}{2} \tau_{C}(z)^{3/2}\sigma_3}, & z\in \D_{C}^{\uparrow}, \\
        \boldsymbol{P}(z) e^{-k(H(z)+\Lambda/2)\sigma_3} \boldsymbol{T}(1) e^{\frac{k}{2} \tau_{C}(z)^{3/2}\sigma_3}, & z\in \D_{C}^{\downarrow}.
    \end{cases}
\end{equation}

\begin{figure}[ht]
    \centering
    \begin{tabular}{ccc}
        \begin{tikzpicture}
            \draw[rotate=-108.313,thick, blue] (0,0) -- (0:2.75cm);
            \draw[rotate=-108.313, very thick,->,blue] (0:2.cm) -- (0:2.01cm);
            \node at (-1.3,-2.2) {$e^{ik\omega\sigma_3}$};
            \draw[rotate=-108.313,thick,blue] (0,0) -- (120:2.75cm);
            \draw[ rotate=-108.313, very thick,->,blue] (120:2.cm) -- (120:2.01cm);
            \node at (1.9,0) {$\boldsymbol{T}(e^{-ik\Omega})$};
        
            \draw[rotate=-108.313,thick] (0,0) -- (60:2.75cm);
            \draw[rotate=-108.313, very thick,->] (60:2.cm) -- (60:2.01cm);
            \node at (.85,-1.95) {$\boldsymbol{U}_k(-\chi)$};
        
            \draw[rotate=-108.313,thick] (0,0) -- (-60:2.75cm);
            \draw[ rotate=-108.313, very thick,->] (-60:2.cm) -- (-60:2.01cm);
            \node at (-2.05,0) {$\boldsymbol{L}_{k}^{-1}(\chi)$};
        
            \draw[rotate=-108.313,thick] (0,0) -- (180:2.75cm);
            \draw[ rotate=-108.313, very thick,->] (180:2.cm) -- (180:2.01cm);
            \node at (1.5,1.9) {$\boldsymbol{U}_{k}(-\chi)$};
        
            \draw[fill=black] (0,0) circle (.075cm);
            \node at (-.25,.25) {$C$};
    
            \path (-2.25,2.25) node [shape=rectangle, draw] {$z$};
        
        \end{tikzpicture} &
        \textcolor{white}{fill me up}&
        \begin{tikzpicture}
            \draw[rotate=-108.313,thick,blue] (0,0) -- (120:2.75cm);
            \draw[ rotate=-108.313, very thick,->,blue] (120:2.cm) -- (120:2.01cm);
            \node at (2.25,.0) {$\boldsymbol{T}(1)$};
        
            \draw[rotate=-108.313,thick] (0,0) -- (60:2.75cm);
            \draw[rotate=-108.313, very thick,->] (60:2.cm) -- (60:2.01cm);
            \node at (1.33,-2) {$\boldsymbol{U}_0$};
        
            \draw[rotate=-108.313,thick] (0,0) -- (-60:2.75cm);
            \draw[ rotate=-108.313, very thick,->] (-60:2.cm) -- (-60:2.01cm);
            \node at (-2.2,-.1) {$\boldsymbol{L}_{0}^{-1}$};
        
            \draw[rotate=-108.313,thick] (0,0) -- (180:2.75cm);
            \draw[ rotate=-108.313, very thick,->] (180:2.cm) -- (180:2.01cm);
            \node at (1.1,2.1) {$\boldsymbol{U}_{0}$};
        
            \draw[fill=black] (0,0) circle (.075cm);
            \node at (-.25,.25) {$C$};
    
            \path (-2.25,2.25) node [shape=rectangle, draw] {$z$};
            \node at (2.25,2.5) {$(1)$};
        
        \end{tikzpicture} \\
        \textcolor{white}{s} & \textcolor{white}{s} & \textcolor{white}{s} \\
        \begin{tikzpicture}
            \draw[rotate=-108.313,thick,blue] (0,0) -- (0:2.75cm);
            \draw[ rotate=-108.313, very thick,->,blue] (0:2.cm) -- (0:2.01cm);
            \node at (-.1,-2.2) {$\boldsymbol{T}(1)$};

            \draw[rotate=-108.313,thick] (0,0) -- (60:2.75cm);
            \draw[rotate=-108.313, very thick,->] (60:2.cm) -- (60:2.01cm);
            \node at (1.2,-2) {$\boldsymbol{L}_0^{-1}$};
        
            \draw[rotate=-108.313,thick] (0,0) -- (-60:2.75cm);
            \draw[ rotate=-108.313, very thick,->] (-60:2.cm) -- (-60:2.01cm);
            \node at (-2.2,-.1) {$\boldsymbol{L}_{0}^{-1}$};
        
            \draw[rotate=-108.313,thick] (0,0) -- (180:2.75cm);
            \draw[ rotate=-108.313, very thick,->] (180:2.cm) -- (180:2.01cm);
            \node at (1.,2.1) {$\boldsymbol{U}_{0}$};
        
            \draw[fill=black] (0,0) circle (.075cm);
            \node at (-.25,.25) {$C$};
    
            \path (-2.25,2.25) node [shape=rectangle, draw] {$z$};
            \node at (2.25,2.5) {$(2)$};
        
        \end{tikzpicture} &
        \textcolor{white}{fill me up}&
        \begin{tikzpicture}
            \draw[thick] (0,0) -- (0:2.75cm);
            \draw[->,very thick] (0:2.cm) -- (0:2.01cm);
            \node at (2,-.33) {$\boldsymbol{U}_{k}(-\tau_C^{3/2})$};
    
            \draw[thick] (0,0) -- (180:2.75cm);
            \draw[->,very thick] (180:2.cm) -- (180:2.01cm);
            \node at (-2.35,-.3) {$\boldsymbol{T}(1)$};
    
            \draw[thick] (0,0) -- (120: 2.75cm);
            \draw[->,very thick] (120:2.cm) -- (120:2.01cm);
            \node at (0, 2) {$\boldsymbol{L}_{k}^{-1}(\tau_{C}^{3/2})$};

            \draw[thick] (0,0) -- (-120: 2.75cm);
            \draw[->,very thick] (-120:2.cm) -- (-120:2.01cm);
            \node at (0, -2) {$\boldsymbol{L}_{k}^{-1}(\tau_{C}^{3/2})$};
    
            \draw[fill=black] (0,0) circle (.075cm);
            \node at (.25,-.25) {$0$}; 
    
            \path (-2.25,2.25) node [shape=rectangle, draw] {$\tau_C$};
            \node at (2.25,2.5) {$(3)$};
        \end{tikzpicture}
    \end{tabular}
    \caption{The top-left panel shows the jumps of $\boldsymbol{P}(z)$ in $\D_C$, where $\chi:=\Real(2H+\Lambda)$. The top-right panel shows the jumps of $\boldsymbol{P}(z)e^{-k(H(z)+\Lambda/2)\sigma_3}$. The bottom-left panel shows the jumps of $\boldsymbol{Q}^{(C)}(z)e^{-\frac{k}{2}\tau_{C}(z)^{3/2}}$. The bottom-right panel shows the jumps of $\boldsymbol{Q}^{(C)}(z)$ in the $\tau_{C}$-plane.}\label{trackingAiry}
\end{figure}
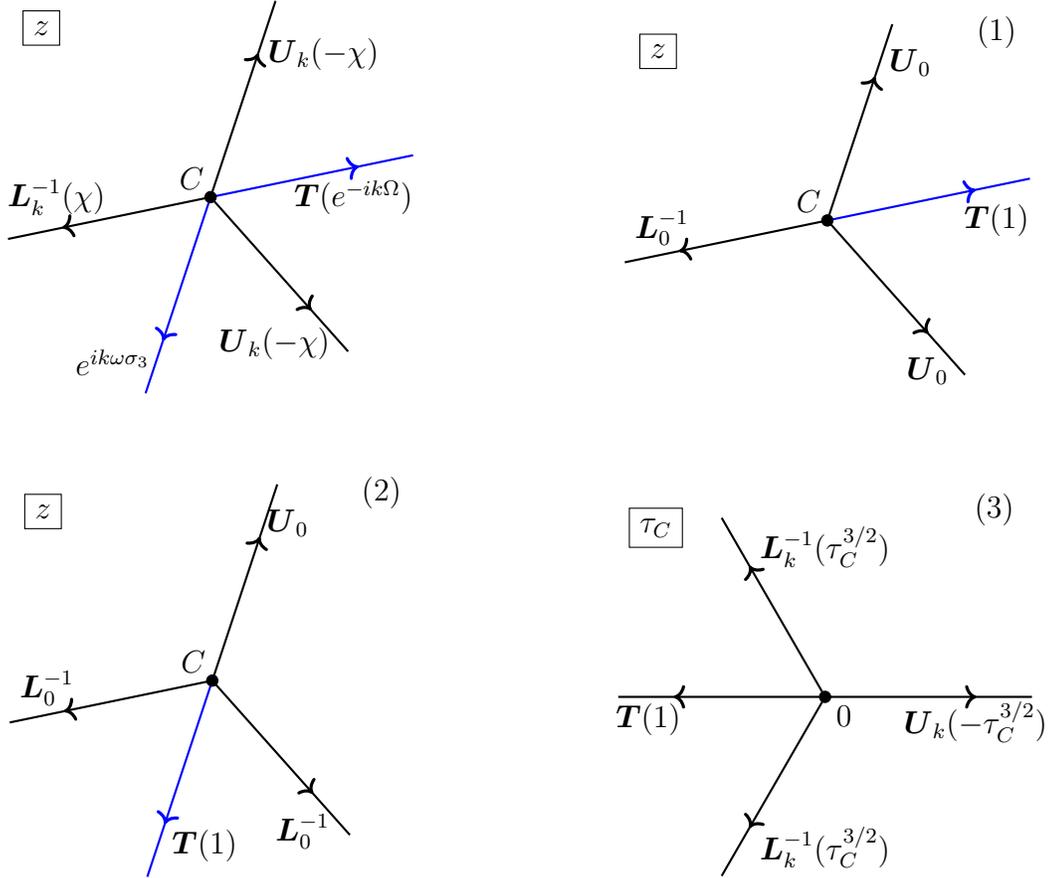
Comparing the jumps of $\boldsymbol{Q}^{(C)}(z)$ to the jumps of $\boldsymbol{A}\left( k^{2/3} \tau_{C}(z)\right)$ (see Appendix \ref{AiryAppendix}) we see that they agree. Notice, if we replace $\boldsymbol{P}(z)$ with $\dotbold{P}^{\out}(z)$ in \eqref{AiryKey}, then the resulting function will only  have a jump of $\boldsymbol{T}(1)$  on $\widetilde{\Gamma}$. Set \begin{equation}\label{AnalyticPrefactorz=C}
    \boldsymbol{H}^{(C)}(z):= \displaystyle\begin{cases}
        \dotbold{P}^{\out}(z) e^{-k(H(z)+\Lambda/2)\sigma_3}  e^{\frac{k}{2} \tau_{C}(z)^{3/2}\sigma_3} \boldsymbol{V}^{-1} \tau_{C}^{-\sigma_{3}/4}, & z\in \D_{C}^{\uparrow}, \\
        \dotbold{P}^{\out}(z) e^{-k(H(z)+\Lambda/2)\sigma_3} \boldsymbol{T}(1) e^{\frac{k}{2} \tau_{C}(z)^{3/2}\sigma_3}\boldsymbol{V}^{-1} \tau_{C}^{-\sigma_{3}/4}, & z\in \D_{C}^{\downarrow},
    \end{cases}
\end{equation}
where $\boldsymbol{V}$ is defined in \eqref{VMatDef}. Seeing as $\tau_C(z)^{-\sigma_3/4}\boldsymbol{V}$ also has a jump of $\boldsymbol{T}(1)$ on $\widetilde{\Gamma}$ (because $\tau_C$ maps $\widetilde{\Gamma}$ to the negative real numbers and the $3/2$-root function is principal), it follows that $\boldsymbol{H}^{(C)}(z)$ is continuous across $\widetilde{\Gamma}$. Further, as $z \to C$, the largest blowup in \eqref{AnalyticPrefactorz=C} is a square-root blowup (coming from the quarter-root blowups of $\dotbold{P}^{\out}(z)$ and $\tau_C(z)^{-\sigma_3/4}$). This implies that $\boldsymbol{H}^{(C)}(z)$ is analytic in $\D_C$ because  $\boldsymbol{H}^{(C)}(z)$ has no square-root jumps in $\D_C$ as it is continuous in $\D_C$.

Now, since $\boldsymbol{A}\left( k^{2/3} \tau_{C}(z) \right)$ and $\boldsymbol{Q}^{(C)}(z)$ satisfy the same jump conditions, and $\boldsymbol{H}^{(C)}$ is analytic in $\D_C$, it follows that $\boldsymbol{H}^{(C)}(z)k^{-\sigma_{3}/6}\boldsymbol{A}\left( k^{2/3} \tau_{C}(z)\right) \boldsymbol{Q}^{(C)}(z)^{-1}$ is analytic in $\D_C$. This in turn implies
\begin{equation}
    \dotbold{P}^{(C)}(z;x,k) \equiv \dotbold{P}^{(C)}(z) := \boldsymbol{H}^{(C)}(z)k^{-\sigma_{3}/6}\boldsymbol{A}\left( k^{2/3} \tau_{C}(z)\right) \boldsymbol{Q}^{(C)}(z)^{-1} \boldsymbol{P}(z)
\end{equation}
satisfies the same jump conditions as $\boldsymbol{P}(z)$. The upshot is that we have an explicit formula for $\dotbold{P}^{(C)}(z)$ because we can solve for $\boldsymbol{Q}^{(C)}(z)^{-1} \boldsymbol{P}(z)$ in \eqref{AiryKey}. Indeed,
\begin{equation}
    \dotbold{P}^{(C)}(z) = \begin{cases}
        \boldsymbol{H}^{(C)}(z)k^{-\sigma_{3}/6}\boldsymbol{A}\left( k^{2/3} \tau_{C}(z)\right) e^{-\frac{k}{2} \tau_{C}(z)^{3/2}\sigma_3} e^{k(H(z)+\Lambda/2)\sigma_3}, & z \in \D_{C}^{\uparrow}, \\
        \boldsymbol{H}^{(C)}(z)k^{-\sigma_{3}/6}\boldsymbol{A}\left( k^{2/3} \tau_{C}(z)\right) e^{-\frac{k}{2} \tau_{C}(z)^{3/2}\sigma_3} \boldsymbol{T}(-1) e^{k(H(z)+\Lambda/2)\sigma_3}, & z \in \D_{C}^{\downarrow}.      
    \end{cases}
\end{equation}
Also, we can express $\dotbold{P}^{\out}(z)^{-1}$ in terms of $\boldsymbol{H}^{(C)}(z)$ to find that 
\begin{equation}
    \dotbold{P}^{(C)}(z) \dotbold{P}^{\out}(z)^{-1}= \boldsymbol{H}^{(C)}(z) k^{-\sigma_3/6} \left( \boldsymbol{A}\left(k^{2/3} \tau_C  \right)\boldsymbol{V}^{-1}\left( k^{2/3} \tau_C(z)\right)^{-\sigma_3/4} \right) k^{\sigma_3/6} \boldsymbol{H}^{(C)}(z)^{-1}.
\end{equation}
Therefore, using the normalization of $\boldsymbol{A}$  (see \eqref{AiryNormalization}), on $\partial \D_C$, as $k \to \infty$,
\begin{equation}
    \dotbold{P}^{(C)}(z) \dotbold{P}^{\out}(z)^{-1} = \I +\begin{bmatrix}
        \Oh\left( k^{-2} \right) & \Oh\left( k^{-1} \right) \\
        \Oh\left( k^{-1} \right) & \Oh\left( k^{-2} \right)
    \end{bmatrix}.
\end{equation}

\subsubsection{Local Parametrices at the Endpoints $z=A$, $z=B$, and $z=D$}

We now construct the other parametrices. In $\D_A$, after clearing the $z$-dependence, the core jumps of $\boldsymbol{P}(z)$ are already in the correct shape. Thus, set
\begin{equation}
    \boldsymbol{H}^{(A)}(z) := \dotbold{P}^{\out}(z) e^{-k(H(z)+\Lambda/2)\sigma_3} e^{\frac{k}{2}\tau_{A}(z)^{3/2}\sigma_3}\boldsymbol{V}^{-1} \tau_{A}(z)^{-\sigma_3/4}
\end{equation}
and
\begin{equation}
    \dotbold{P}^{(A)}(z) := \boldsymbol{H}^{(A)}(z)k^{-\sigma_{3}/6}\boldsymbol{A}\left( k^{2/3} \tau_{A}(z)\right) e^{-\frac{k}{2} \tau_{A}(z)^{3/2}\sigma_3} e^{k(H(z)+\Lambda/2)\sigma_3}.
\end{equation}

In $\D_B$, after clearing the $z$-dependence, the core jumps of $\boldsymbol{P}(z)$ are all off by a sign in the off-diagonal entries. So, set
\begin{equation}
    \boldsymbol{H}^{(B)}(z) := \dotbold{P}^{\out}(z) e^{-k(H(z)+\Lambda/2)\sigma_3} e^{\frac{i\pi}{2}\sigma_3} e^{\frac{k}{2}\tau_{B}(z)^{3/2}\sigma_3}\boldsymbol{V}^{-1} \tau_{B}(z)^{-\sigma_3/4}
\end{equation}
and
\begin{equation}
    \dotbold{P}^{(B)}(z) := \boldsymbol{H}^{(B)}(z)k^{-\sigma_{3}/6}\boldsymbol{A}\left( k^{2/3} \tau_{B}(z)\right) e^{-\frac{k}{2} \tau_{B}(z)^{3/2}\sigma_3} e^{-\frac{i\pi}{2}\sigma_3} e^{k(H(z)+\Lambda/2)\sigma_3}.
\end{equation}

In $\D_D$, after clearing the $z$-dependence, the core triangular jumps of $\boldsymbol{P}(z)$ are all the wrong triangularity. Hence, set
\begin{equation}
    \boldsymbol{H}^{(D)}(z) := \dotbold{P}^{\out}(z) e^{-k(H(z)+\Lambda/2)\sigma_3}\boldsymbol{T}(1) e^{\frac{i\pi}{2}\sigma_3} e^{\frac{k}{2}\tau_{D}(z)^{3/2}\sigma_3}\boldsymbol{V}^{-1} \tau_{D}(z)^{-\sigma_3/4}
\end{equation}
and
\begin{equation}
    \dotbold{P}^{(D)}(z) := \boldsymbol{H}^{(D)}(z)k^{-\sigma_{3}/6}\boldsymbol{A}\left( k^{2/3} \tau_{D}(z)\right) e^{-\frac{k}{2} \tau_{D}(z)^{3/2}\sigma_3} e^{-\frac{i\pi}{2}\sigma_3} \boldsymbol{T}(-1) e^{k(H(z)+\Lambda/2)\sigma_3}.
\end{equation}
Further, for each endpoint the local parametrix matches well with the outer parametrix. That is, for $p \in \left\{ A,B,C,D\right\}$, for each $z \in \partial \D_p$, as $k \to \infty$, the mismatch jumps satisfy
\begin{equation}
   \dotbold{P}^{(p)}(z) \dotbold{P}^{\out}(z)^{-1} = \I +\begin{bmatrix}
        \Oh\left( k^{-2} \right) & \Oh\left( k^{-1} \right) \\
        \Oh\left( k^{-1} \right) & \Oh\left( k^{-2} \right)
    \end{bmatrix}.
\end{equation}

\subsection{Error Analysis}
In this section we quantify the error introduced by omitting the triangular jumps that decay to the identity.

\subsubsection{The Global Parametrix} We are now ready to define the global parametrix. In $\Ci \setminus (\Sigma_1 \cup \Gamma \cup \Sigma_2)$ set
\begin{equation}
   \dotbold{P}(z;x,k) \equiv \dotbold{P}(z) = \begin{cases}
        \dotbold{P}^{(A)}(z), & z \in \D_{A}, \\
        \dotbold{P}^{(B)}(z), & z \in \D_{B}, \\
        \dotbold{P}^{(C)}(z), & z \in \D_{C}, \\
        \dotbold{P}^{(D)}(z), & z \in \D_{D}, \\
        \dotbold{P}^{\out}(z),& \text{otherwise}.
    \end{cases}
\end{equation}
Notice, for $z$ large enough, $\dotbold{P}(z) = \dotbold{P}^{\out}(z)$ and we can use the normalization condition of $\dotbold{P}^{\out}(z)$ to write
\begin{equation}\label{LargezExpanofPDotg1}
    \dotbold{P}(z) = \I + \frac{\dotbold{P}_{-1}}{z} + \frac{\dotbold{P}_{-2}}{z^2} + \Oh(1/z^3), \hspace{1cm} z \to \infty.
\end{equation}
We introduce the error as
\begin{equation}
    \boldsymbol{E}(z;x,k) \equiv \boldsymbol{E}(z) := \boldsymbol{P}(z) \dotbold{P}^{-1}(z).
\end{equation}
As in the pole-free case, when $z$ is not on a jump contour of $\boldsymbol{P}(z)$ or $\dotbold{P}(z)$ then $\boldsymbol{E}(z)$ is the product of analytic functions and thus is analytic. Further, for each $p \in \left\{ A, B ,C ,D \right\}$, $\boldsymbol{P}(z)$ and $\dotbold{P}(z)$ have exactly the same jump conditions inside $\D_p$, and thus $\boldsymbol{E}(z)$ is analytic in $\D_p$. Similarly, $\boldsymbol{P}(z)$ and $\dotbold{P}(z)$ satisfying the same jump conditions on $\Sigma_1 \cup \Gamma \cup \Sigma_2$ implies that $\boldsymbol{E}(z)$ is analytic on $\Sigma_1 \cup \Gamma \cup \Sigma_2$. Therefore, the only jumps of $\boldsymbol{E}(z)$ are the boundary of each disk and the jump contours of $\boldsymbol{P}(z)$ that are bounded away from the endpoints $A$, $B$, $C$, and $D$.

\begin{RHP}
    For each $k \in \N$, find $\boldsymbol{E}(z)$ so that 
    \begin{enumerate}
        \item \textbf{Analyticity}:
            $\boldsymbol{E}(z)$ is analytic in $z$ except along the jump contours described above.
        \item \textbf{Jump Condition:}
            $\boldsymbol{E}(z)$ can be continuously extended to the boundary and the boundary values taken by $\boldsymbol{E}(z)$ are related by the jump condition $\boldsymbol{E}_+(z)=\boldsymbol{E}_-(z)\boldsymbol{V}^{\boldsymbol{(E)}}(z)$, where
            \begin{equation}
                \boldsymbol{V}^{\boldsymbol{(E)}}(z) =
                \begin{cases}
                    \dotbold{P}^{(p)}(z) \dotbold{P}^{\out}(z)^{-1}, & z \in \partial \D_{p} \text{ for } p \in \{A,B,C,D\}, \\
                    \dotbold{P}^{\out}(z) \boldsymbol{V}^{(\boldsymbol{P})}(z) \dotbold{P}^{\out}(z)^{-1}, & \text{otherwise}.
                \end{cases}
            \end{equation}    
        \item \textbf{Normalization:} As $z \to \infty$, the matrix $\boldsymbol{E}(z)= \I + \Oh\left(1/z\right)$.    
    \end{enumerate}
\end{RHP}
Just like in the pole-free case, since $\boldsymbol{E}(z)$ normalizes to the identity and all of its jump matrices are close to the identity matrix we have the following proposition:
\begin{proposition}\label{error in pole}
    Let $x$ be in the pole region of the $x$-plane. For large $k \in \N$, if $x \in \mathscr{S}_k$ then 
    \begin{equation}\label{Error Result genus1}
        \boldsymbol{E}(z) = \I + \Oh(1/k).
    \end{equation}
\end{proposition}

\subsubsection{Extraction Formula} The formula \eqref{pk asymptotic Formula in g0} is still valid in the pole region given that $x \in \mathscr{S}_k$. Moreover, Proposition \ref{error in pole} implies \eqref{pk asymptotic Formula in g0} can be written as
\begin{equation}\label{somethingSilly}
    p_k(s) = 2i \left( \frac{k}{2} \right)^{1/3} \left(  \frac{-[\dotbold{P}_{-2}]_{12}+[\dotbold{P}_{-1}]_{22}[\dotbold{P}_{-1}]_{12}+ \Oh(1/k)}{ [\dotbold{P}_{-1}]_{12} + \Oh(1/k)}  \right).
\end{equation}
A direct computation shows that
\begin{equation}\label{PdotToQdot}
    \frac{-[\dotbold{P}_{-2}]_{12}+[\dotbold{P}_{-1}]_{22}[\dotbold{P}_{-1}]_{12}+ \Oh(1/k)}{ [\dotbold{P}_{-1}]_{12} + \Oh(1/k)} =   \frac{-[\dotbold{Q}_{-2}]_{12}+[\dotbold{Q}_{-1}]_{22}[\dotbold{Q}_{-1}]_{12}+ \Oh(1/k)}{ [\dotbold{Q}_{-1}]_{12} + \Oh(1/k)},
\end{equation}
where the quantities in the right-hand side are defined in \eqref{QdotNormalization}.
We now recover the large-$z$ expansion coefficients that are present in the right-hand side of \eqref{PdotToQdot}.
First, using \eqref{LargeZAbelMap} and the fact $\Theta(z)$ is entire we have for any constant $v$, as $z \to \infty$,
\begin{equation}\label{dumbtemp4}
    \Theta(\A(z)+v) = \Theta(\A(\infty)+v) + \Theta'(\A(\infty)+v) \frac{\A_{-1}}{z} + \Oh\left(1/z^2\right).
\end{equation}
Recall
\begin{equation}\label{dumbtemp3}
    [\dotbold{Q}]_{22} e^{-kF_1z} = C_{22} \frac{\Theta(\A(z)-C_{1}+kF_1U)}{\Theta(\A(z)-C_{1})} e^{kF_1 \int_A^z \Upsilon - z} f^{\text{D}}(z),
\end{equation}
where $C_1 = - \A(Q)-\K$. Using \eqref{defOfC22}, \eqref{dumbtemp4}, and \eqref{LargeZupsilon} we can the above equation as
\begin{align}
    \begin{split}
        [\dotbold{Q}]_{22} e^{-kF_1z} &= \left(\frac{\Theta(\A(\infty)-C_{1})}{\Theta(\A(\infty)-C_{1}+kF_1U)} \frac{\Theta(\A(z)-C_{1}+kF_1U)}{\Theta(\A(z)-C_{1})}\right) \left(e^{kF_1 \left(\int_A^z\Upsilon - z+ \Upsilon_{0}\right)}\right) f^{\text{D}}(z) \\[5pt]
        &= \left(1+\frac{\LD_{22} \A_{-1}}{z} + \Oh(1/z^2)\right) \left(1 - \frac{k F_{1}\Upsilon_{-1}}{z} + \Oh(1/z^2)\right) (1+ \Oh(1/z^2)) \\[5pt]
        &= 1 + \frac{\LD_{22} \A_{-1}- kF_{1}\Upsilon_{-1}}{z} +  \Oh(1/z^2),
    \end{split}
\end{align}
where $\LD_{22}(x;k) \equiv \LD_{22}$ is defined in \eqref{L22}. Therefore, we find that 
\begin{equation}\label{dotQ 22-entry -1}
    [\dotbold{Q}_{-1}]_{22} = \A_{-1}\LD_{22} - kF_1\Upsilon_{-1}.
\end{equation}
Next, we find the quotient $[\dotbold{Q}_{-2}]_{12}/[\dotbold{Q}_{-1}]_{12}$. We write
\begin{align}
    \begin{split}
        [\dotbold{Q}]_{12} e^{-kF_1z} &=  C_{11} \frac{\Theta(\A(z)+C_{1}+kF_1U)}{\Theta(\A(z)+C_{1})} f^{\text{OD}}(z)  e^{kF_1 \left(\int_A^z\Upsilon - z\right)}\\[5pt]
        &=  \mu e^{-kF_1\Upsilon_{0}}  \left(1+\frac{\LD_{12}\A_{-1}-kF_1\Upsilon_{-1}}{z}+\Oh(1/z^2)\right)\\ 
        &\hspace{.33cm} \cdot \left(\frac{-i(B+D-A-C)}{4z} - \frac{i(B^2+D^2-A^2-C^2)}{8z^2} + \Oh(1/z^3)\right),
    \end{split}
\end{align}
where
\begin{equation}
    \mu(x;k) \equiv \mu := C_{11} \frac{\Theta(\A(\infty)+C_{1}+kF_1U)}{\Theta(\A(\infty)+C_{1})}
\end{equation}
and $\LD_{12}(x;k) \equiv \LD_{12}$ is defined in \eqref{L12}. Thus, we have that
\begin{equation}\label{dotQ 12-entry -1}
    [\dotbold{Q}_{-1}]_{12} =  \mu e^{-kF_1\Upsilon_{0}} \left( -i\frac{B+D-A-C}{4}  \right)
\end{equation}
and
\begin{align} \label{dotQ 12-entry -2}
    \begin{split}
        &[\dotbold{Q}_{-2}]_{12} =\\
        & \mu e^{-kF_1\Upsilon_{0}} \left(-\frac{i(B+D-A-C)(\LD_{12}\A_{-1}-kF_1\Upsilon_{-1})}{4}- \frac{i(B^2+D^2-A^2-C^2)}{8} \right).
    \end{split}
\end{align}
Therefore,
\begin{equation}\label{chunk2}
    \frac{[\dotbold{Q}_{-2}]_{12}}{[\dotbold{Q}_{-1}]_{12}} = \A_{-1}\LD_{12}-kF_1\Upsilon_{-1} + \frac{B^2+D^2-A^2-C^2}{2(B+D-A-C)}.
\end{equation}
We are now ready to prove Theorem \ref{genus1 theorem}.
\begin{proof}[Proof of Theorem \ref{genus1 theorem}]

    Notice \eqref{somethingSilly} and \eqref{PdotToQdot} give
    \begin{equation}
        p_k(s) = 2i \left( \frac{k}{2} \right)^{1/3} \left(  \frac{-[\dotbold{Q}_{-2}]_{12}+[\dotbold{Q}_{-1}]_{22}[\dotbold{Q}_{-1}]_{12}+ \Oh(1/k)}{ [\dotbold{Q}_{-1}]_{12} + \Oh(1/k)} \right).
    \end{equation}
    Writing $p_k$ in terms of $u^{(k+1/2)}_{\HM}$ using \eqref{AnnoyingScalingFormula} and \eqref{Scalings} results in
    \begin{equation}
        -(2k)^{-1/3}u^{(k+1/2)}_{\HM}\left( - \frac{k^{2/3}}{2^{1/3}}x\right) = i\left(  \frac{-[\dotbold{Q}_{-2}]_{12}+[\dotbold{Q}_{-1}]_{22}[\dotbold{Q}_{-1}]_{12}+ \Oh(1/k)}{ [\dotbold{Q}_{-1}]_{12} + \Oh(1/k)} \right).
    \end{equation}
    This is the first statement of the theorem. Now, assume that $[\dotbold{Q}_{-1}]_{12} \ne 0$. In this case, we can geometrically expand the denominator of the above equation to find
    \begin{equation}
        -(2k)^{-1/3}u^{(k+1/2)}_{\HM}\left( - \frac{k^{2/3}}{2^{1/3}}x\right) = -i\left( [\dotbold{Q}_{-1}]_{22}-\frac{[\dotbold{Q}_{-2}]_{12}}{[\dotbold{Q}_{-1}]_{12}} + \Oh\left( 1/k \right)\right).
    \end{equation}
    Finally, \eqref{dotQ 22-entry -1} and \eqref{chunk2} give our desired result.
\end{proof}

\section{Numerical Approximation}\label{Numerics}

We use numerical methods to obtain approximations of the exact generalized Hastings-McLeod functions. However, since $u_{\HM}^{(\alpha)}$ has poles in the pole region of the plane, we will employ two distinct solvers: a spectral  collocation boundary condition solver to handle the smooth, pole-free region, and a pole-vaulting method for the pole region developed by Fornberg and Weideman \cite{Fornberg2011ANM}. This numerical scheme has been used to survey the pole fields of various Painlev\'e transcendents; see \cite{fasondini2018computational,fornberg2014computational,reeger2014painleve}. In this section we will give a brief description of each method.

\subsection{Spectral Collocation Method}

For $N \in \N$, we consider the interpolation nodes 
\begin{equation}
    t_k = \cos\left(\frac{(k-1)\pi}{N-1}\right),\hspace{1cm} k=1, \dots, N,
\end{equation}
which are  Chebyshev points of the second kind, i.e.\ the extreme points on $[-1,1]$ of $T_{N-1}(t)$, the Chebyshev polynomial of degree $N-1$. For a smooth function $f$ on $[-1,1]$ we set the interpolant $P_{N-1}(t)$ to be a linear combination of Lagrangian interpolating polynomials. So,
\begin{equation}\label{interpolantCol}
    f(t) \approx P_{N-1}(t) = \sum\limits_{j=1}^N \varphi_j(t) f_j,
\end{equation}
where  
\begin{equation}
    f_j:=f(t_j) \hspace{.66 cm} \text{and} \hspace{.66 cm} \varphi_j (t) :=\prod\limits_{k=1, k \ne j}^N \frac{(t-t_k)}{(t_j-t_k)}.
\end{equation}
In fact, since $t_2, \dots, t_{N-1}$ are the critical points of $T_{N-1}(t)$, 
\begin{equation}
    \varphi_j(t) = \frac{(-1)^j}{c_j} \frac{1-t^2}{(N-1)^2} \frac{T_{N-1}'(t)}{t-t_j},
\end{equation} where $c_1=c_N=2$ and $c_2,\dots,c_{N-1}=1$.
Notice, on the set of interpolation nodes  $\{t_1, \dots, t_N\}$, $\varphi_j(t_k)$ acts as a Kronecker delta function of $j$ and $k$. That is,
\begin{equation}
    \varphi_{j}\left(t_k\right) = \delta_{jk} = \begin{cases}
        0 & \text{if } j \ne k, \\
        1, & \text{otherwise}.
    \end{cases}
\end{equation}
We now define a collocation derivative operator, $\boldsymbol{D}$, on the interpolant $P_{N-1}(t)$ by 
\begin{equation}
    \boldsymbol{D}p_{N-1}(t) := \left[P'_{N-1} \left(t_1\right), \dots, P'_{N-1}\left(t_N\right)\right]^T.
\end{equation}
It was shown in \cite[p. 69]{diffMat} that
\begin{equation}
    [\boldsymbol{D}]_{k,j} = \begin{cases}\displaystyle
        \frac{c_k(-1)^{j+k}}{c_j (t_k-t_j)}, &  j \ne k, \\[10pt]
        \displaystyle \frac{2(N-1)^2+1}{6}, & j=k=1, \\[10pt]
        \displaystyle -\frac{2(N-1)^2+1}{6}, & j=k=N, \\
        \displaystyle \frac{-t_k}{2(1-t_k^2)}, & \text{otherwise.}
    \end{cases}.
\end{equation}
Furthermore, in our context, taking $k$ derivatives of $P_{N-1}(t)$ and evaluating at the nodes $\{t_1, \dots, t_N \}$ can be accomplished by the operator $\boldsymbol{D}^{(k)}$ where
\begin{equation}\label{discreteDerivative}
    \boldsymbol{D}^{(k)} := \left(\boldsymbol{D}\right)^{k}.
\end{equation}

\subsubsection{Collocation Method for Generalized Hastings-McLeod} 
We now turn our attention to applying this method to the generalized Hastings-McLeod function for fixed $\alpha > -1/2$. First, we fix complex numbers $y_1$ and $y_2$ in the $y$-plane with $\Imag\left(y_1\right) = \Imag\left(y_2\right)$, $\Real\left(y_1\right) \ll 0 \ll  \Real\left(y_2\right)$, so that the asymptotics \eqref{BonCon} are good approximations for both $u(y_1)$ and $u(y_2)$. We also require that, after the appropriate scaling, the horizontal line segment connecting $x_1$ to $x_2$ (which we will denote as $[x_1,x_2]$) is contained in the pole-free region of the $x$-plane. Next, we transform  the ODE \eqref{PII} from $[y_1,y_2]$, to $[-1,1]$. In particular, we use the linear transformation
\begin{equation}
    f(t) := \frac{y_1+y_2 +\left(y_2-y_1\right)t}{2}.
\end{equation}
Then, if $u$ satisfies \eqref{PII}, then  $v:[-1,1] \to \Ci$ defined by $v(t):= u(f(t))$ satisfies
\begin{equation}\label{transformedODE}
    \frac{\dd^2v}{\dd t^2}(t) = \frac{\left(y_2-y_1\right)^2}{4} \left(2v(t)^3 + f(t)v(t)-\alpha\right).
\end{equation}
We interpolate $v(t)$ as in \eqref{interpolantCol}. That is, for $N$ unknown values $v_1, \dots, v_N$ we set
\begin{equation}\label{PNinterpolant}
     P_{N-1}(t):= \sum\limits_{j=1}^{N}\varphi_j(t)v_j.
\end{equation}
Our aim is for $v_1,\dots,v_N$ to approximate the values $v(t_1), \dots, v(t_N)$ respectively. Then, $v(t)\approx P_{N-1}(t)$. Now, we apply the discrete derivative operator \eqref{discreteDerivative} on the left-hand side of \eqref{PNinterpolant} and plug in $P_{N-1}$ for $v$ in the right-hand side of \eqref{transformedODE}. As a result, we obtain an $N \times N$ system of equations. We can enforce the boundary conditions of the generalized Hastings-McLeod solutions \eqref{BonCon} by setting
\begin{equation}
    v_N := f\left(\sqrt{-a/2}\right),
\end{equation}
since $x_1$ was chosen so that \eqref{BonCon} holds. Hence, $ f\left(\sqrt{-a/2}\right) \approx f\left(u^{(\alpha)}_{\HM}\left(x_1\right)\right) \approx v_N$. Similarly, we set $v_1 := f\left(\alpha/2\right)$. Finally, we use a Newton solver on the remaining $N-2$ equations and variables.

\subsection{Pad\'e Approximant Method}

The collocation method is a boundary condition solver and hence assumes the function is smooth between the two boundary conditions. However, this is not always the case for the generalized Hastings-McLeod functions. Indeed, we know they blow up in the pole region of the $x$-plane.  Therefore, we use a Pad\'e approximant method due to Fornberg and Weideman, which is an initial-value problem solver, to handle this case. We now give a brief overview of Pad\'e approximants and apply this method to the generalized Hastings-McLeod functions.

\subsubsection{Pad\'e Expansion} Let $U$ be an open set in $\Ci$. Say $f$ is a meromorphic function in $U$ and, for some fixed $x_0 \in U$, $f$ is analytic at $x_0$. As such, $f$ has the following Taylor expansion about $x_0$:
\begin{equation}\label{taylorSeriesNum}
    f(x_0+h) = \sum\limits_{k=0}^nc_k h^k + \Oh\left(h^{n+1}\right).
\end{equation}
However, this approximation fails in a neighborhood of a pole. To obtain a more robust approximation that allows for some poles, we convert the Taylor expansion into a Pad\'e expansion. That is, instead of approximating $f$ as a polynomial, we approximate $f$ as a nontrivial rational function. So, we need to  find complex coefficients $a_0, \dots, a_{\nu_1}, b_1, \dots b_{\nu_2}$ such that
\begin{equation}\label{padeExpand}
    f(x_0+h)= \frac{a_0+a_1h+\cdots+a_{\nu_1}h^{\nu_1}}{1+b_{1}h+\cdots+b_{\nu_2}h^{\nu_2}} + \Oh\left(h^{n+1}\right),
\end{equation}
where $\nu_1+\nu_2=n$. In our problem we take $n$ to be even and set $\nu=\nu_1=\nu_2:=n/2$. The neighborhood in which the Pad\'e approximation \eqref{padeExpand} is decent is larger than the corresponding neighborhood of the Taylor approximation \eqref{taylorSeriesNum} because if $x_{*}$ is a pole, then the denominator of  \eqref{padeExpand} can vanish near $x_*$.  As the denominator has $\nu_2$ zeros, we see that the closure of the domain of \eqref{padeExpand} can contain $\nu_2$ poles before we lose control of the error term. 
\subsubsection{Pad\'e Approximant Method for Generalized Hastings-McLeod.} Fix $\alpha > -1/2$ and let $u(y)$ be the generalized Hastings-McLeod solution to \eqref{PII}. In order to obtain an approximation of $u$ in the form of \eqref{padeExpand} we will need a pair of initial values $u(y_0)$ and $u'(y_0)$. First, we use the forward Euler numerical scheme to obtain a partial Taylor series approximation of $u$ about $y_0$ for some even integer $n$ (in our problem we set $n=24$). That is, we set $c_0=u(y_0)$ and $c_1= u'\left(y_0\right)$ and write
\begin{equation}\label{numStepOne}
    u(y_0+h) = \sum\limits_{k=0}^{n} c_k h^k + \Oh\left(h^{n+1}\right)
\end{equation}
for $n-1$ unknown complex values $c_2, \dots, c_n$. To solve for  $c_2, \dots, c_n$  we plug  \eqref{numStepOne} into the ODE \eqref{PII} and find that
\begin{equation*}
\sum\limits_{k=2}^n k(k-1)c_{k}h^{k-2} + \Oh\left(h^{n-1}\right) = 2 \left(\sum_{k=0}^n c_k h^k\right)^3 + (y_0 +h)\left(\sum\limits_{k=0}^{n}c_k h^k\right) - \alpha + \Oh\left(h^{n-1}\right).
\end{equation*}
So, we can compare the coefficients  of $1, h, \dots, h^{n-2}$ and obtain an $(n-1) \times (n-1)$ system of equations (that can be solved recursively) for the variables $c_2,\dots, c_{n}$.
To convert \eqref{numStepOne} into a Pad\'e approximant, with $n+1$ unknown complex values $a_0,\dots, a_{\nu}, b_1, \dots b_{\nu}$ we write
\begin{equation}\label{padeTaylor}
    \frac{a_0+\cdots+ a_{\nu}h^{\nu}}{1+ b_1 h+ \cdots + b_{\nu} h^{\nu}} = \sum\limits_{k=0}^n c_kh^k + \Oh\left(h^{n+1}\right).
\end{equation}
Now, by multiplying both sides of \eqref{padeTaylor} by $1+b_1h+\cdots+b_{\nu}h^\nu$, we can compare the coefficients of the $n+1$ monomials $1,h,\dots, h^{n}$ and obtain an $(n+1) \times (n+1)$ system of equations that once again can be solved recursively. Hence, we find that
\begin{equation}\label{numStepTwo}
    u(y_0+h) = \frac{a_0+\cdots+a_{\nu}h^\nu}{1+ b_1h+\cdots + b_{\nu}h^{\nu}} + \Oh\left(h^{n+1}\right)
\end{equation}
in a neighborhood of $y_0$ for which $\nu$ poles have been excised.

Notice, in order to obtain an approximation resilient to $\nu$ poles, we just needed the initial values $u(y_0)$ and $u'(y_0)$. Further, \eqref{numStepTwo} can generate decent approximations of $u$ (and hence $u'$) at points that are not $y_0$. Thus, by choosing a point $y_1$ that is close enough to $y_0$ such that \eqref{numStepTwo} is still a good approximation of $u(y_1)$, we are able to extend the region of the $y$-plane where we can approximate $u$. As such, given $U$, a ``pole filled'' region of the complex $y$-plane and $y_0 \in U$ such that the values $u(y_0)$ and $u'(y_0)$ are known, we create a coarse grid of node points filling out $U$, with the aim of finding a Pad\'e approximant near each node point. Fornberg and Weideman first implemented such an algorithm in \cite{Fornberg2011ANM}. 

Their idea is to start at $y_0$ and to generate integration paths that branch out so that each node point is near a path. Further, to preserve accuracy by minimizing numerical cancellations, we want to avoid integrating near a pole. Therefore, we want to choose paths along which the solution remains small in magnitude. This is accomplished by running the following routine.
\begin{enumerate}
    \item 
        Uniformly sample a target node from the aforementioned grid of node points.
    \item 
        Find the nearest point with a recorded Pad\'e approximation. (Since we start with initial values we know such a point will exist.)
    \item 
        Take a step of length $h$ (in our case $h=1/2$) towards the target. We consider five candidate directions: one aiming directly at the target and the remaining ones veering off by $22.5^\circ$ and $45^\circ$ to either side. Evaluate the known Pad\'e approximant at each of the five candidate points, and pick the direction whose candidate had the smallest norm evaluation.
    \item 
        Apply the forward Euler scheme and Pad\'e approximation method with initial data $u$ and $u'$ corresponding to the chosen candidate. 
    \item    
        Add the candidate point to the list of points with Pad\'e approximants, record the Pad\'e coefficients, and remove all node points within distance $1/2$ of the Pad\'e approximant.
    \item 
        Repeat steps $(3)$ through $(5)$ until we are within distance $h$ away from our target.
\end{enumerate}
Finally, we repeat the above routine until the target points are empty. Now given $y \in U$, we approximate $u(y)$ by finding $w$, the nearest point to $y$ with a Pad\'e approximant, and evaluating this Pad\'e approximant at $w-y$. See Figure \ref{integrationPath} for an example when $\alpha=3/2$. Here, $y_0 = -\alpha^{2/3}\cdot 2^{-1/3} \cdot 3 \cdot e^{-\frac{2\pi}{3}i}$, and the values $u\left( y_0 \right)$ and $u'\left( y_0 \right)$ were found using the collocation method.

\begin{figure}[h]
    \centering
   \scalebox{.65}{\includegraphics{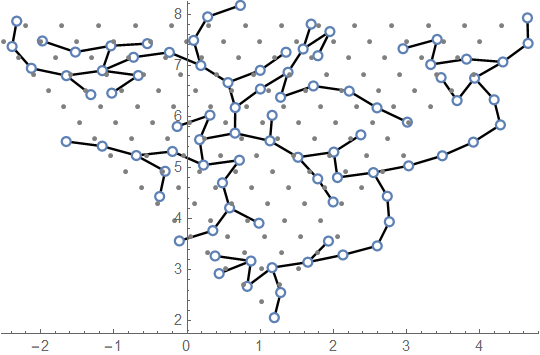}}
   \caption{Illustration of integration paths in the complex $y$-plane with $\alpha=3/2$. The open circles are the points with Pad\'e approximants about them. Notice, each gray dot (a node point) is within distance $h=1/2$ of an open circle.}\label{integrationPath}
\end{figure}

\appendix

\section{The Airy Parametrix}\label{AiryAppendix}
Here we give the Riemann-Hilbert problem corresponding to the standard Airy parametrix $\boldsymbol{A}(\zeta)$.
\begin{RHP} \label{AiryRHP}
    Find the $2 \times 2$ matrix-valued function $\boldsymbol{A}(\zeta)$ that satisfies the following properties:
    \begin{enumerate}
        \item \textbf{Analyticity:} $\boldsymbol{A}(\zeta)$ is analytic off the rays $\arg(\pm\zeta)=0$ and $\arg(\zeta)=\pm 2 \pi/3$.
        \item \textbf{Jump Condition:} The boundary values taken by $\boldsymbol{A}(\zeta)$ are related by the jump condition $\boldsymbol{A}_{+}(\zeta)=\boldsymbol{A}_{-}(\zeta)\boldsymbol{V}^{(\boldsymbol{A})}(\zeta)$, where the jump matrix is as shown in Figure \ref{AiryParametrixJumps}.
        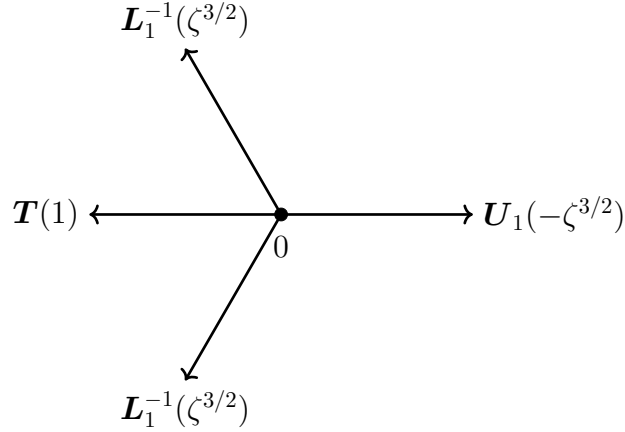
\begin{figure}[h]
            \centering
            \scalebox{.85}{
                \begin{tikzpicture}[]
                    
                    \draw[fill=black] (0,0) circle (.1cm) node[below=.2cm]{\large$0$};
                    \draw[->, very thick] (0,0) -- ++(0: 3cm) node[right=0cm] {\large$\boldsymbol{U}_{1}(-\zeta^{3/2})$}; 
                    \draw[->, very thick] (0,0) -- ++(120: 3cm) node[above=0cm] {\large$\boldsymbol{L}_{1}^{-1}(\zeta^{3/2})$};
                    \draw[->, very thick] (0,0) -- ++(180: 3cm) node[left=0cm] {\large$\boldsymbol{T}(1)$}; 
                    \draw[->, very thick] (0,0) --++ (-120:3cm) node[below=0cm] {\large$\boldsymbol{L}_{1}^{-1}(\zeta^{3/2})$};

                \end{tikzpicture} }
            \caption{The jumps of the standard Airy parametrix $\boldsymbol{A}(\tau)$.}
            \label{AiryParametrixJumps}
        \end{figure}
        \item \textbf{Normalization:} As $\zeta \to \infty$,
        \begin{equation}\label{AiryNormalization}
            \boldsymbol{A}(\zeta)\boldsymbol{V}^{-1} \zeta^{-\sigma/4} = \I +
            \begin{bmatrix}
                \Oh\left( \zeta^{-3} \right) & \Oh\left( \zeta^{-1} \right) \\
                \Oh\left( \zeta^{-2} \right) & \Oh\left( \zeta^{-3} \right)
            \end{bmatrix},
        \end{equation}
        where $\boldsymbol{V}$ is the unimodular and unitary matrix
        \begin{equation}\label{VMatDef}
            \boldsymbol{V} := \frac{1}{\sqrt{2}}
            \begin{bmatrix}
                1 & -i \\
                -i & 1
            \end{bmatrix}.
        \end{equation}
    \end{enumerate}
\end{RHP}
The Airy parametrix is solved in terms of the Airy function $\Ai(\cdot)$. For the explicit formula see Appendix A of \cite{RationalPainleveII}.

\section{Behavior of Re$(2h+\lambda)$ when $x$ is real.} \label{proofApp}

In this appendix we provide proofs of statements \eqref{symOfhg0} and \eqref{Re(2h+lambda) on the real line}. In both statements $x$ is a fixed, real number (and hence, in the pole-free region). The first statement is when $x \in \R$, then
\begin{equation}
    \Real \left(2h\left(-\overline{z}; x \right) + \lambda(x) \right) =  \Real \left(2h(z;x)+\lambda(x)\right).
\end{equation}
We prove the equivalent statement that 
\begin{equation}
    \Real \left( h(-\overline{z}:x) \right) = \Real\left( h(z;x) \right).
\end{equation}
\begin{proof}
    Consider the function $T(x) := iS(x)$. Then, $T(x)$ is one of the solutions to the cubic equation
    \begin{equation}\label{Tcubic}
    T^3 -xT -2 =0.
    \end{equation}
    The discriminant of \eqref{Tcubic} is $4x^3-108$, which vanishes at $x=3$. Indeed, at $x=3$, \eqref{Tcubic} has a double root at $T=-1$ and a simple root at $T=2$. Since $S(x)$ is continuous at $x=3$, it follows that $T(3)=2$. Further,  $4x^3-108$ does not vanish along the real line (except at $x=3$), so when $x$ is real $T(x)$ is the largest real root of \eqref{Tcubic}. So, when $x<3$, $T(x)$ will be the only real root of \eqref{Tcubic}; on the other hand, when $x>3$, $T(x)$ will be the largest root among the three distinct real roots of \eqref{Tcubic}. Notice, for each $x \in \R$ the cubic polynomial $p(T):=T^3-xT-2$ is negative at $T=0$. Indeed, $P(0)=-2$. However, since $P(T) \to +\infty$ as $T \to +\infty$, $P(T)$ must have at least one positive root. Therefore, $T(x)$, being the largest real root of $p$,  must be positive, which implies that $S(x)$ is negative-imaginary for $x \in \R$.

    Since the square root in \eqref{Delta a and b} is principal, it follows that $\Delta(x)$ is on the positive real axis and 
    \begin{equation}\label{relaBetaAndbxReal}
        a(x) = -\overline{b(x)}.
    \end{equation}
    Further, \eqref{relaBetaAndbxReal} (along with basic trigonometric identities) implies when $x$ is real, the square-root function $r(z;x)$ (defined in \eqref{defOfr(z)}) exhibits the symmetry 
    \begin{equation}\label{symOfrWhenxReal}
        r\left(-\overline{z};x \right) = - \overline{r(z;x)}.
    \end{equation}

    Fix $x \in \R$ and consider the bivariate function of real variables $u$ and $v$ defined by
    \begin{equation}
        j(u,v;x) \equiv j(u,v):= h(u+iv) - h(-u+ iv) =  h(z) - h\left( -\overline{z} \right)
    \end{equation}
    (here $z=u+iv$ and $j(u,v)$ is well defined whenever $h(z)$ is).  Notice, \eqref{symOfhg0} holds if and only if the real part of $j(u,v)$ is identically zero. Moreover, since $j(0,0)=0$, it suffices to show that the real part of the partials of $j$ vanish. Using the Cauchy-Riemann equations one can show this is equivalent to  $\Real(h'(z)+h'(-\overline{z}))=0$ and $\Imag(h'(z)-h'(-\overline{z}))=0$.

    Looking at \eqref{defhg0} and \eqref{gprimeG0} we see that $h'(z)=i(S+2z)r(z)$. Find $t\in \R$ so that $S(x)=it$. We compute
    \begin{align}
        \begin{split}
            h'(z) \pm h'(-\overline{z}) & = i (it+ 2z) r(z) \pm i(it- 2\overline{z})r(-\overline{z}) \\
            &=i (it+ 2z) r(z) \pm  i(-it+2\overline{z})\overline{r(z)} \\
            &= i\big((it+2z)r(z) \pm \overline{(it+2z)r(z)} \big).
        \end{split}
    \end{align}
    Thus, we have that $h'(z)+h'(-\overline{z})$ is purely imaginary and $h'(z)-h'(-\overline{z})$ is purely real, our desired results.
\end{proof}
We now prove when $x \in \R$, $\Real(2h(z;x)+\lambda(x))$ can at most vanish two times on the real line (in the $z$-plane).
\begin{proof}
    Fix $x \in \R$, and consider the function $\mathcal{F}:\R \to \R$ defined by
   \begin{equation}
    \mathcal{F}(u) := \Real(2h(u;x)+\lambda(x)).
   \end{equation}
   Rolle's theorem tells us for any two zeros of $\mathcal{F}$, there exists a point between them such that $\Real(h')$ vanishes at that point. In particular, if $\mathcal{F}$ has three or more zeros, then $\Real(h')$ would vanish at least two times on $\R$. Thus, in order to show that $\mathcal{F}$ can have at most two zeros on $\R$ it suffices to show that $\Real(h')$ only vanishes once on $\R$, namely at $0$. Recall
   \begin{equation}
    h'(z) = i(S +2z)r(z).
   \end{equation}
   Using \eqref{relaBetaAndbxReal}  we see that $h'(0)=-S|b|$. Seeing as $S$ is purely imaginary, we have that $\Real(h'(0))=0$. Since $h'$ only vanishes at $z=a$, $z=b$, and $z=c$ (none of which are real numbers), it follows $|h'(u)|>0$ for $u \in \R$. This means we need only to control  $\arg(h'(u))$. Moreover, the symmetry \eqref{symOfrWhenxReal} implies that for $u \in \R$, we have
   \begin{equation}
    \arg(h'(-u)) = \pi - \arg(h(u)).
   \end{equation}
   Hence, it suffices to show that for $u>0$ $\arg(h'(u)) \in (\pi/2, \pi)$.

   For the upper bound we see that, since $S$ is negative-imaginary and $u$ is positive-real, it follows
   \begin{equation}
    \arg(h'(u)) = \arg(i) + \arg(S+2u) + \arg(r(u)) < \frac{\pi}{2} + \frac{\arg(-a)+\arg(\overline{a})}{2} = \pi. 
   \end{equation} 
   Next, a direct calculation shows that for, $u> 0$,
   \begin{equation}
    \frac{\dd}{\dd u} \arg(h'(u)) =  \frac{-\Imag(a) \Real(a)^2 \left( |a|^2-3u^2 \right)}{\left( \Imag(a)^2 + \left( \Real(a)-u \right)^2 \right)\left( \Imag(a)^2+ u^2 \right)\left( \Imag(a)^2 + \left( \Real(a)+u \right)^2 \right)}.
   \end{equation}
   So, $\arg(h'(u))$ is strictly decreasing on the interval $(|a|/\sqrt{3},\infty)$. Further, the normalization of $h$ implies that $h'(z)= i4z^2+ \Oh(1)$ as $z \to \infty$. Thus, as $u \to +\infty$, $\arg(h'(u)) \to \pi/2$, and for finite, positive $u$, $\arg(h'(u))>\pi/2$. This completes our proof. 
\end{proof}

\end{document}